\documentclass[12pt]{article}
\usepackage{latexsym,amsmath,amssymb,amsfonts,amsthm,bbm,mathrsfs,breakcites,dsfont,xcolor}
\usepackage{stmaryrd,epsf}
\usepackage{natbib}
\usepackage{url}
\usepackage{dsfont}
\usepackage{enumerate}
\usepackage{graphicx}
\usepackage{graphics}
\usepackage{psfrag}
\usepackage{caption}
\usepackage[colorlinks,linkcolor=red,citecolor=blue,urlcolor=blue,breaklinks = true]{hyperref}

\newtheorem{theorem}{Theorem}
\newtheorem{lemma}[theorem]{Lemma}
\newtheorem{proposition}[theorem]{Proposition}
\newtheorem{corollary}[theorem]{Corollary}
\newtheorem{assumption}{A\!\!}
\newtheorem{example}{Example}

\newcommand{\var}{\mathrm{Var}}
\newcommand{\cov}{\mathrm{Cov}}
\newcommand{\bbE}{\mathbb{E}}

\providecommand{\keywords}[1]
{
  \small	
  \textbf{Keywords:} #1
}

\DeclareMathOperator*{\argmin}{argmin}

\usepackage{xcolor}
\usepackage[draft,inline,nomargin,index]{fixme}
\fxsetup{theme=color,mode=multiuser}
\FXRegisterAuthor{yf}{ayf}{\color{blue} YF}
\FXRegisterAuthor{tc}{atc}{\color{red} TC}

\title{The correlation-assisted missing data estimator}
\date{}
\author{Timothy I. Cannings and Yingying Fan \\ \textit{University of Edinburgh} and \textit{University of Southern California}}

\begin{document}

\def\spacingset#1{\renewcommand{\baselinestretch}%
{#1}\small\normalsize} \spacingset{1}

\maketitle
\begin{abstract}
We introduce a novel approach to estimation problems in settings with missing data. Our proposal -- the \emph{Correlation-Assisted Missing data} (CAM) estimator -- works by exploiting the relationship between the observations with missing features and those without missing features in order to obtain improved prediction accuracy.  In particular, our theoretical results elucidate general conditions under which the proposed CAM estimator has lower mean squared error than the widely used complete-case approach in a range of estimation problems.  We showcase in detail how the CAM estimator can be applied to $U$-Statistics to obtain an unbiased, asymptotically Gaussian estimator that has lower variance than the complete-case $U$-Statistic.  Further, in nonparametric density estimation and regression problems, we construct our CAM estimator using kernel functions, and show it has lower asymptotic mean squared error than the corresponding complete-case kernel estimator.  We also include practical demonstrations throughout the paper using simulated data and the \emph{Terneuzen birth cohort} and \emph{Brandsma} datasets available from \texttt{CRAN}. 
\end{abstract} 
\noindent


\spacingset{1.1} 

\section{Introduction}
Data is a primary commodity in today's economy, it is valued and traded like any other asset.  Statistics and machine learning allow us to extract this value by improving operational efficiency, increasing revenue, or understanding the behaviour of customers.  A common complication in modern applications is that the data may be incomplete.  For example, some users may choose not to disclose their personal details (age, gender, geographic location, etc.) to a smartphone application; optional questions on an on-line form are often left blank; or data is sometimes removed or hidden to guarantee privacy.   In other situations, missing data problems can arise when two or more different data sources have been combined.

Missing data is not a new problem. As early as the 1950s, \citet{Anderson:1957} found the maximum likelihood estimator in a multivariate normal distribution when some of the observations were missing.  In a seminal paper, \citet{Rubin:1976} studied missing data in a rigorous general framework, introduced the notion of data \emph{missing at random}, and specified conditions under which the process that causes data to be missing may be ignored.  See also the comprehensive book on the subject by \citet{LittleRubin:2002} and the recent special issue of \textit{Statistical Science} \citep{JosseReiter:2018}.

A simple and widely used approach to deal with missing data is to discard any incomplete observations -- a technique referred to as complete-case analysis \citep[Chapter~3]{LittleRubin:2002}.  There is an obvious drawback with this method that perhaps much of the data is ignored.  An alternative approach is to impute the missing values \citep{Ford:1983}.  There is an extensive body of work on different imputation techniques; see, for instance, \citet[Chapters 4 and 5]{LittleRubin:2002} and \citet[Chapter IV]{HMDM:2015} for an overview.   Other techniques are based on the expectation--maximisation (EM) algorithm \citep{Dempster:1977}.  Missing data has also been studied in a range of high-dimensional settings, including regression \citep{LohWainwright:2012}, classification \citep{CaiZhang:2018}, and (sparse) principal component analysis \citep{ElsenerGeer:2018,ZhuWangSamworth:2019}.

In this paper, we develop a novel approach to missing data problems.  Our new proposal,  the correlation-assisted missing data (CAM) estimator, exploits the relationship between the complete cases and the observations with missing values, in order to improve on the performance of the complete-case estimator.  More precisely, we construct an (approximately) mean-zero statistic, using both the complete-cases and the data with missing entries, which is correlated with the complete-case estimator.  We then exploit this correlation to construct our new estimator, by making a linear adjustment to the complete-case estimator. 

Our first main result, Proposition~\ref{prop:MSE}, elucidates when the proposed CAM estimator will be more accurate than the complete-case estimator in terms of mean squared error.  The result does not require any assumptions on the data generating mechanism. In particular, we do not assume the data to be missing completely at random. Further, Proposition~\ref{prop:MSE} motivates an optimal (but typically unknown) choice of the adjustment term used in the construction of the CAM estimator.  This optimal choices leads to the greatest reduction in mean squared error.  We also show that we can then construct a data-driven version of the CAM estimator that performs well in many practical settings.

As a second main contribution, we showcase how the CAM estimation technique can be applied in specific settings.   First, when the complete-case estimator is a U-Statistic, we show that the optimal adjustment term also takes the form of a U-statistic. This motivates us to consider a CAM estimator with a general U-statistic as the adjustment term and allows for detailed theoretical analysis. In particular, when the data is missing completely at random, we show that our CAM estimator is unbiased, asymptotically Gaussian, and has a smaller asymptotic variance than the complete-case $U$-Statistic (cf.~Theorem~\ref{thm:Ustat2}).    We provide two concrete examples where a clear improvement can be shown analytically.  Moreover, the numerical properties of the CAM estimator are demonstrated using simulated and real data, including an application using the \emph{Terneuzen birth cohort} dataset available from \texttt{CRAN}.

We then investigate an application of the CAM technique to kernel based methods in nonparametric density estimation and regression problems.  Under standard nonparametric assumptions on the data generating distribution, we quantify the leading order asymptotic improvement in mean squared error obtained by the CAM estimator compared with the complete-case approach.  We provide further theoretical justification for the construction of the CAM estimator in these settings, in terms of approximating the optimal adjustment term.  This leads to a fast, fully data-driven construction of our estimator.   Finally, we demonstrate our method and the improvement it offers over the compete-case estimator in a  simulation study and show how it can be used in an application with the Brandsma dataset on \texttt{CRAN}.
 
Related methods to the CAM approach have been utilised in various double-sample design settings.  \citet{Fuller:1998} considers the problem of marginal mean estimation when one observes a small set of pairs of data and a larger set of univariate observations.  \citet{ChenChen:2000} proposed an estimator of the regression parameters in a generalised linear model, where the practitioner has one sample in which the observations may be noisy or proxy versions of the variables of interest, and a second validation sample where complete and exact observations of the features are available.  \citet{ChenChen:2000} show that their estimator of the regression parameter is asymptotically unbiased and has smaller (asymptotic) variance than a naive estimator based solely on the validation sample.  Similar ideas have been used more recently in different statistical problems.  \citet{Jiang:2011} propose a nonparametric kernel-based regression estimate in double sampling designs, where the response is missing in one of the samples, but a surrogate outcome is observed instead. \citet{Yang:2019} propose an estimator of the average causal treatment effect in a general setting, by combining multiple observational datasets; see also \citet{Lin:2014}, who focus on the logistic regression setting.  Very recently, \citet{Zhangetal:2019} considered estimating the marginal mean response in a semi-supervised setting, where one has a large number of \emph{unlabelled} observations alongside a small labelled training dataset.   Our work extends these existing works to a much wider range of estimation problems, including $U$-Statistics and nonparametric density estimation and regression.

Finally, we note here that, in contrast to many imputation approaches, the CAM technique has a clear and direct aim -- that is to reduce the mean squared error of the complete-case estimator.  On the other hand, while imputation is widely applicable, the properties of estimators derived from imputed data are typically not well understood.  Indeed,  imputation methods are often treated as a black-box solution to missing data problems, and the subsequent steps of estimation and inference are conducted independently of the imputation step, which may lead to unreliable results.  In our numerical work in Section~\ref{sec:sims}, we see that a regression estimator computed with imputed data can in fact perform worse than simply using a complete-case approach.  Moreover, in problems such as density estimation, it is unclear whether imputation offers a viable solution -- the distribution of the imputed data may not be the same as the target distribution.

The remainder of this paper is as follows.   In Section~\ref{sec:setting} we fix our general statistical and missing data settings and introduce the CAM estimator.  Section~\ref{sec:UStatistics} is dedicated to studying $U$-Statistics.  We then demonstrate how the CAM estimator can be applied using kernel methods in density estimation and regression problems in Section \ref{sec:RCD}.   We also provide, throughout the paper, a number of practical demonstrations of the method using real and simulated datasets.  We conclude our paper with a discussion of various practical considerations and possible extensions in Section~\ref{sec:discuss}.  All technical details and proofs of our theoretical results are presented in Section~\ref{sec:tech} in the Appendix.  We first end this section with an illustrative example that demonstrates how our estimator is constructed.  

\subsection{Illustrative example} 
\label{sec:toy}
Let $(X_1, Y_1)^T, \ldots, (X_{2n},Y_{2n})^T$ be independent and identically distributed $N_2(\nu, \Gamma)$ variables, where $\nu = (\nu_{X}, \nu_{Y})^T \in \mathbb{R}^2$ is unknown but the covariance $\Gamma = (\Gamma_{ij})$ is known.   Suppose we observe $\{(X_1, Y_1), \ldots, (X_n,Y_n)\}$ and $\{Y_{n+1}, \ldots, Y_{2n}\}$; in other words, in the language of missing data, the first component is missing completely at random in the second set of observations. Our task is to estimate $\nu_{X}$. 

\begin{figure}[ht!]
	\centering
	\includegraphics[width=0.48\textwidth]{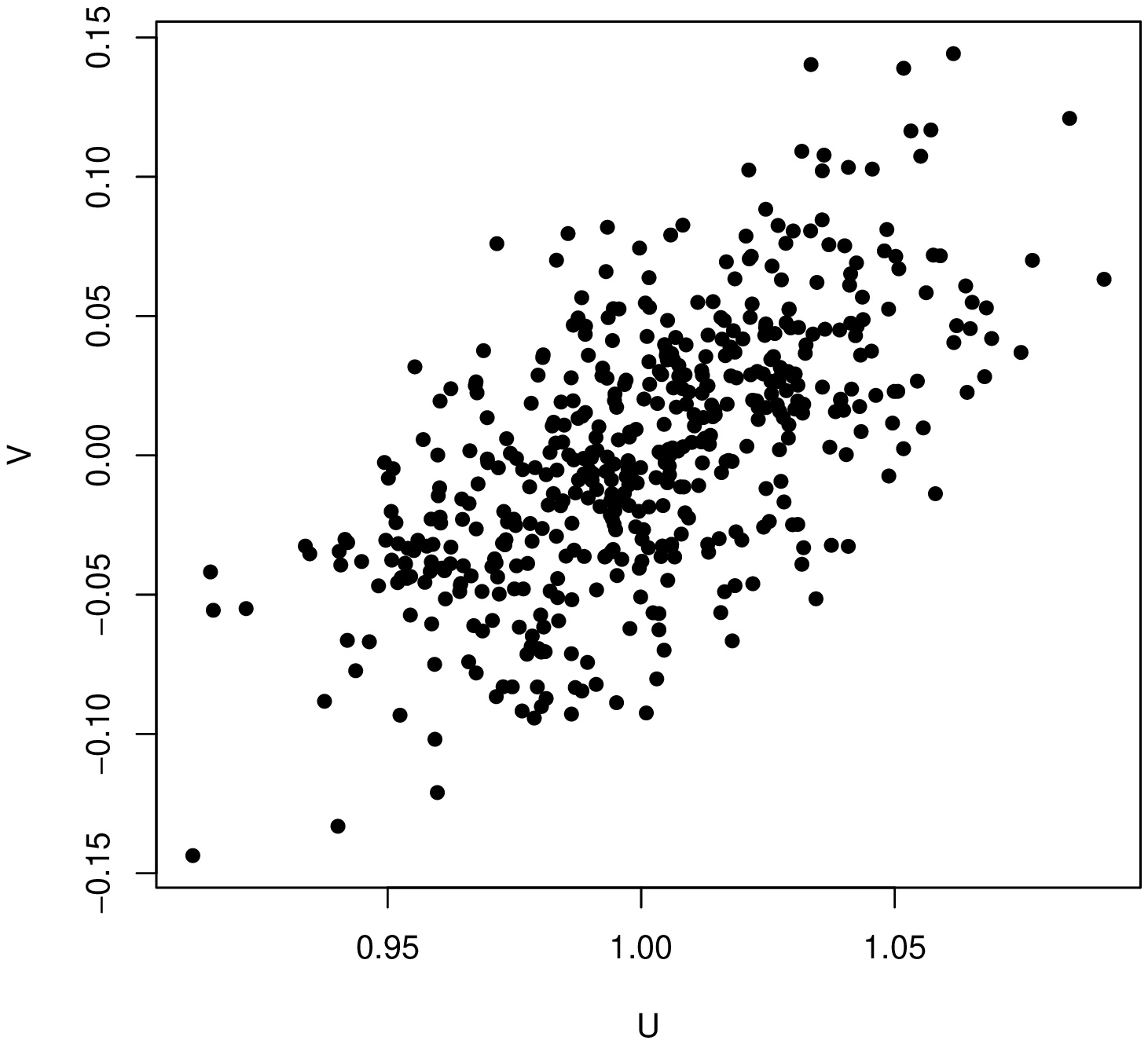} 	\includegraphics[width=0.48\textwidth]{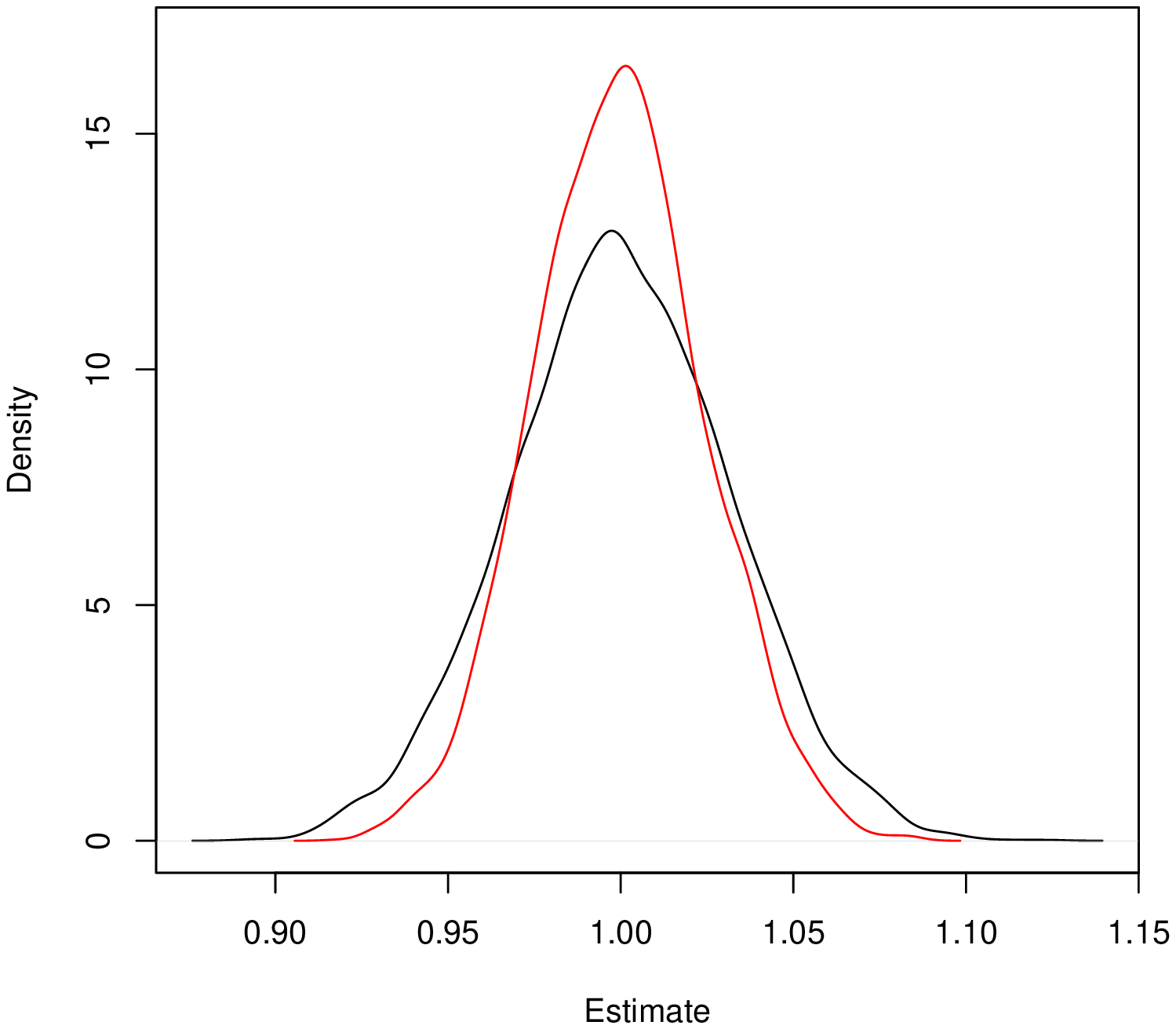} 
	\caption{Left: 500 observations from the joint distribution of $(U,V)$ (see~\eqref{eq:UV}), with $\nu_{X} = \nu_{Y} = 1$, $\Gamma = \frac{1}{10} I + \frac{9}{10} (1,1)^T(1,1)$ and $n = 1000$. Right: Sampling distributions of $\hat{\nu}_{X,1}$ in black and $\tilde{\nu}_X$ in red.}
	\label{fig:prelim}
\end{figure}

The complete-case estimator in this example ignores the second set of observations and takes the sample mean of the $X$ observations only. That is $\hat{\nu}_{X,1}:= \frac{1}{n} \sum_{i=1}^n X_i \sim N(\nu_X, \Gamma_{11}/n)$; this is unbiased and is the maximum likelihood estimator if there are no $Y$ observations.    Now consider $\hat{\nu}_{Y,1} := \frac{1}{n} \sum_{i=1}^n Y_i$ and $\hat{\nu}_{Y,2} := \frac{1}{n} \sum_{i=n+1}^{2n} Y_i$.   Then, letting $U := \hat{\nu}_{X,1}$ and $V := \hat{\nu}_{Y,1} - \hat{\nu}_{Y,2}$, we have 
\begin{equation}
\label{eq:UV}
(U,V)^T \sim N_2((\nu_X,0)^T, \tilde{\Gamma}),
\end{equation}
where $\tilde{\Gamma} = n^{-1}\Gamma + n^{-1} \Gamma_{22}(0,1)^T(0,1)$; see Figure~\ref{fig:prelim}. This motivates the estimator
\begin{equation}
\label{eq:prelim}
\tilde{\nu}_X := \hat{\nu}_{X,1} - \frac{\tilde{\Gamma}_{12}}{\tilde{\Gamma}_{22}} (\hat{\nu}_{Y,1} - \hat{\nu}_{Y,2}) \sim N\Bigl(\nu_X, \frac{1}{n}\Bigl(\Gamma_{11} - \frac{\tilde{\Gamma}_{12}^2}{\tilde{\Gamma}_{22}}\Bigr) \Bigr).
\end{equation}
We see that $\tilde{\nu}_X$ is also unbiased for $\nu_X$, but has a strictly smaller variance than  $\hat{\nu}_{X,1}$ whenever $\Gamma_{12}\neq 0$.  

There is a neat connection between our estimator and the Rao--Blackwell Theorem in this example.  The statistic $T = (T_1, T_2)^T = (\hat{\nu}_{X,1} - \hat{\nu}_{Y,1}\Gamma_{12}/\Gamma_{22}, \hat{\nu}_{Y,1} + \hat{\nu}_{Y,2})^T$ is sufficient (the technical details are presented in Section~\ref{sec:techprelim}) for $\nu$, and the \emph{Rao--Blackwellised} version of $\hat{\nu}_{X,1}$ is $\mathbb{E}(\hat{\nu}_{X,1} | T ) = \tilde{\nu}_{X}$.  In fact, one can show that $\tilde{\nu}_{X}$ is the maximum likelihood estimator of $\nu_X$ in this setting (cf.~\citet{Anderson:1957}).  




\section{Missing data and the CAM estimator}
\label{sec:setting}
Suppose $Z = (X,Y)$ is a random pair taking values in $\mathbb{R}^{d} \times \mathbb{R}$ with joint distribution $P$. We are interested in estimating some function of this distribution, $\theta = \theta(P)  \in \mathbb{R}$, say.   Examples studied in detail in this paper include the mean of the first component of $X$; the covariance between $X$ and $Y$; the value of a regression function of $Y$ on $X$ at the point $x \in \mathbb{R}^d$; or the value of the density $f_X(x)$, if it exists, of $X$ at $x \in \mathbb{R}^d$. 

We study a setting where some of the features of $X$ are missing, but where the response $Y$ is always observed.  In order to model this, suppose that we have $(Z,M)$, where the marginal distribution of $Z$ is $P$, and $M$ is a \emph{missingness} indicator taking values in $\{0,1\}^d$.  More precisely, we only observe the features $j \in \{1,\ldots,d\}$, with $M^j = 0$, that is $Z^{M} := (X^{M}, Y)$, where for $x = (x^1, \ldots, x^d)^T \in \mathbb{R}^d$ and $m = (m^1, \ldots, m^d)^T \in \{0,1\}^d$, we write $x^{m} := (x^j : m^j = 0) \in \mathbb{R}^{d_m}$, where $d_m := \sum_{j = 1}^d \mathbbm{1}_{\{m^j = 0\}}$.  Define the missingness probabilities $p_m(x,y) := \mathbb{P}(M = m | X = x, Y = y)$.   Further write $P_m$ and $Q_m$ for the joint distribution of $(X^{m}, Y)$ and $(X^{m},Y) | \{M = m\}$, respectively.  Note that under certain missing data assumptions we have $P_m = Q_m$, but that this is not the case in general. 

We say the data is \textit{missing completely at random} (MCAR) if $M$ is independent of the pair $(X,Y)$. In this case we write $p_m := \mathbb{P}(M = m)$.  The data is \textit{missing at random} (MAR) if the missingness indicator only depends on the observed data.  Formally, this means that $M$ is conditionally independent of $Z$ given $Z^M$.  
Whereas, the data is said to be \textit{missing not at random} (MNAR) if $M$ depends on the unobserved value. See, for example, \citet[Chapter~1.3]{LittleRubin:2002} for further discussion of these three scenarios. Our most general results in this paper make no assumption on the missingness type.  For our technical analysis in the $U$-Statistics and nonparametric learning problems, we focus on the MCAR case. Further discussion of the more challenging non-MCAR scenarios is given in Section~\ref{sec:discuss}. 

Let $(Z_1, M_1), \ldots, (Z_n, M_n)$ be independent and identically distributed pairs with the same distribution as $(Z,M)$.  We observe $Z_1^{M_1}, \ldots, Z_n^{M_n}$. The popular complete-case approach uses only the observations with $M_i = (0, \ldots, 0) \in \mathbb{R}^d$.  Our goal in this paper is to construct an estimator which also uses the observations with missing values in order to improve on the performance of the complete-case estimator. 

It is convenient to consider the missingness indicators $M_1, \ldots, M_n$ as fixed and equal to $m_1, \ldots, m_n$, and from this point on all probability statements should be interpreted to be conditional on $(M_1, \ldots, M_n) = (m_1, \ldots, m_n)$.  It is also useful to introduce some further notation here.  For $m \in \{0,1\}^d$, let $A_m := \{i \in \{1, \ldots,n\} : m_i = m\}$ be the set of indices of the data missing $m$, let $n_m := |A_m|$.  In particular, $A_0$ is the set of indices of the complete cases, where here and throughout we use the shorthand $0$ in place of $0_d := (0, \ldots, 0) \in \mathbb{R}^d$.  We have that $A_{m_1}\cap A_{m_2} = \emptyset$ for $m_1\neq m_2$.   Finally, for a set $A \subseteq \{1, \ldots, n\}$ and $m \in \{0,1\}^d$, let $\mathcal{T}_{A,m} := \{Z^{m}_i : i \in A\}$.  Of course, $\mathcal{T}_{A,m}$ is not necessarily observed for every $A$ and $m$, but we do observe $\mathcal{T}_{A_m,m}$.    We assume that $A_0$ is non-empty, and moreover, for each $m \in \{0,1\}^d$ with $A_m$ non-empty, we have that $\lim_{n \rightarrow \infty} \frac{n_m}{n} = q_m \in (0,1)$, almost surely.  That is, either a set of features $m$ is never missing, or, for those that are missing, the relative sample sizes are more or less balanced.  If the data is MCAR, then we have that $q_m = p_m$, for $m \in \{0,1\}^d$.

We now define a generic version of our \emph{correlation-assisted missing data} (CAM) estimator.   The main idea underpinning our proposal is to mimic the approach in the toy example in Section~\ref{sec:toy} by combining appropriate, correlated estimators, which are constructed using different parts of the data.   We first assume that there is a suitable complete-case estimator of $\theta$ that only uses the data in $\mathcal{T}_{A_0, 0}$; this is denoted by $\hat{\theta}_0 = \hat{\theta}_{A_0,0}$.  Furthermore, suppose that, for each $m$ and each $A \subseteq \{1, \ldots, n\}$, we have access to a statistic denoted by $\hat{\theta}_{A, m}$, which only depends on the data in $\mathcal{T}_{A, m}$.  Notice that, for $m \neq 0$, $\hat{\theta}_{A, m}$  is not an estimator for $\theta$. A detailed discussion of the choice of $\hat{\theta}_{A,m}$ is given at the end of this section.

Consider $m \in \{0,1\}^d \setminus\{0_d\}$  such that $A_m$ is non-empty. The CAM estimator is constructed using $\hat{\theta}_0$, $\hat{\theta}_{A_m,m}$ and $\hat{\theta}_{A_0,m}$.  The hope is that the latter two statistics have similar expected values, and that $\hat{\theta}_{A_0,m}$ is (highly) correlated with the complete-case estimator $\hat{\theta}_0$.  We then exploit this correlation in the same way as we did in the illustrative example -- see~\eqref{eq:prelim}.  

In fact, we can construct similar statistics using many $m \in \{0,1\}^d\setminus\{0_d\}$ simultaneously.   Let $\mathcal{M} \subseteq \{0,1\}^d \setminus\{0_d\}$ denote the set of values of $m$ that we would like to use (a detailed discussion of how to choose $\mathcal{M}$ in practice is postponed until Section~\ref{sec:discuss}).  Consider the two column vectors of length $|\mathcal{M}|$, given by $\hat{\theta}_{0,\mathcal{M}} :=  \bigl(\{\hat{\theta}_{A_0,m} : m \in \mathcal M\}\bigr)^T$ and  $\hat{\theta}_{\mathcal{M}} :=  \bigl(\{\hat{\theta}_{A_m,m} : m \in \mathcal{M}\}\bigr)^T$.  Finally, for $\gamma \in \mathbb{R}^{|\mathcal {M}|}$, define the \emph{correlation-assisted missing data} (CAM) estimator 
\[
\hat{\theta}_\gamma^{\mathcal{M}} :=  \hat{\theta}_{0} - \gamma^T(\hat{\theta}_{0,\mathcal M} - \hat{\theta}_{\mathcal{M}}).
\]
In practice we will use a data-driven choice of $\gamma$, which aims to minimise the mean squared error (cf.~the discussion after the statement of Proposition~\ref{prop:MSE}). 

We now study properties of the CAM approach in the general estimation problem.  For an estimator $\hat{\theta}$, let $\mathrm{MSE}(\hat{\theta}) = \mathbb{E} \{ (\hat{\theta}- \theta)^2\}$ denote its mean squared error.  Let $b(\hat{\theta}) := \mathbb{E}(\hat\theta - \theta)$ be the bias of an estimator $\hat\theta$ and let $B_{\mathcal {M}} := \mathbb{E}(\hat{\theta}_{0,\mathcal {M}} - \hat{\theta}_{\mathcal {M}})$.   Let $\Omega$ denote the $|\mathcal {M}|$-dimensional vector of covariances $\mathrm{Cov}(\hat{\theta}_{0}, \hat{\theta}_{0,\mathcal{M}})$ and let $\Lambda$ be the $|\mathcal {M}| \times |\mathcal{M}|$ covariance matrix $\mathrm{Var}(\hat{\theta}_{0,\mathcal{M}} - \hat{\theta}_{\mathcal{M}})$. 
\begin{proposition}
\label{prop:MSE}
We have that
\begin{align*} 
\mathrm{MSE}(\hat{\theta}_{\gamma}^{\mathcal M}) - \mathrm{MSE}(\hat{\theta}_{0}) =  \gamma^T (\Lambda + B_{\mathcal M} B_{\mathcal M}^T)\gamma -2\gamma^T \{\Omega +  b(\hat{\theta}_{0}) B_{\mathcal M}\}.
\end{align*}
In particular, if $\Lambda$ is nonsingular and $B_{\mathcal M}= 0$,  then $\gamma = \gamma^* :=  \Lambda^{-1}\Omega$ is the optimal weight vector achieving the maximum reduction in MSE, and 
\[
\mathrm{MSE}(\hat{\theta}_{\gamma^*}^{\mathcal M}) -  \mathrm{MSE}(\hat{\theta}_{0}) = - \Omega^T \Lambda^{-1} \Omega \leq 0.
\]
\end{proposition} 
Proposition \ref{prop:MSE} compares the MSE of our CAM estimator and the complete-case estimator. We have made no assumption on the missing data mechanism and, in particular, we do not assume here that the data is missing completely at random.  It is worth noting, however, that if the data is not MCAR, the complete-case estimator may be (even asymptotically) biased (cf.~Section~\ref{sec:discuss}).  In which case, simply improving on the performance of the complete-case estimator will not necessarily be effective. Furthermore, for general missing data mechanisms, we do not have control of $B_{\mathcal{M}}$.  However, we will see that under appropriate conditions in many estimation problems, $\hat{\theta}_{0,\mathcal{M}} - \hat{\theta}_{\mathcal{M}}$ will be (asymptotically) mean zero.
 
We see from the second part of Proposition \ref{prop:MSE} that to achieve maximum mean squared error reduction when $B_{\mathcal{M}} = 0$, we should set $\gamma = \gamma^* :=  \Lambda^{-1}\Omega$.  If, moreover, $\mathcal{M} = \{m\}$, then we have that
\[
\mathrm{MSE}(\hat{\theta}_{\gamma^*}^{\mathcal{M}}) -  \mathrm{MSE}(\hat{\theta}_{0}) =-\var(\hat\theta_0)\frac{ \var(\hat\theta_{0,m})}{\var(\hat\theta_{0,m})+\var(\hat\theta_{m})}\text{Corr}^2(\hat \theta_0, \hat\theta_{0,m}).
\]
Here we have used that $\hat\theta_{0,m}$ and $\hat\theta_{m}$ are independent since there are constructed using disjoint sets of observations.  Thus, to achieve a maximal reduction in MSE, we'd like $\hat\theta_{0,m}$ to be maximally correlated with $\hat\theta_{0}$ and $\var(\hat{\theta}_m)$ to be minimised.  The first is achieved by the conditional expectation $\hat\theta_{0,m}^* = \mathbb{E}(\hat\theta_{0}| \mathcal{T}_{A_0,m})$.   Moreover,  $\var(\hat\theta_m)$ is minimised by $\hat{\theta}_m^* = \mathbb{E}(\hat\theta_{0,m})$, but this is typically unknown.  In practice, we use the data $\mathcal{T}_{A_m, m}$ to construct an estimate of $\mathbb{E}(\hat\theta_{0,m})$ that has low variance.  The situation when $|\mathcal M|>1$ is similar by noting the independence of $\mathcal T_{A_{m_1},m_1}$ and $\mathcal T_{A_{m_2},m_2}$, for $m_1 \neq m_2$.

To understand this further, note that with the optimal choice $\hat{\theta}_{0,m}^*$ we have  
\[
\mathrm{Cov}(\hat{\theta}_{0} , \hat{\theta}^*_{0, m}) = \mathbb{E}\{\mathrm{Cov}(\hat{\theta}_{0} ,\hat{\theta}^*_{0, m} |  \mathcal{T}_{A_0,m}) \} + \mathrm{Cov}\{\mathbb{E}(\hat{\theta}_{0} | \mathcal{T}_{A_0,m}) ,\mathbb{E}(\hat{\theta}^*_{0, m} |  \mathcal{T}_{A_0,m})\}  = \var(\hat{\theta}^*_{0,m}).
\] 
Thus, in the ideal case that $\var(\hat{\theta}_{m})$ is negligible compared with $\var(\hat\theta^*_{0,m})$ (e.g.~if $n_m \gg n_0$), the improvement in MSE is simply $\var(\hat{\theta}^*_{0,m}) = \var\{\mathbb{E}(\hat{\theta}_{0}|\mathcal{T}_{A_0,m})\}$.

Of course, the conditional expectation $\mathbb{E}(\hat\theta_{0}| \mathcal{T}_{A_0,m})$ is also typically unknown.  We will see in practice that, for instance, assuming a parametric form for $\hat{\theta}_{0,m}^*$ works well.    In particular, in our study of $U$-Statistics in Section~\ref{sec:UStatistics}, we see that a data-driven choice of $\hat{\theta}_{0,m}$ will often lead to similar performance to the optimal choice.  Moreover, for nonparametric methods using kernels, the optimal $\hat{\theta}_{0,m}^*$ can	often be well approximated by the same type of nonparametric estimator with a practical choice of kernel (see Section~\ref{sec:RCD}).

\section{$U$-Statistics}
\label{sec:UStatistics} 
In this section we specialise to the setting of $U$-Statistics. Suppose we are interested in estimating a parameter of the form $\theta = \theta(P) = \mathbb{E}\bigl\{\phi\bigl(Z_1,\ldots, Z_r\bigr)\bigr\}$, for $r \geq 1$ and some function $\phi : (\mathbb{R}^d \times \mathbb{R})^{\otimes{r}} \rightarrow \mathbb{R}$, which is permutation symmetric in its $r$ arguments.  In the non-missing setting, an unbiased estimator of $\theta$ is given by 
\[
\hat{\theta} = \frac{1}{{n \choose r}  } \sum_{\{i_1, \ldots, i_r\} \subseteq \{1, \ldots, n\}} \phi\bigl(Z_{i_1},\ldots, Z_{i_r}\bigr),
\]
where the sum is taken over all unordered subsets $\{i_1, \ldots, i_r\} \subseteq \{1, \ldots, n\}$ of size $r$.  Statistics of this form have been studied in detail in the non-missing setting, see for instance \citet[Chapter 12.1]{vanderVaart:1998}.  In particular, if $\mathbb{E}\bigl\{\phi^2\bigl(Z_1,\ldots, Z_r\bigr)\bigr\} < \infty$, then
$n^{1/2} (\hat{\theta} - \theta) \rightarrow^d N(0, r^2 \psi_{U})$, where $\psi_{U} = \mathrm{Cov}\{ \phi\bigl(Z_1,\ldots, Z_r\bigr), \phi\bigl(Z_1, Z_{r+1},\ldots, Z_{2r-1}\bigr)\}.$

We now construct the CAM $U$-Statistic.  First, the complete-case $U$-Statistic is
\[
\hat{\theta}_0 = \hat{\theta}_{A_0,0} = \frac{1}{{n_0 \choose r}  } \sum_{\{i_1, \ldots, i_r\} \subseteq A_0} \phi\bigl(Z_{i_1},\ldots, Z_{i_r}\bigr),
\]
where now the sum is taken over all unordered subsets $\{i_1, \ldots, i_r\} \subseteq A_0$ of size $r$.   In this case, for $m \in \mathcal{M}$, we have
\[
\mathbb{E}(\hat\theta_{0}| \mathcal{T}_{A_0,m}) = \frac{1}{{n_0 \choose r}  } \sum_{\{i_1, \ldots, i_r\} \subseteq A_0} \mathbb{E}\{\phi\bigl(Z_{i_1},\ldots, Z_{i_r}\bigr) | Z_{i_1}^m, \ldots, Z_{i_r}^m\}.
\]
In other words, the optimal form of the adjustment term in the construction of the CAM estimator is itself a $U$-Statistic with kernel $\mathbb{E}\{\phi\bigl(Z_{1},\ldots, Z_{r}\bigr) | Z_{1}^m, \ldots, Z_{r}^m\}$.  This depends on $\phi$ and the conditional distribution of $Z$ given $Z^m$, which is typically unknown.  
We therefore consider a general construction of the adjustment term as follows: for $m \in \mathcal M$, let $\phi_m : (\mathbb{R}^{d_m} \times \mathbb{R})^{\otimes{r}} \rightarrow \mathbb{R}$ be a permutation symmetric function in its $r$ arguments, and, for $A \subseteq \{1, \ldots, n\}$, define
\[
\hat{\theta}_{A,m} = \frac{1}{{|A| \choose r}  } \sum_{\{i_1, \ldots, i_r\} \subseteq A} \phi_m\bigl(Z_{i_1}^m,\ldots, Z_{i_r}^m\bigr).
\]
We will make use of $\hat{\theta}_{0, m} = \hat{\theta}_{A_0,m}$ and $\hat{\theta}_{m} = \hat{\theta}_{A_m,m}$.    Recall that $\hat{\theta}_{0,\mathcal M} = (\hat{\theta}_{0,m} : m \in \mathcal M)^T$ and $\hat{\theta}_{\mathcal M} = (\hat{\theta}_{m} : m \in \mathcal M)^T$.  For  $\gamma \in \mathbb{R}^{|\mathcal M|}$, define the CAM $U$-Statistic
\[
\hat{\theta}_{\gamma}^{\mathcal M} := \hat{\theta}_0 - \gamma^T (\hat{\theta}_{0,\mathcal M} - \hat{\theta}_{\mathcal M}).
\]  
\sloppy Here $\phi_m$ is left unspecified, in practice one would aim to choose $\phi_m$ to mimic to optimal choice above.  In fact, we show in Theorem~\ref{thm:Ustat} that the CAM $U$-Statistic has smaller mean squared error than the complete-case $U$-Statistic whenever $\mathrm{Cov}\{\phi\bigl(Z_{1},\ldots, Z_{r}\bigr),\phi_m\bigl(Z_{1}^m, Z_{r+1}^m,\ldots, Z_{2r-1}^m\bigr)\} \neq 0$.

Suppose that the data is missing completely at random. Then we have that $B_{\mathcal M} = \mathbb{E} ( \hat{\theta}_{0,\mathcal M} - \hat{\theta}_{\mathcal M}) = 0$. Recall also that $\Omega = \mathrm{Cov}(\hat{\theta}_{0}, \hat{\theta}_{0,\mathcal M})$ and $\Lambda = \var(\hat{\theta}_{0,\mathcal M} - \hat{\theta}_{\mathcal M})$.  It  follows directly from Proposition~\ref{prop:MSE} that $\mathrm{MSE}(\hat{\theta}_{\gamma}^{\mathcal M}) - \mathrm{MSE}(\hat{\theta}_{0}) =  \gamma^T \Lambda \gamma - 2\gamma^T \Omega.$

Our next two results concern the asymptotic properties of the CAM $U$-Statistic. Let $\Omega_{U}$ be the $|\mathcal M|$-dimensional vector with entries
\[
\Omega_{U,m} :=  \mathrm{Cov}\bigl\{ \phi \bigl(Z_{1},\ldots, Z_{r}\bigr), \phi_{m}\bigl(Z_{1}^{m}, Z_{r+1}^{m},\ldots, Z_{2r-1}^{m}\bigr)\bigr\}.
\]
Further, let $\Lambda_{U}$ be the $|\mathcal M| \times |\mathcal M|$ symmetric matrix with diagonal entries
\[
\Lambda_{U, m, m} := \Bigl(1 + \frac{p_0}{p_m}\Bigr) \mathrm{Cov}\bigl\{ \phi_{m} \bigl(Z_{1}^{m},\ldots, Z_{r}^{m}\bigr), \phi_{m}\bigl(Z_{1}^{m},Z_{r+1}^{m},\ldots, Z_{2r-1}^{m}\bigr)\bigr\} 
\]
and off-diagonal entries
\[
\Lambda_{U, m_1, m_2} := \mathrm{Cov}\bigl\{ \phi_{m_1} \bigl(Z_{1}^{m_1},\ldots, Z_{r}^{m_1} \bigr), \phi_{m_2}\bigl(Z_{1}^{m_2}, Z_{r+1}^{m_2},\ldots, Z_{2r-1}^{m_2}\bigr)\bigr\}.
\]
We see in Theorem~\ref{thm:Ustat} that, under moment assumptions, we have $\Omega \rightarrow \Omega_{U}$ and $\Lambda  \rightarrow \Lambda_{U}$ as $n \rightarrow \infty$, and that the CAM $U$-Statistic is unbiased and asymptotically Gaussian.   

\begin{theorem} 
\label{thm:Ustat}
Suppose the data is missing completely at random, $\mathbb{E}\{\phi^2\bigl(Z_1,\ldots, Z_r\bigr)\} < \infty$ and, for $m \in \mathcal{M}$, $\mathbb{E}\{\phi_m^2\bigl(Z_1^m,\ldots, Z_r^m\bigr)\} < \infty$. 
Then, for $\gamma \in \mathbb{R}^{|\mathcal M|}$,
\[
\sqrt{n_0} (\hat{\theta}_{\gamma}^{\mathcal M} - \theta) \rightarrow^d N\bigl(0,  r^2 (\psi_U + \gamma^T \Lambda_U \gamma - 2\gamma^T \Omega_U ) \bigr)
\]
as $n \rightarrow \infty$.
\end{theorem}
Theorem~\ref{thm:Ustat} shows that an asymptotically optimal choice of $\gamma$, which minimises the asymptotic variance of the CAM $U$-Statistic, is $\gamma^* = \Lambda_U^{-1} \Omega_U$. The optimal leading order asymptotic variance reduction is $n_0^{-1}r^2\Omega_U^T \Lambda_U^{-1} \Omega_U$. This is of the same order as the asymptotic variance of the complete-case $U$-Statistic, which is $n_0^{-1}r^2 \psi_U$.   Of course $\Omega_U$ and $\Lambda_U$ are typically unknown.  However, to estimate these we can further exploit the use of $U$-Statistics.   Note that 
\begin{align*}
\Omega_{U, m}  & = \frac{1}{2} \mathbb{E}\bigl[\bigl\{ \phi(Z_{1},\ldots, Z_{r}) - \phi(Z_{2r}, \ldots, Z_{3r-1})\bigr\} 
\\ & \hspace{150pt} \bigl\{ \phi_m(Z^m_{1}, Z^m_{r+1},\ldots, Z^m_{2r-1}) - \phi_{m}(Z_{2r}^{m}, Z_{3r}^{m},\ldots, Z_{4r-2}^{m})\bigr\} \bigr].
\end{align*}
Thus, we can estimate $\Omega_{U}$ using a $U$-Statistic of order $4r-2$ (see \eqref{eq:OmegaU} in Section~\ref{sec:Utech}).  A similar expression can be derived for the entries of $\Lambda_{U}$, but for brevity we exclude the formulas here -- they are given in \eqref{eq:Lambdam} and \eqref{eq:Lambdam12} in Section~\ref{sec:Utech}.  Let  $\hat{\Lambda}_U$ and $\hat{\Omega}_U$ denote resulting $U$-Statistic estimators of $\Lambda_U$ and $\Omega_U$, respectively. (In practice, averaging over all subsamples of size $4r-2$ will be computationally expensive, in our simulations $\hat{\Omega}_U$ and $\hat{\Lambda}_{U}$ are approximated using $10^5$ random subsamples.)

Now, let $\hat{\gamma} := \hat{\Lambda}_U^{-1}\hat{\Omega}_U$ and consider the practical CAM $U$-Statistic $\hat{\theta}^{\mathcal M}_{\hat{\gamma}}$. Theorem~\ref{thm:Ustat2} shows that we can mimic the performance of the optimal CAM $U$-Statistic $\hat{\theta}^{\mathcal M}_{\gamma^*}$ using the data-driven choice of $\gamma$. 
\begin{theorem} 
 	\label{thm:Ustat2}
 	Suppose the data is missing completely at random, $\mathbb{E}\{\phi^4\bigl(Z_1,\ldots, Z_r\bigr)\} < \infty$ and, for $m \in \mathcal M$, $\mathbb{E}\{\phi_m^4\bigl(Z_1^m,\ldots, Z_r^m\bigr)\} < \infty$. 
 	Then
 	\begin{equation}
	\label{eq:UStat2}
 	\sqrt{n_0} (\hat{\theta}_{\hat{\gamma}}^{\mathcal M} - \theta) \rightarrow^d N\bigl(0,  r^2 (\psi_U - \Omega_U \Lambda_U^{-1}\Omega_U ) \bigr).
 	\end{equation}
 \end{theorem}
 The asymptotic variance in~\eqref{eq:UStat2} can be estimated by plugging in the estimators of $\psi_U$,  $\Omega_U$, and $\Lambda_U$.   Here a $U$-Statistic estimator of $\psi_U$, denoted by $\hat \psi_U$, can be constructed in the same way as $\hat\Omega_U$ and $\hat\Lambda_U$.  Then one can show that $\sqrt{n_0}(\hat\psi_U - \hat\Omega_U \hat\Lambda_U^{-1}\hat\Omega_U )^{-1/2} (\hat{\theta}_{\hat{\gamma}}^{\mathcal M} - \theta) \rightarrow^d N(0,r^2)$, which can be used for statistical inference such as constructing  confidence intervals and testing hypotheses. 
 
 In order to understand the improvement the CAM $U$-Statistic achieves over the complete-case method it is helpful to consider some examples.

\begin{example}[Marginal mean estimation]
Suppose that $d = 1$ and we are interested in estimating the parameter $\theta = \mathbb{E}(X)$.  Suppose further that the $X$ variable is missing completely at random. 
We observe $\mathcal{T}_{A_0,0} = \{(X_i, Y_i) : m_i = 0\}$ and $\mathcal{T}_{A_1,1} = \{Y_i : m_i = 1\}$, where the respective sample sizes are $n_0$ and $n_1$. Let $n=n_0 + n_1$.  In contrast to the illustrative example in Section \ref{sec:toy}, we are not assuming the joint distribution of $(X,Y)$ is Gaussian.   In this setting the complete-case $U$-Statistic  is $\hat{\theta}_0 = \frac{1}{n_0} \sum_{i \in A_0} X_i.$  Of course, by the Central Limit Theorem, if $\mathbb{E}(X^2) < \infty$, then $\sqrt{n_0} (\hat{\theta}_0 - \theta) \rightarrow^d N(0, \var(X))$. 

Next we consider the CAM estimator with $\mathcal{M} = \{1\}$ which takes into account the information from variable $Y$. For a generic function $\phi_1 : \mathbb{R} \rightarrow \mathbb{R}$ satisfying $\var(\phi_1(Y_1))<\infty$ define
\[
\hat{\theta}_{0,1} =  \frac{1}{n_0} \sum_{i \in A_0} \phi_1(Y_i); \quad \hat{\theta}_{1} =  \frac{1}{n_1} \sum_{i \in A_1} \phi_1(Y_i).
\]
Then $\Lambda = \bigl(\frac{1}{n_0} + \frac{1}{n_1}\bigr) \var\{\phi_{1} (Y_{1})\}$ and  $\Omega = \frac{1}{n_0} \mathrm{Cov} \{X_1, \phi_{1} (Y_{1})\}$. We have that
\[
\mathrm{MSE}(\hat{\theta}_{\gamma^*}^{\mathcal M}) - \mathrm{MSE}(\hat{\theta}_{0})= \var(\hat{\theta}_{\gamma^*}^{\mathcal M}) - \var(\hat{\theta}_{0})  =  - \frac{n_1 \mathrm{Cov}^2 \{X, \phi_{1} (Y)\}}{n_0 n \var\{\phi_{1}(Y)\}} \leq 0.
\]
Thus, there is a guaranteed improvement in MSE as long as $X$ and $\phi_1(Y)$ are correlated. 

The optimal choice of $\phi_1$ in this case is $\phi_1^*(y) := \mathbb{E}(X | Y= y)$, (cf.~the discussion at the end of the previous section), and the corresponding first order variance reduction is $\frac{n_1}{n_0 n} \var\{\mathbb{E}(X|Y)\}$.  If $X$ and $Y$ are independent, then $\var\{\mathbb{E}(X|Y)\} = 0$, i.e.~as expected, the CAM estimator will not lead to an improvement over the complete-case estimator, since the $Y$ variable tells us nothing about the marginal $X$ distribution. On the other hand, in the pathological case that $X$ can be written as a deterministic function of $Y$, we see that the variance reduction is  $\frac{n_1}{n_0n}\var(X)$. Consequently,  $\var(\hat{\theta}_{\gamma^*}^\mathcal{M}) = \frac{\var(X)}{n_0}\bigl(1 - \frac{n_1}{n}\bigr) = \frac{\var(X)}{n}$, i.e. the variance that could be achieved by using a fully observed dataset!

Of course the regression function $\mathbb{E}(X | Y= y)$ will typically be unknown to the user. Consider instead therefore the practical choice $\phi_m(y) := y$. Then we have that
\[
\var(\hat{\theta}_{\gamma^*}^{\mathcal M}) =  \frac{\var(X)}{n_0} \Bigl\{1 - \frac{n_1 }{n}\mathrm{Corr}^2 (X, Y)\Bigr\}. 
\]
In fact, the above derivation holds as long as $\phi_m(y)$ is a linear function of $y$.

In the left panel of Figure~\ref{fig:UStat1}, we present the sampling distributions of the complete-case and CAM $U$-Statistic for the mean of $X \sim \mathrm{Exp}(1)$, where $Y | \{X = x\} \sim N(x, \sigma^2)$. We set $n = 1000$, $\sigma = 0.2$ and the $X$ variable is missing with probability $0.5$.  We present the results for the practical choice $\phi_m(y) = y$ and the optimal choice $\phi_m(y) = \mathbb{E}(X | Y = y) = y - \sigma^2 + \sigma\frac{\Phi'(\sigma - y/\sigma)}{1 - \Phi(\sigma - y/\sigma)}$, where $\Phi'(\cdot)$ and $\Phi(\cdot)$ denote the standard Normal density and distribution function, respectively.  The variance reduction can be clearly seen from the plots.  We see also that the practical CAM $U$-Statistic has very similar performance to the optimal CAM $U$-Statistic. 
\end{example}

 \begin{figure}[ht!]
 	\centering
 	\includegraphics[width=0.45\textwidth]{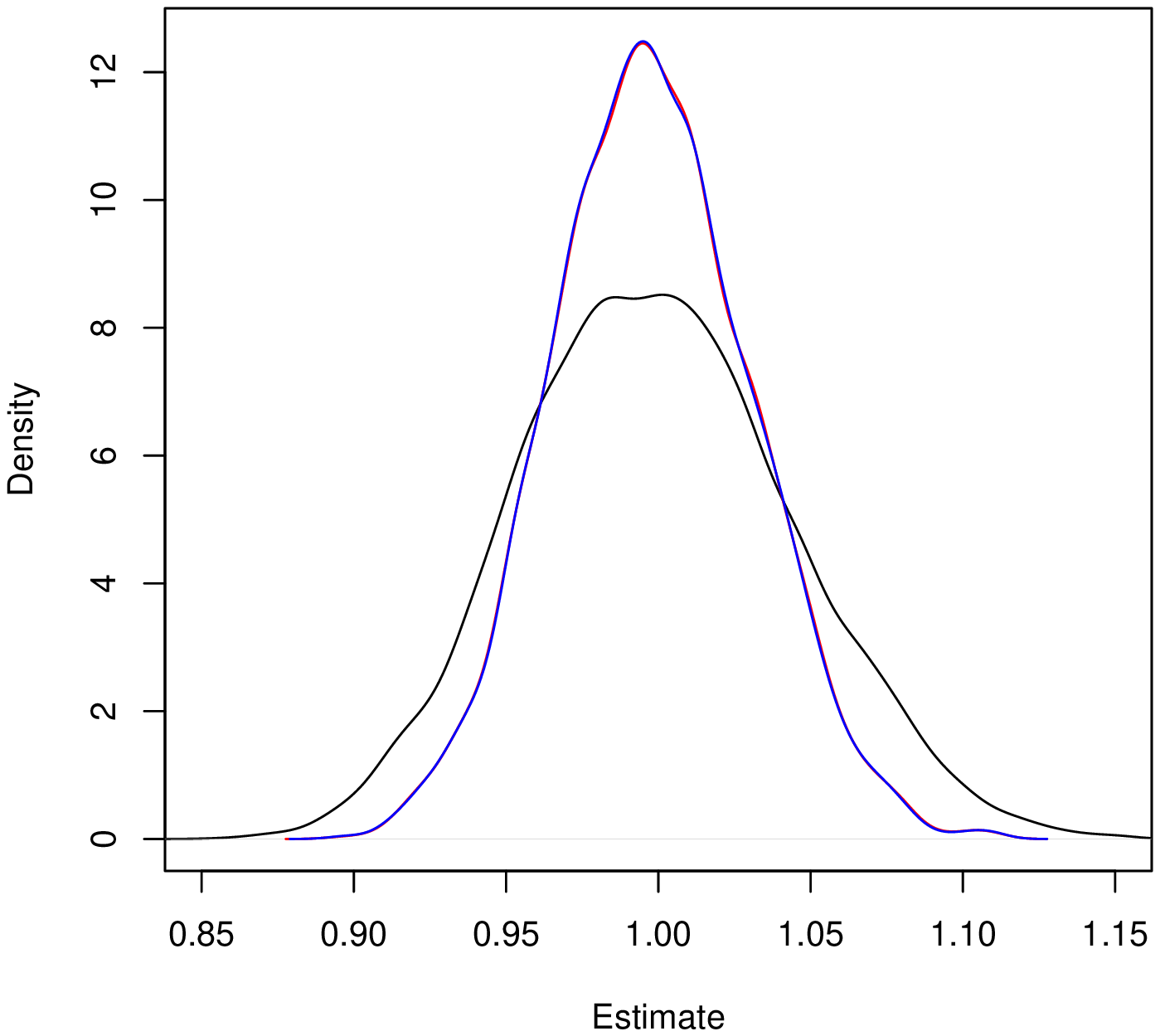} 	\includegraphics[width=0.45\textwidth]{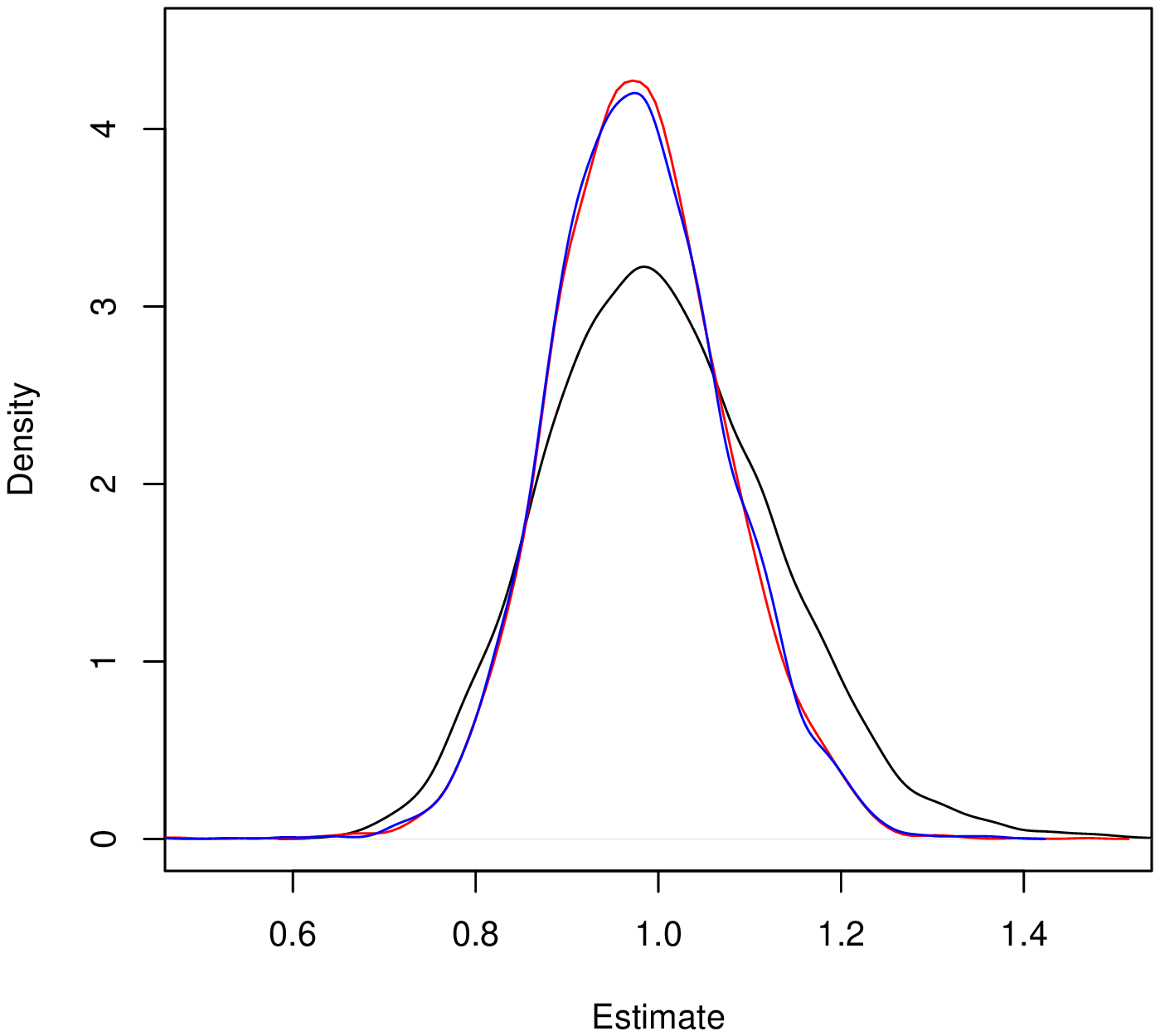} 
 	\caption{Sampling distributions of the complete-case $U$-Statistic $\hat{\theta}_0$ in black, the CAM $U$-Statistic $\hat{\theta}_{\hat{\gamma}}^{\mathcal M}$ in red and the optimal CAM $U$-Statistic (i.e. with $\phi_m = \phi_m^*$ and $\gamma = \gamma^*$) in blue for Example~1 (left) and~2 (right).}
 	\label{fig:UStat1}
 \end{figure}

\begin{example}[Covariance estimation]
Consider the same set-up as in Example~1, but suppose now we are interested in the parameter $\theta = \mathrm{Cov}(X,Y) = \frac{1}{2} \mathbb{E}\{(X_1 - X_2)(Y_1 - Y_2)\}$.   In this case, we have the complete-case $U$-Statistic 
\[
\hat{\theta}_0 = \frac{1}{2{n_0 \choose 2}  } \sum_{\{i, j\} \subseteq A_0} (X_{i} - X_j)(Y_{i} - Y_{j}).
\]
If  $\mathbb{E}\{(X_1 - X_2)^2(Y_1 - Y_2)^2\} < \infty$, then by \citet[Theorem 12.3]{vanderVaart:1998} we have that $\sqrt{n_0} (\hat{\theta}_0 - \theta) \rightarrow^d N(0, \psi_1 )$,  where $\psi_1 := \frac{1}{4}\mathrm{Cov}\{(X_1 - X_2)(Y_1 - Y_2), (X_1 - X_3)(Y_1 - Y_3)\}$.

Now, for a generic function $\phi_1 : \mathbb{R}^2 \rightarrow \mathbb{R}$ consider
\[
\hat{\theta}_{0,1} =  \frac{2}{n_0(n_0-1)} \sum_{\{i,j\} \subseteq A_0} \phi_1(Y_i,Y_j); \quad \hat{\theta}_{1} =  \frac{2}{n_1(n_1 - 1)} \sum_{\{i,j\} \subseteq A_1} \phi_1(Y_i,Y_j).
\]
Here the optimal function 
 is $\phi_1^*(y_1,y_2) =\frac{1}{2} \{\mathbb{E}(X | Y = y_1) - \mathbb{E}(X | Y =  y_2)\}(y_1 - y_2)$.  Recall that $p_1 = \lim_{n\rightarrow \infty} \frac{n_1}{n}$. The corresponding CAM U-statistic satisfies
\[
\lim_{n \rightarrow \infty} n_0 \var (\hat{\theta}^{\mathcal{M}}_{\gamma^*} )= \psi_1 - \frac{p_1}{4} \mathrm{Cov}\{\phi_1^*(Y_1,Y_2), \phi_1^*(Y_1,Y_3)\}.
\]
Thus, we have first order variance reduction as long as
$\mathrm{Cov}\{\phi_1^*(Y_1,Y_2), \phi_1^*(Y_1,Y_3)\}\neq 0$.
 
However, since  $\phi^*_1$ is generally unknown, consider the practical choice $\phi_1(y_1,y_2) = \frac{1}{2}(y_1 - y_2)^2$.  This is motivated by supposing that $\mathbb{E}(X | Y = y)$ is a linear function of $y$.  Then we have
\[
\lim_{n\rightarrow \infty} n_0 \var(\hat{\theta}_{\gamma^*}^\mathcal{M})  =  \psi_1 -  \frac{p_1 \mathrm{Cov}^2\{(X_1 - X_2)(Y_1 - Y_2), (Y_1 - Y_3)^2\}}{4\mathrm{Cov}\{(Y_1 - Y_2)^2,(Y_1 - Y_3)^2\}}
\]
In the right panel of Figure~\ref{fig:UStat1}, we present the sampling distributions of the complete-case, practical CAM, and optimal CAM $U$-Statistics for $\theta = \mathrm{Cov}(X,Y)$.  As in Example~1, the data generating distribution is $X \sim \mathrm{Exp}(1)$ and $Y | \{X = x\} \sim N(x, \sigma^2)$, $n = 1000$, and  $p_1 = 0.5$. 
\end{example}

\subsection{The Terneuzen birth cohort dataset}
We now demonstrate how the CAM $U$-Statistic can be used in practice.  In particular, we will apply our proposal from the previous two examples to the Terneuzen birth cohort data available from the \texttt{mice} package on CRAN.  The full dataset consists of 3951 observations of 11 features covering 306 people.  We simplify the problem by taking a subset of the data and only include the first measurement for each person. Furthermore, we retain only 4 of the features, namely ``sex'', ``height $Z$-score'', ``weight $Z$-score'', and ``bmi $Z$-score''.  In the resulting dataset, there are 306 observations (one for each patient), of which 105 are missing both the height and bmi features.  In order to fit this in the framework introduced above let $Y$ denote sex (1 for female, 0 for male), and let $X$ be the 3-dimensional vector of weight, height, and bmi.  

We have 201 complete cases in $A_0$, and 105 cases in $A_{m}$ for $m = (1,0,1)^T$, where only $Y$ (sex) and $X^{(2)}$ (weight) are observed.  We consider two problems; (i) to estimate the average bmi $Z$-score in the cohort, and (ii) estimate the covariance between the height and weight $Z$-scores. In both cases, we consider two choices of $\phi_m$, a simple choice, and a regression estimate. 

In problem (i), recall that we can write the marginal mean as a $U$-Statistic with $\phi(Z) = X^{(3)}$, and the complete-case estimator is $\frac{1}{n_0}\sum_{i \in A_0} X^{(3)}_i$.  To construct the CAM $U$-Statistics, we first consider $\phi_m(Z^m) = X^{(2)}$, i.e.~the weight $Z$-score.  For our second choice, write 
\begin{equation}
\label{eq:Tbc1}
\phi_m(Z^m) = \beta_0 + \beta_1 Y + \beta_2 X^{(2)} + \beta_3 X^{(2)} Y.
\end{equation}
We choose $(\beta_0, \beta_1, \beta_2, \beta_3)$ by fitting a linear model (with an interaction) of bmi on height and sex using the complete cases (in \texttt{R} this is simply done using the \texttt{lm} function). The idea here is to approximate $\phi^*_m(Z^m) = \mathbb{E}\{\phi(Z)| Z^m \}$.  Then $\hat{\theta}_{0,m}$ is the sample average of the fitted values, and $\hat{\theta}_{m}$ is the average of the predictions made on the data in $A_m$.  

\begin{table}[ht!]
	\centering
	\caption{\label{tab:tbc} Comparison of the complete-case and CAM $U$-Statistics using the Terneuzen birth cohort dataset}
	\begin{tabular}{l | c | c | c}
		Method & Point est. & 95$\%$ CI & CI width \\
		\hline
		\multicolumn{4}{l}{(i) Marginal mean of bmi score}\\ 
		\hline
		Complete-case & 				0.55 & (0.38, 0.73) & 0.35\\
		CAM: $\phi_m(Z^m) = X^{(2)}$ & 	0.55 & (0.44, 0.66) & \textbf{0.22}\\
		CAM: $\phi_m(Z^m)$ linear -- see \eqref{eq:Tbc1}       & 	0.55 & (0.44, 0.66) & \textbf{0.22}\\
		\hline
		\multicolumn{4}{l}{(ii) Covariance between height and weight}\\
		\hline
		Complete-case & 				 1.27 & (0.50, 2.04) & 1.54\\
		CAM: $\phi_m(Z_1^m,Z_2^m) = \frac{1}{2} (X_1^{(2)} - X_2^{(2)})^2$ & 	 1.17 & (0.66, 1.68) & 1.02\\
		CAM: $\phi_m(Z_1^m,Z_2^m)$ linear  -- see \eqref{eq:Tbc2} 		&1.19 & (0.70, 1.69) & \textbf{0.99}\\
	\end{tabular}
\end{table}

For problem (ii), we have $\phi(Z_1, Z_2) = \frac{1}{2} (X_1^{(1)} - X_2^{(1)})(X_1^{(2)} - X_2^{(2)})$.  In this case, to construct the CAM estimator, we first use the simple choice $\phi_m(Z_1^m, Z_2^m) = (X_1^{(2)} - X_2^{(2)})^2$. Then, similarly to the previous problem we consider 
\begin{equation}
\label{eq:Tbc2}
\phi_m(Z_1^m,Z_2^m) = \{(\beta_1 (Y_1 - Y_2) + \beta_2 (X_1^{(2)} - X_2^{(2)}) + \beta_4 (X_1^{(2)} Y_1 - X_2^{(2)} Y_2)\}(X_1^{(2)} - X_2^{(2)}),
\end{equation}
where, again, the idea is to approximate the optimal $\phi^*_m(Z_1^m, Z_2^m) = \mathbb{E}\{\phi(Z_1, Z_2)| Z_1^m, Z_2^m \}$. 

We compare the performance to the complete-case estimator with both versions of the CAM $U$-Statistic in Table~\ref{tab:tbc}. In each case, we present the point estimate, an approximate 95\% confidence interval based on the result in Theorem~\ref{thm:Ustat2} and the corresponding interval width.  We see that the CAM estimator has a much narrower interval width in both problems.  Moreover, in problem (i) the two CAM approaches lead to identical results (up to 3 significant figures), whereas in the second problem, the second CAM $U$-Statistic performs slightly better. 

\section{Nonparametric statistical learning}
\label{sec:RCD}
We now study two fundamental statistical learning problems, namely density estimation and regression.   Typically in these problems, we are interested in estimating a function from $\mathbb{R}^d$ to $\mathbb{R}$, we show how the CAM estimator can be applied locally, i.e. for each $x \in \mathbb{R}^d$.    Throughout this section we assume that the data is missing completely at random.

Density estimation and regression are canonical problems in statistics, and many nonparametric approaches have been proposed and studied in detail -- see, for instance, \citet{Rosenblatt:1956}, \citet{Parzen:1962}, \citet{Wahba:1990}, \citet{WandJones:1995}, \citet{FanGijbels:1996}, \citet{CarrollRuppertWalsh:1998}, \citet{Tsybakov:2004} and \citet{BiauDevroye:2015}.  We focus our study kernel based methods.

\subsection{Kernel density estimation} 
\label{sec:density}
In this subsection, assume we only observe $X_1^{m_1}, \ldots, X_n^{m_n}$ and we are interested in estimating $f_X$, the density of the marginal distribution of $X$. We specialise the setting introduced in Section \ref{sec:setting} by letting $\theta(P) = f_X(x)$, the marginal density of $X$ at a fixed $x \in \mathbb{R}^d$. 

Let $h > 0$ be the bandwidth and, for $m \in \mathcal{M}$, let $K_m : \mathbb{R}^{d_m} \rightarrow [0,\infty)$ be a $d_m$-dimensional Kernel function.  For $x \in \mathbb{R}^d$, $A \subseteq \{1,\ldots,n\}$ and $m \in \mathcal{M}$, let 
\begin{equation}
\label{eq:kdedef}
\hat{f}_{A,m}(x^{m}) = \hat{f}_{A,m,h,K_m}(x^{m}) :=  \frac{1}{|A| h^{d_m}} \sum_{i \in A} K_m\Bigl( \frac{X_i^{m} - x^{m}}{h}\Bigr).  
\end{equation}
This can be thought of as an estimator of the marginal density $f_{X^m}(x^m)$ of $X^m$ at $x^m$. In particular, the complete-case estimator of $f_X(x)$ is $\hat{f}_{0}:= \hat{f}_{A_0,0,h,K}(x)$, where $K = K_0$ is the $d$-dimensional kernel.  Our CAM density estimator is constructed using $\hat{f}_{0}$, as well as $\hat{f}_{0,m} := \hat{f}_{A_0,m,h,K_m}(x^{m})$ and $\hat{f}_{m} := \hat{f}_{A_m,m,h,K_m}(x^{m})$, for $m \in \mathcal{M}$.  To understand our choice of $\hat{f}_{0,m}$ and $\hat{f}_{m}$, recall from the discussion after Proposition~\ref{prop:MSE} that the optimal choice of $\hat{f}_{0,m}$ is
\begin{equation}
\label{eq:oracleK}
\hat{f}^*_{0,m} = \mathbb{E} (\hat{f}_{0}|\mathcal T_{A_0,m}) = \frac{1}{n_0 h^{d}} \sum_{i \in A_0} \mathbb{E}\Bigl\{K\Bigl( \frac{X_i - x}{h}\Bigr) \Big| X_i^m\Bigr\}. 
\end{equation} 
This takes a similar form as the kernel density estimator in \eqref{eq:kdedef}.  At the end of this subsection we will see that $\hat{f}^*_{0,m}$ can often be well 
	approximated using a local constant estimator with a practical choice of kernel $K_m$ that depends on $K$.  Finally, we have also chosen $\hat{f}_{m}$ so that $\mathbb{E}(\hat{f}_{m}) = \mathbb{E}(\hat{f}_{0,m})$.

Now let $\hat{f}_{0,\mathcal{M}} := (\hat{f}_{0,m} : m \in \mathcal{M})^T \in [0,\infty)^{|\mathcal{M}|}$, and $\hat{f}_{\mathcal{M}} := (\hat{f}_{m} : m \in \mathcal{M})^T \in [0,\infty)^{|\mathcal{M}|}$.  
Then, for $\gamma \in \mathbb{R}^{|\mathcal{M}|}$, define the CAM kernel density estimator
\[
\hat{f}_{\gamma}^{\mathcal{M}} = \hat{f}_{\gamma}^{\mathcal{M}}(x)  := \hat{f}_0  -  \gamma^T (\hat{f}_{0,\mathcal{M}} - \hat{f}_{\mathcal{M}}).
\]

Our theoretical results in this section will make use of conditions \textbf{A1} and \textbf{A2} given in Section~\ref{sec:RCDtech}.  The Lipshitz assumption on $f_X$ and $f_{X^m}$ in $\textbf{A1}$ allow us to approximate the accuracy of the kernel density estimates of $f_X$ and $f_{X^m}$, respectively.  Whereas the assumption on the Kernel functions in $\textbf{A2}$ is satisfied by many commonly used kernels. 

For an estimate $\hat{f}$ of $f_X(x)$, let $\mathrm{MSE}(\hat{f}) = \mathrm{MSE}(\hat{f})(x) := \mathbb{E}[\{\hat{f}  - f_X(x)\}^2]$.  Under our assumptions, it is well-known that the complete-case estimator satisfies $\mathrm{MSE}(\hat{f}_{0}) = O(1/(n_0 h^d) + h^2)$ as $n_0 \rightarrow \infty$; see, for example, \citet[Propositions 1.1 and 1.2]{Tsybakov:2004} for the $d=1$ case. 

Let $\nu = \nu(K) := \int_{\mathbb{R}^d} K^2(z) \, dz < \infty$, and, for each $m \in \mathcal{M}$, let $\nu_m = \nu_{m}(K_m) := \int_{\mathbb{R}^{d_m}} K_m^2(z) \, dz < \infty$. Furthermore, for $m_1 \neq m_2  \in \mathcal{M}$, let $m^{1,2} = \mathrm{pmax}\{m_1,m_2\} \in \{0,1\}^d$ and $m_{1,2} = \mathrm{pmin}\{m_1,m_2\} \in \{0,1\}^d$ denote the entrywise maximums and minimums, respectively, of $m_1$ and $m_2$. For $m_1 \neq m_2$,  let $\nu_{m_1,m_2} = \nu_{m_1,m_2}(K_{m_1},K_{m_2}) := \int_{\mathbb{R}^{d_{m_{1,2}}}} K_{m_1}(z^{m_1}) K_{m_2}(z^{m_2})  \, dz^{m_{1,2}}$. 

Our next result shows that the asymptotic difference between the mean squared error of our proposal and the complete-case estimator can be written in terms of $\gamma$, the $|\mathcal{M}|$-dimensional vector $\Omega_{\mathrm{D}} := (\frac{\nu_{0,m}f_X(x)}{n_0 h^{d_m}} : m \in \mathcal{M})^T$ and $\Lambda_{\mathrm{D}}$, the symmetric $|\mathcal{M}|\times|\mathcal{M}|$ matrix with entries
\begin{equation} 
\label{eq:lambdaD}
\Lambda_{\mathrm{D}, m,m} :=  \frac{\nu_{m}f_{X^m}(x^m) }{h^{d_m}}\Bigl(\frac{1}{n_0} + \frac{1}{n_m}\Bigr); \quad \Lambda_{\mathrm{D},m_1, m_2} :=  \frac{\nu_{m_{1},m_{2}}f_{X^{m_{1,2}}}(x^{m_{1,2}}) } {n_0 h^{d_{m^{1,2}}}}, 
\end{equation}
for $m, m_1 \neq m_2 \in \mathcal{M}$.

\begin{theorem}
	\label{thm:DensityPi} 
	Assume \textbf{A1} and \textbf{A2}. For $0 < \alpha < \beta < 1/d$, we have
   \[
	\mathrm{MSE}(\hat{f}^\mathcal{M}_{\gamma}) - \mathrm{MSE}(\hat{f}_{0}) =  \bigl(\gamma^T \Lambda_{\mathrm{D}} \gamma - 2\gamma^T\Omega_{\mathrm{D}}\bigr)\{1 + o(1)\}
	\]
	as $n \rightarrow \infty$, uniformly for $h \in [n^{-\beta}, n^{-\alpha}]$. 
	\end{theorem} 
The optimal $\gamma$, which maximises the leading order asymptotic improvement in mean squared error over the complete-case estimator, is $\gamma_{\mathrm{D}}^* :=  \Lambda_{\mathrm{D}}^{-1}\Omega_{\mathrm{D}}$. This leads to an improvement of $ \Omega_{\mathrm{D}}^T \Lambda_{\mathrm{D}}^{-1} \Omega_{\mathrm{D}}$.  Suppose $\mathcal{M} = \{m\}$, then the optimal asymptotic improvement simplifies to 
\[
\frac{n_m \nu_{0,m}^2 f_X(x)^2}{n_0 h^{d_m}(n_0 + n_m)\nu_{m} f_{X^m}(x^m)} = \frac{n_m \nu_{0,m}^2 f_X(x) f_{X | X^m} (x|x^m)}{n_0h^{d_m}(n_0 + n_m)\nu_{m}} = O\Bigl(\frac{1}{n_0h^{d_m}}\Bigr),
\]
where $f_{X | X^m} (x|x^m)$ denotes the conditional density of $X$ at $x$ given $X^m = x^m$. We see that a larger improvement is possible when $n_m$ is large compared to $n_0$, or if $f_{X | X^m} (x|x^m)$ is large. Note, however, that we only obtain a second order improvement over the complete-case approach. To understand this further, in contrast to the $U$-Statistics setting, in the density estimation problem $\hat{f}_0$ and $\hat f_{0,m}-\hat{f}_m$ have different convergence rates, with the later converging at a faster rate because of the smaller dimension $d_m$. Therefore the covariance between $\hat{f}_0$ and $\hat f_{0,m}-\hat{f}_m$ is negligible compared to the asymptotic variance of $\hat{f}_0$.   Nevertheless, we will see in our numerical study in Section~\ref{sec:sims}, that the improvement CAM offers over the complete-case method is appreciable in finite sample problems.   Of course, $ \Omega_{\mathrm{D}}$ and $\Lambda_{\mathrm{D}}$ are unknown.  Nonetheless, we have an immediate corollary that for any $\gamma$ such that $\gamma^T \Lambda_{\mathrm{D}} \gamma < 2\gamma^T\Omega_{\mathrm{D}}$ the corresponding CAM estimator will lead to an improvement over the complete-case approach. 

It remains to propose practical choices of tuning parameters. First, we suppose that the complete-case kernel $K$ and bandwidth $h$ are given to us; if needed these can be chosen using cross-validation on the complete cases.   Now, to choose $\gamma$  we attempt to approximate the optimal choice $\gamma = \Lambda_{\mathrm{D}}^{-1}\Omega_{\mathrm{D}}$.  More precisely, let $\hat{\gamma}_{\mathrm{D},m} = \frac{\nu_{0,m} n_m \hat{f}_0}{\nu_{m}(n_0\hat{f}_{0,m} + n_m \hat{f}_m)}$, where we have used $\hat{f}_0$ and $\frac{n_0\hat{f}_{0,m} + n_m \hat{f}_m}{n_0 + n_m}$ as estimates of $f_X(x)$ and $f_{X^m}(x^m)$, respectively.  We also approximate the off-diagonal terms in $\Lambda_{\mathrm{D}}$ by $0$, since they are of smaller order than the terms on the diagonal, i.e.~$\Lambda_\mathrm{D} = \mathrm{diag}(\Lambda_\mathrm{D})\{1+o(1)\}$ -- see \eqref{eq:lambdaD}.   

We choose the kernel $K_m$ in an attempt to mimic the optimal choice in \eqref{eq:oracleK}.  Lemma~\ref{lem:Kernel} in Section \ref{sec:RCDtech} shows that for a large family of kernels, under appropriate smoothness conditions on $f_{X|X^m}$, we can approximate $\hat{f}_{0,m}^*$ up to first order using a kernel density estimator with practical choice of kernel $K_m$ that depends only on $K$.  For instance, if $K$ is the Gaussian kernel $K(t) = \frac{1}{ (2\pi)^{d/2}} \exp(-\|t\|^2/2)$, for $t \in \mathbb{R}^d$, then $\hat{f}_{0,m}^*$ is well-approximated by using the $d_m$-dimensional Gaussian kernel $K_m(z) = \frac{1}{(2\pi)^{d_m/2}} \exp(-\|z\|^2/2)$,  for $z \in \mathbb{R}^{d_m}$, in \eqref{eq:kdedef}. 

\subsection{Local constant regression}
\label{sec:reg}
We now consider the standard homoscedastic nonparametric regression problem, where the pair $(X,Y)$ takes values in $\mathbb{R}^d \times \mathbb{R}$ and satisfies the relationship
\[
Y = \eta(X) + \sigma \epsilon.
\]
Here $\sigma > 0$ and $\eta: \mathbb{R}^d \rightarrow \mathbb{R}$ is the regression function, i.e. $\eta(x) := \mathbb{E}(Y | X = x)$. The random variable $\epsilon$ has mean zero and variance one, and is independent of $X$.  We are interested in estimating $\theta(P) = \eta(x)$, the regression function at a fixed $x\in \mathbb{R}^d$.  

Consider also the regression model when $X$ is missing the features $m \in \{0,1\}^d$. It is convenient to define $\eta_m(x^m): = \mathbb{E}( Y | X^m = x^m)$, and $\tau_m(x^m) := \mathrm{Var}\{\eta(X) | X^m = x^m\}$, where $\tau_0(\cdot) = \var\{\eta(X)|X=x\}= 0$.  Finally, for $m_1, m_2 \in \mathcal{M},$ let $\tau_{m_1,m_2}(x^{m_{1,2}}) := \mathbb{E}[\{\eta(X) - \eta_{m_1}(x^{m_1})\}\{\eta(X) - \eta_{m_2}(x^{m_2})\}|X^{m_{1,2}} = x^{m_{1,2}}]$. 

Recall the bandwidth $h>0$ and kernel function $K_m$, for $m \in \mathcal{M}$, used in the previous section.  For $x \in \mathbb{R}^d$ and $A \subseteq \{1,\ldots,n\}$, the \emph{local constant estimator} of $\eta_m(x^m)$ is  
\begin{equation}
\label{eq:loccon}
\hat{\eta}_{A,m}(x^m) = \hat{\eta}_{A,m,h,K}(x^{m})  := \argmin_{\alpha \in \mathbb{R}} \Biggl\{ \sum_{i \in A} K_m\Bigl(\frac{X^{m}_i - x^{m}}{h}\Bigr) (Y_i - \alpha)^2 \Biggr\}.
\end{equation}
In particular, the complete-case estimator of $\eta(x)$ is $\hat{\eta}_0 := \hat{\eta}_{A_0,0, h, K}(x)$, where $K = K_0$ is a $d$-dimensional kernel.  Further, let 
$\hat{\eta}_{0,m} := \hat{\eta}_{A_0,m, h, K_m}(x)$ and $\hat{\eta}_m := \hat{\eta}_{A_m,m, h, K_m}(x)$; in contrast to the density estimation setting, it is less clear why the form of $\hat{\eta}_{0,m}$ is effective here -- we postpone discussion of this until the end of this subsection.

Let $\hat{\eta}_{0, \mathcal{M}} = (\hat{\eta}_{0,m} : m \in \mathcal{M})^T$ and  $\hat{\eta}_{\mathcal{M}} = (\hat{\eta}_{m} : m \in \mathcal{M})^T$.  Then, for $\gamma \in \mathbb{R}^{|\mathcal{M}|}$, we define the CAM local constant regression estimator
\[
\hat{\eta}^{\mathcal{M}}_{\gamma} = \hat{\eta}^{\mathcal{M}}_{\gamma}(x) :=  \hat{\eta}_0 - \gamma^T (\hat{\eta}_{0, \mathcal{M}} - \hat{\eta}_{\mathcal{M}}).
\]

Our main theoretical result in this section will make use of two further assumptions on the regression function; see \textbf{A3} and \textbf{A4} given in Section \ref{sec:RCDtech}. In particular, we ask that the functions $\eta$, $\eta_m$,  $\tau_m$ and $\tau_{m_1,m_2}$ are Lipschitz.    Now, for $m \in \mathcal{M}$, let $\mu_{0,m} = \mu_{0,m}(K_m) := \int_{\mathbb{R}^{d_m}} K_m(z) \, dz < \infty$.  Further, let 
\[
\Omega_{\mathrm{R}} := \Bigl(\frac{\sigma^2 \nu_{0,m}}{\mu_{0,m}f_{X^m}(x^m) n_0 h^{d_m}} : m \in \mathcal{M}\Bigr)^T.
\]  
Let $\Lambda_{\mathrm{R}}$ be the $|\mathcal{M}| \times |\mathcal{M}|$ matrix with diagonal entries 
\[
\Lambda_{\mathrm{R},m,m} :=  \frac{\nu_{m}\{\sigma^2 + \tau_m(x^m)\}}{\mu^2_{0,m} f_{X^m}(x^m)h^{d_m}} \Bigl(\frac{1} {n_0} + \frac{1}{n_m}\Bigr),
\] 
and off diagonal entries 
\[
\Lambda_{\mathrm{R},m_1,m_2} := \frac{ \nu_{m_{1},m_{2}} f_{X^{m_{1,2}}}(x^{m_{1,2}}) \bigl\{\sigma^2 + \tau_{m_1,m_2}(x^{m_{1,2}}) \bigr\}}{\mu_{0,m_{1}}  f_{X^{m_{1}}}(x^{m_{1}})\mu_{0,m_2} f_{X^{m_{2}}}(x^{m_{2}}) n_0 h^{d_{m^{1,2}}}}.
\]
Note that the off-diagonal terms of $\Lambda_{\mathrm{R}}$ are of smaller order than the terms on the diagonal, i.e.~$\Lambda_\mathrm{R} = \mathrm{diag}(\Lambda_\mathrm{R})\{1+o(1)\}$.  Finally, for a regression estimator $\hat{\eta}$ of $\eta(x)$, we write 
\[
\mathrm{MSE}(\hat{\eta}) = \mathrm{MSE}(\hat{\eta})(x) := \mathbb{E}[\{\hat{\eta} - \eta(x)\}^2 | Z_1^{m_1}, \ldots, Z_n^{m_n}].
\]
\begin{theorem}
	\label{thm:RegressionPi}
Assume \textbf{A1}, \textbf{A2}, \textbf{A3} and \textbf{A4}.  Then, for each $0 < \alpha < \beta < 1/d$, we have 
\[
	\mathrm{MSE}(\hat{\eta}^{\mathcal{M}}_{\gamma}) - 	\mathrm{MSE}(\hat{\eta}_0) = \bigl(\gamma^T \Lambda_{\mathrm{R}}\gamma - 2\gamma^T\Omega_{\mathrm{R}}\bigr)\{1 + o_p(1)\}
\]
as $n \rightarrow \infty$, uniformly for  $h \in [n^{-\beta}, n^{-\alpha}].$
\end{theorem}
Theorem~\ref{thm:RegressionPi} gives the leading order asymptotic difference in mean squared error between the CAM estimator and the complete-case estimator. We see that the optimal leading order improvement in this case is $\Omega_{\mathrm{R}}^T\Lambda_{\mathrm{R}}^{-1}\Omega_{\mathrm{R}}$, which can achieved by taking $\gamma = \gamma^*_{\mathrm{R}} := \Lambda_{\mathrm{R}}^{-1}\Omega_{\mathrm{R}}$. Furthermore, for any $\gamma$ such that $\gamma^T \Lambda_{\mathrm{R}} \gamma < 2\gamma^T\Omega_{\mathrm{R}}$, the CAM estimator $\hat\eta_{\gamma}^{\mathcal{M}}$ leads to an improvement over the complete-case estimator.

Again, in practice, we attempt to mimic the performance of the optimal estimator. First, ignoring the off-diagonal terms in $\Lambda_{\mathrm{R}}$, we have 
\[
\gamma_{\mathrm{R}}^* \approx \Bigl(\frac{\sigma^2 \nu_{0,m} \mu_{0,m} n_m}{\{\sigma^2 + \tau_m(x^m)\} \nu_{m} (n_0 + n_m)} : m \in \mathcal{M} \Bigr)^T.
\]
For the unknown terms, write $\sigma^2 = \mathbb{E}[\{Y - \eta(x)\}^2 | X = x]$ and $
\sigma_m^2 :=  \sigma^2 + \tau_m(x^m) = \mathbb{E}[\{Y - \eta_m(x^m)\}^2 | X^m = x^m]$. These can be estimated in a natural way using the local constant method. In practice, the estimates $\hat{\sigma}^2$ and $\hat{\sigma}_m^2$ say can be calculated directly by reusing the weights from the regression estimators; cf.~\eqref{eq:weights}.
Finally, this leads to a practical choice of 
\[
\hat{\gamma}_{\mathrm{R}} = \Bigl(\frac{\hat{\sigma}^2 \nu_{0,m} \mu_{0,m} n_m}{ \hat{\sigma}_m^2 \nu_{m}(n_0 + n_m)}: m \in \mathcal{M}\Bigr)^T.
\]

As noted above, the choice of kernel $K_m$ is less straightforward than in the density estimation setting.  Here we can write $\hat{\eta}_0 = \sum_{i \in A_0} Y_i K\bigl(\frac{X_i - x}{h}\bigr)/ \{\sum_{i \in A_0} K\bigl(\frac{X_i - x}{h}\bigr)\}$, and the optimal choice of $\hat{\eta}_{0,m}$ is
\[
\hat{\eta}_{0,m}^* = \mathbb{E}(\hat{\eta}_0 | \mathcal{T}_{A_0, m}) = \sum_{i \in A_0} Y_i \mathbb{E} \Bigl\{\frac{K\bigl(\frac{X_i - x}{h}\bigr)}{\sum_{j \in A_0} K\bigl(\frac{X_j - x}{h}\bigr)} \Big| \mathcal{T}_{A_0, m} \Bigr\}.
\]
Lemma~\ref{lem:KernelReg} in Section~\ref{sec:Kernelchoice} shows that, under certain conditions, this optimal choice is well-approximated by our practical choice $\hat{\eta}_{0,m}$ with kernel $K_m$ depending only on $K$. In particular, if $K$ is the $d$-dimensional Gaussian kernel, then $\hat{\eta}_{0,m}$ can be constructed as in~\eqref{eq:loccon} using the $d_m$-dimensional Gaussian kernel.  We see in our numerical study, that our CAM approach with this practical choice of $\hat{\eta}_{0,m}$ does often lead to an appreciable improvement over the complete-case estimator.

\subsection{Numerical examples} 
\label{sec:sims}
We now demonstrate the CAM density and regression estimators with some numerical examples.  Consider the following models 
\begin{itemize}
	\item[(a)] Density model 1: $X \sim N_2(0, \Sigma)$, where $\Sigma = 0.3 I_2 + 0.7(1_21_2^T)$, where $1_2 := (1,1)^T$. 	
	\item[(b)] Density model 2: $X \sim U(B_1(0))$, where $B_1(0) \subseteq \mathbb{R}^2$ denotes the unit disk centred at 0.
	\item[(c)] Density model 3: $X \sim \frac{1}{4} U([-2,-1]\times[-1/2,1/2]) + \frac{3}{4} U([1,2] \times [-1/2, 1/2])$. 
	\item[(d)] Regression model 1:  Let $X \sim U([0,1]^3)$, and $Y = X^{(1)} + X^{(2)} + 0.1\epsilon$, where $\epsilon \sim N(0,1)$ is independent of $X$.	
	\item[(e)] Regression model 2: Let $X \sim U([0,1]^3)$, and $Y = (X^{(1)} - X^{(2)})^2 + 0.1\epsilon$, where $\epsilon \sim N(0,1)$ is independent of $X$.
  	\item[(f)] Regression model 3: Let $X \sim N_2(0, \Sigma)$, where $\Sigma = 0.3 I_2 + 0.8(1_21_2^T)$ and $Y = \sin(2X^{(1)}) + 0.3\epsilon$, where $\epsilon \sim N(0,1)$ is independent of $X$.
\end{itemize} 
In each case, we generate a training set of size $n \in \{200,500\}$, and then introduce missingness by removing first component of $X$ independently with probability $p_1$. 

The kernel density estimators are computed using the \texttt{ks} package available from CRAN.  In particular, we use the \texttt{kde} function with a Gaussian kernel, and the diagonal bandwidth matrices were chosen using the \texttt{Hpi.diag} function.  In the regression settings, we make use of the \texttt{regpro} package available from \texttt{CRAN}.

To measure the performance, recall that for two density functions $f$ and $g$, say, the \emph{Total Variation} distance between $f$ and $g$ is $\mathrm{TV}(f,g) := \frac{1}{2}\int_{\mathbb{R}^d} |f(x) - g(x)| \, dx$.  We present boxplots of 
\begin{equation}
\label{eq:denperf}
\frac{ \mathrm{TV}(\hat{f}_0, f_X) - \mathrm{TV}(\hat{f}^\mathcal{M}_{\gamma}, f_X)} {\mathrm{TV}(\hat{f}_{0}, f_X)}.
\end{equation}   
In the regression problems, we use the mean integrated squared error: for an estimate $\hat{\eta}$ of $\eta$ that is $\mathrm{MISE}(\hat{\eta}) := \int_{\mathbb{R}^d} \{\hat{\eta}(x) - \eta(x)\}^2 \, dP_X(x)$.  Then we present boxplots of
\begin{equation}
\label{eq:regperf}
\frac{ \mathrm{MISE}(\hat{\eta}_0, \eta) - \mathrm{MISE}(\hat{\eta}^\mathcal{M}_{\gamma}, \eta)} {\mathrm{MISE}(\hat{\eta}_{0}, \eta)}.
\end{equation} 

  \begin{figure}[ht!]
 	\centering
 	\includegraphics[width=0.49\textwidth]{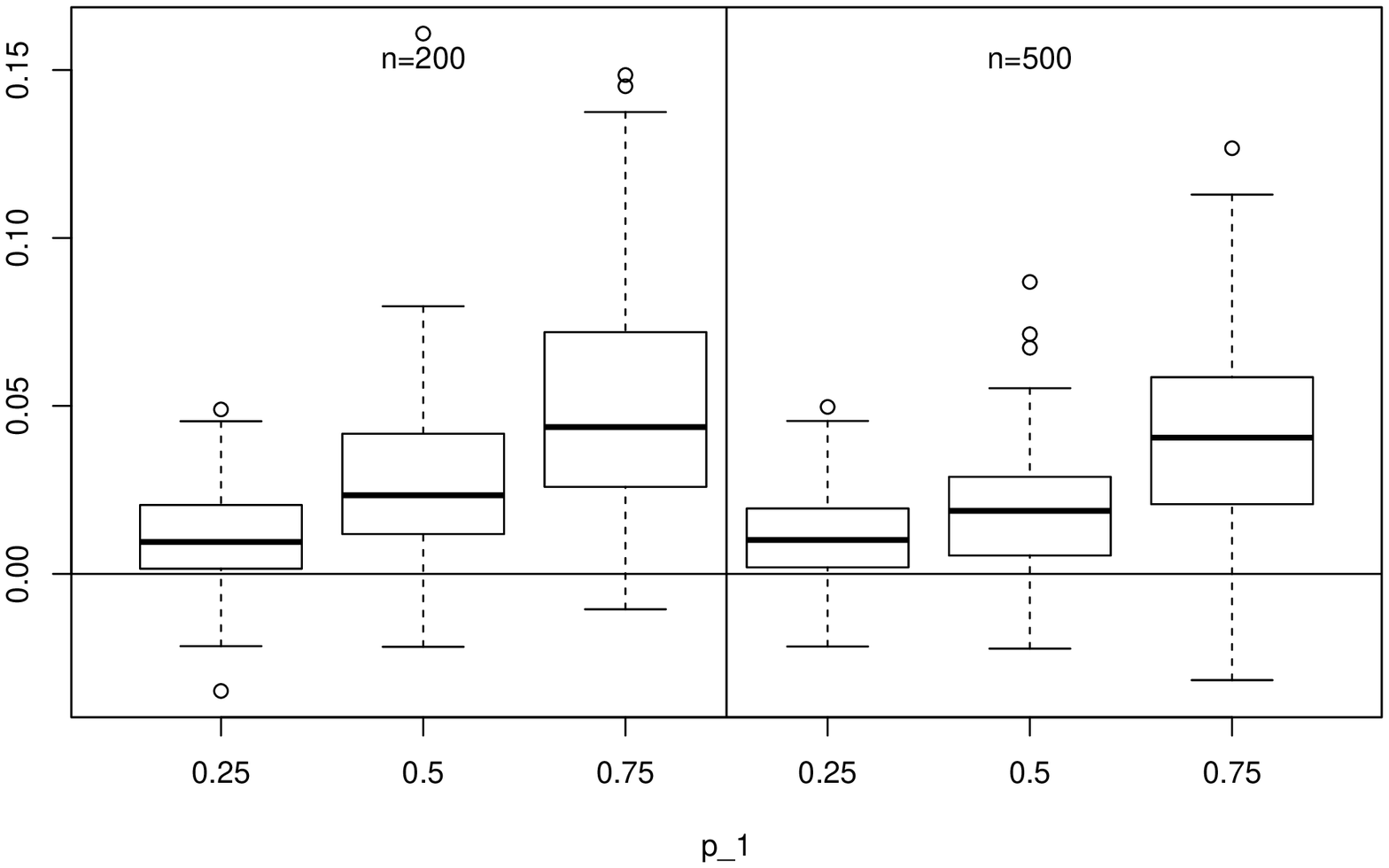}  	\includegraphics[width=0.49\textwidth]{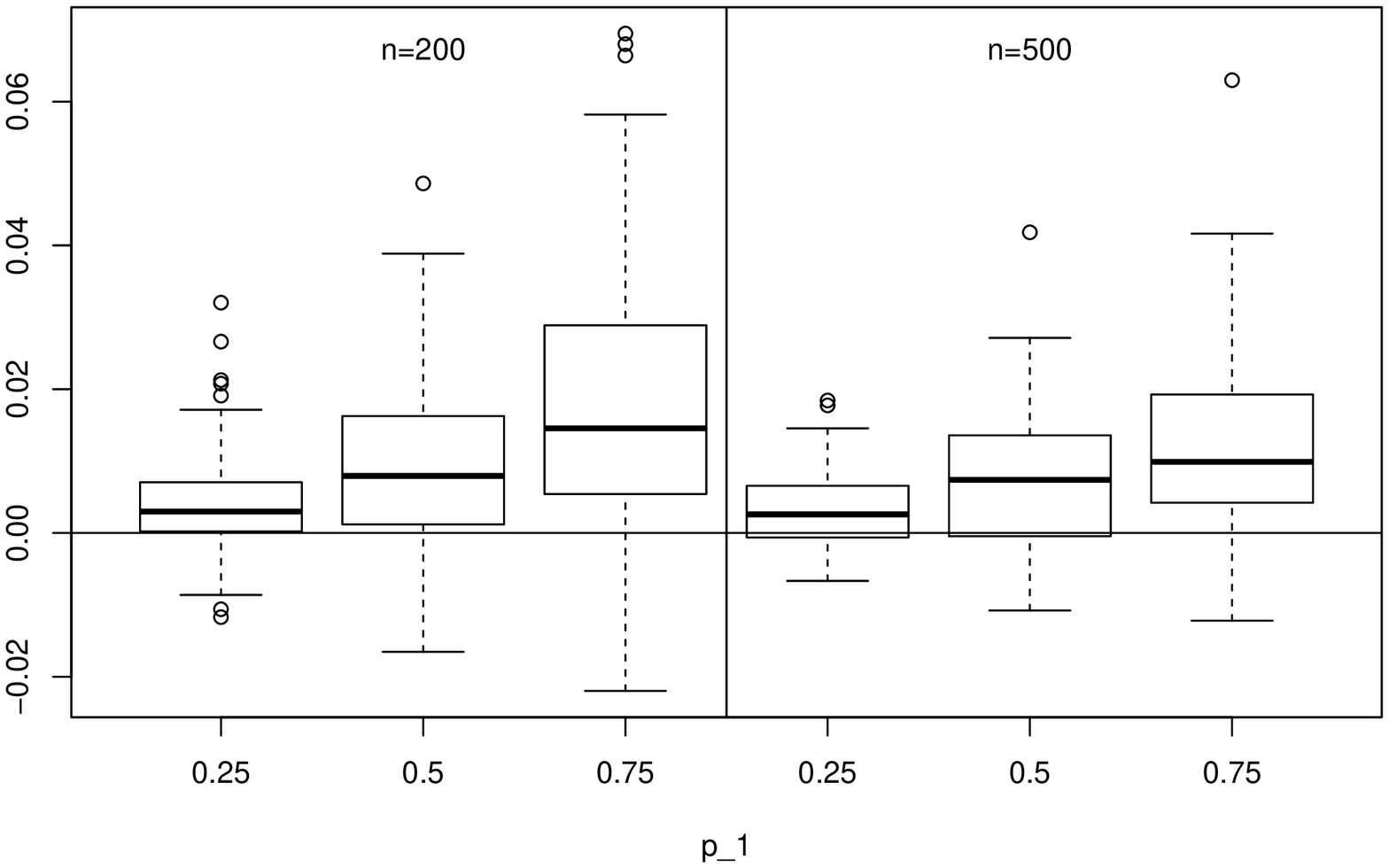}
	 	\includegraphics[width=0.49\textwidth]{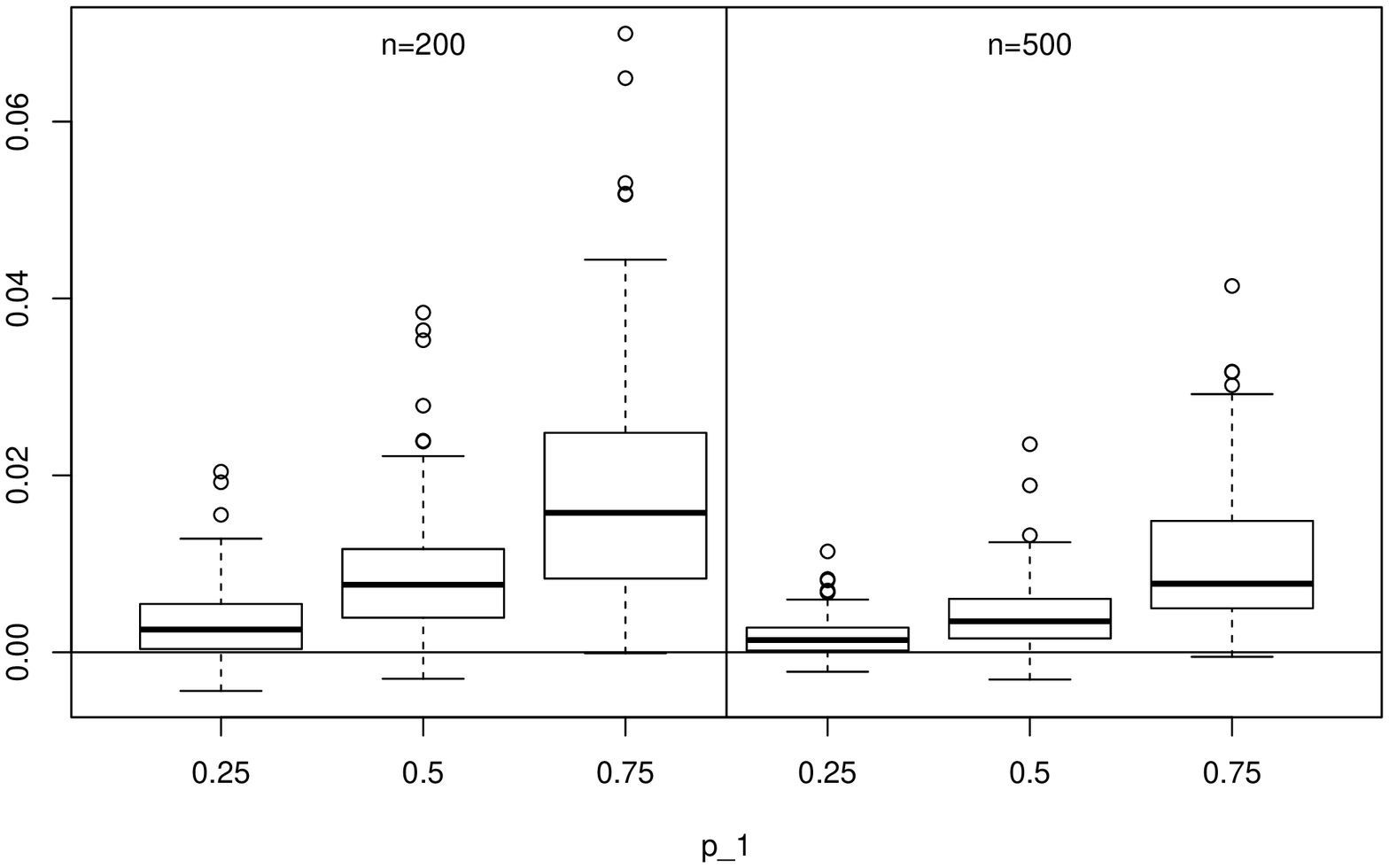} 	
	\caption{Boxplot of the relative performance of the CC and CAM kernel density estimators given by \eqref{eq:denperf} for 100 repetitions of the experiment for density model 1 (top left), density model 2 (top right) and density model 3 (bottom). }
 	\label{fig:Density1}
 \end{figure}
  
\begin{figure}[ht!]
 	\centering
 	\includegraphics[width=0.49\textwidth]{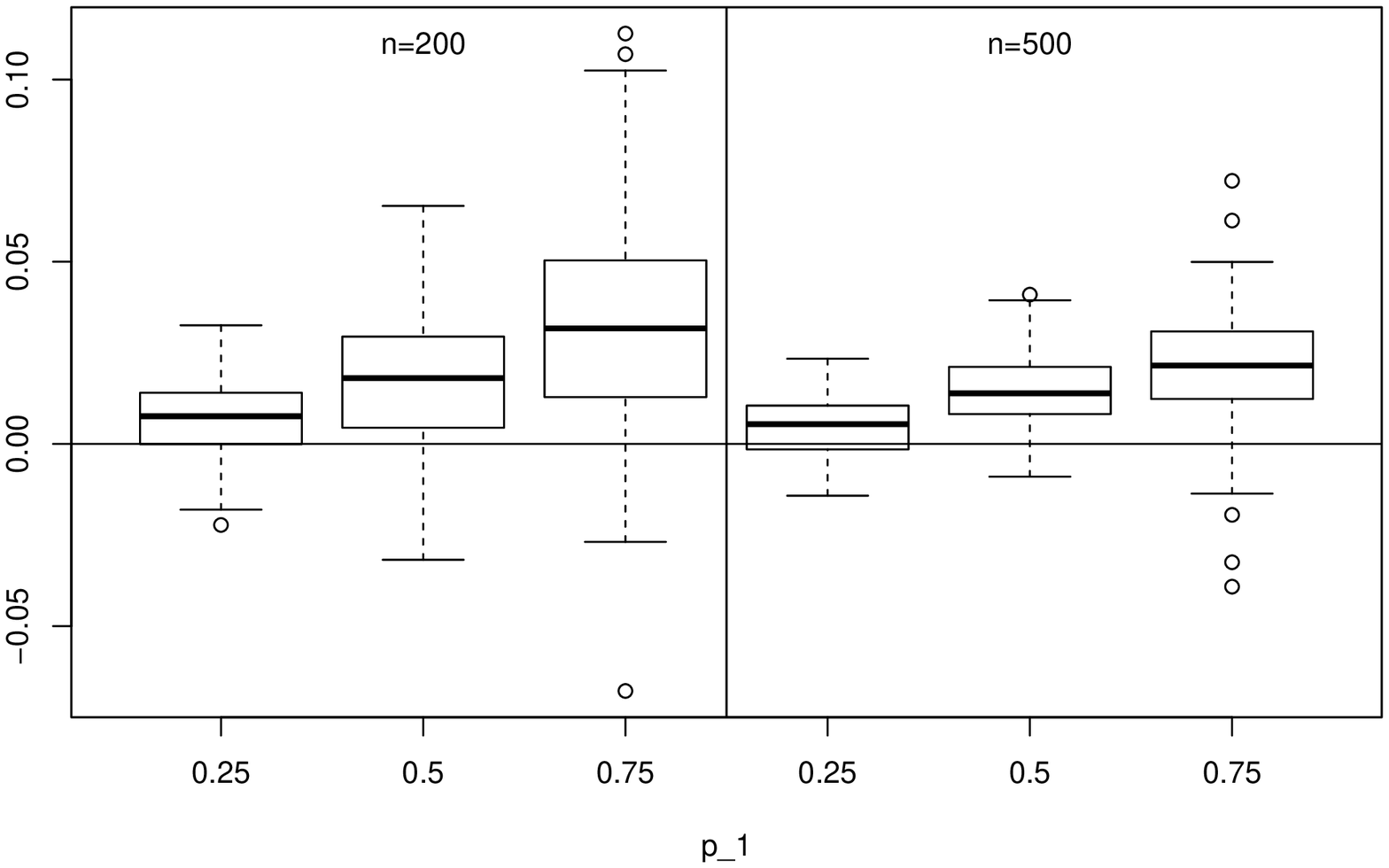} 	 \includegraphics[width=0.49\textwidth]{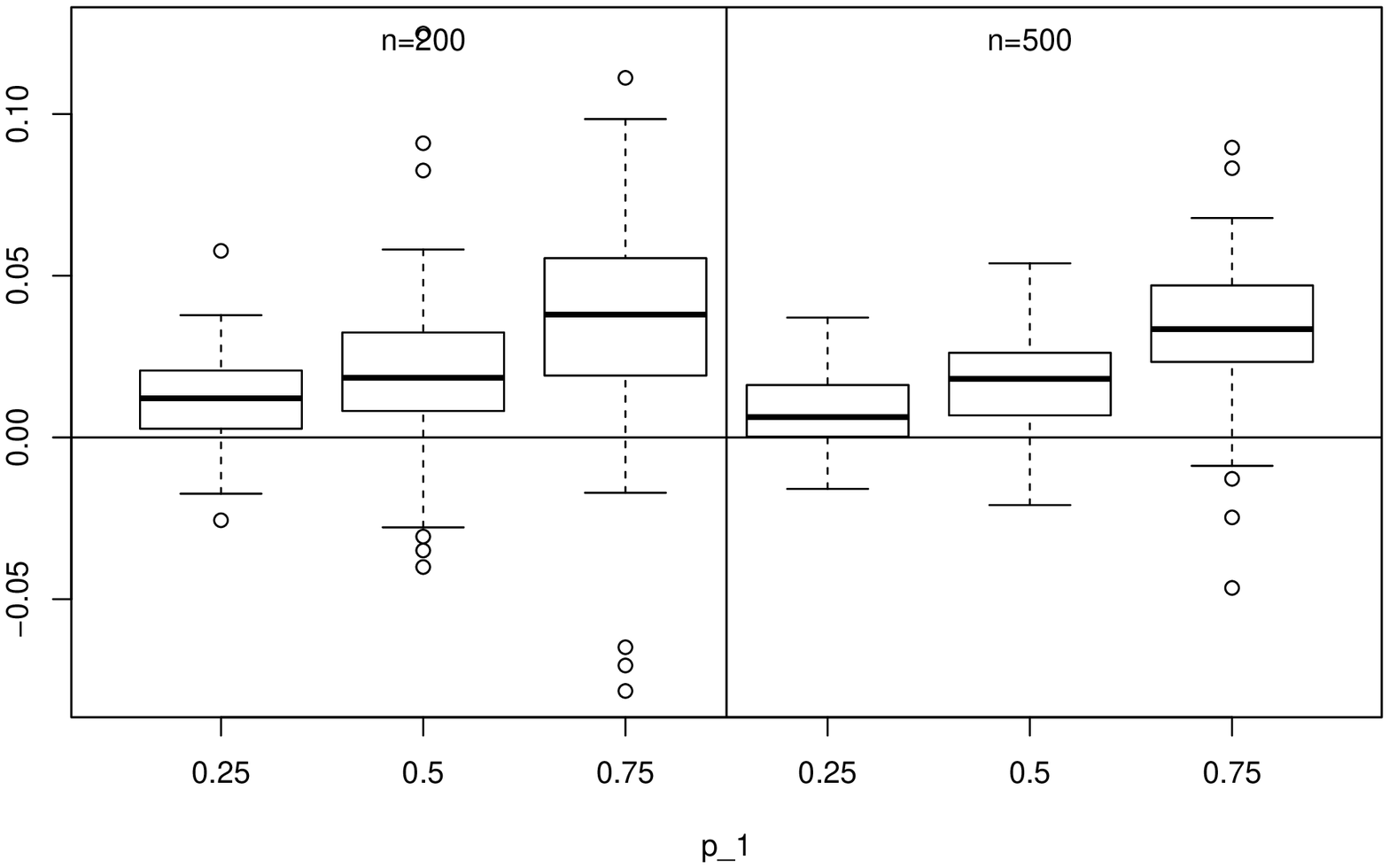}  	
 	\includegraphics[width=0.49\textwidth]{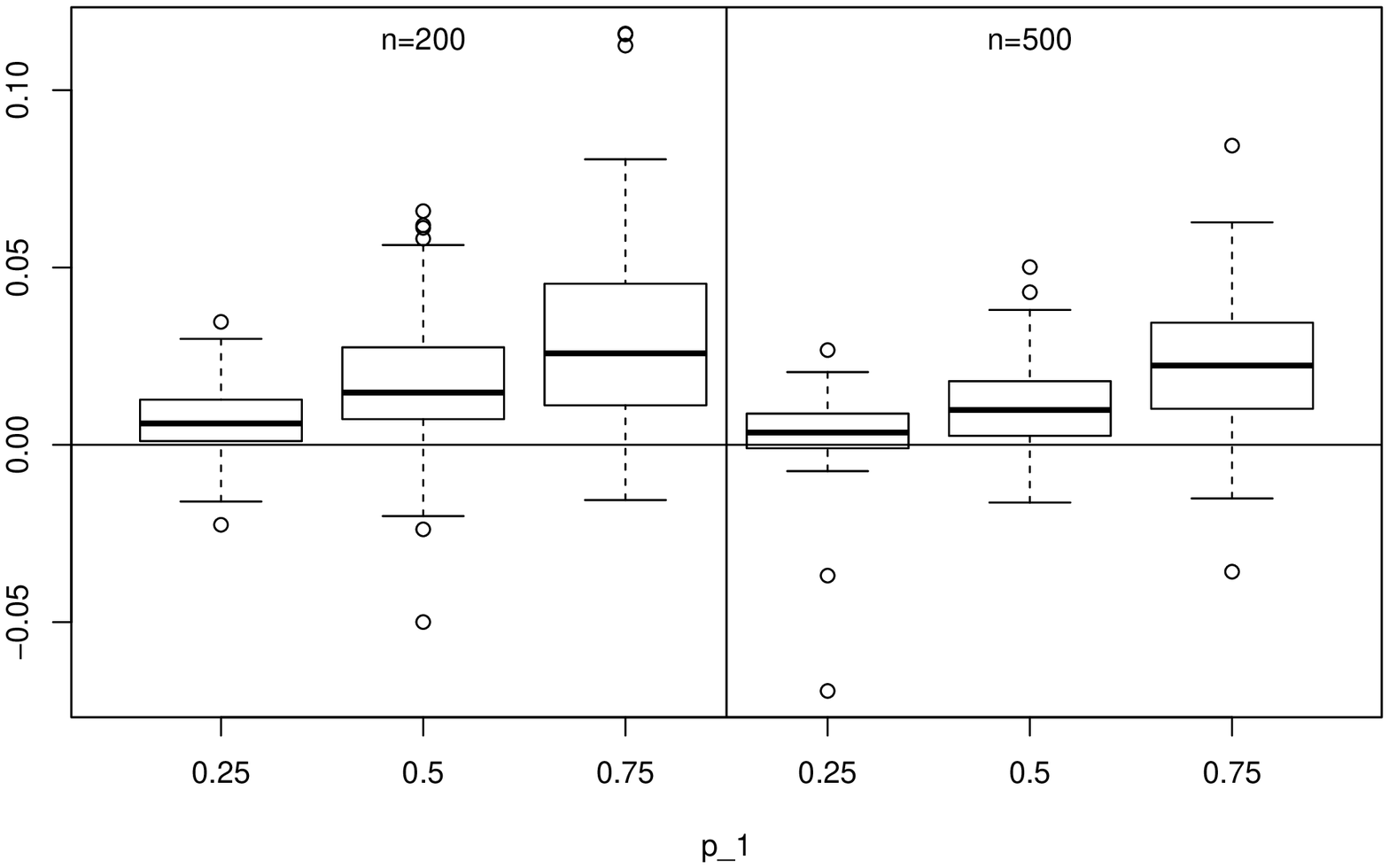} 
	\caption{Boxplots of the relative performance of the CC and CAM local constant regression estimators given by \eqref{eq:regperf} for 100 repetitions of the experiment for  regression problem 1 (top left), regression problem 2 (top right) and regression problem 3 (bottom). }
 	\label{fig:Regression1}
 \end{figure}
 
We see in Figure~\ref{fig:Density1} that the CAM estimator outperforms the CC estimator in terms of total variation distance.  As expected the CAM estimator leads to a greater reduction in error as the missingness probability $p_1$ increases.  On the other hand, as $n$ increases, the relative reduction in error is smaller.  Finally, note that there are a small number of repeats of the experiment in each case where the difference in mean squared error is negative -- this is likely due to the fact that we need to estimate the $\gamma$ using the observed data, there is also a small Monte Carlo error.  

Figure~\ref{fig:Regression1} shows that the CAM local constant regression estimator outperforms the CC estimator. Similarly to the density estimation problems, the CAM estimator leads to a greater reduction in error as the missingness probability $p_1$ increases.  Here we see further that for larger sample sizes, the improvement relative to using the full data set appears to decrease slightly -- this is in agreement with our Theorem~\ref{thm:RegressionPi}, since we only expect a second order improvement in asymptotic mean squared error. 

Finally, as discussed the introduction, a popular approach to overcome issues with missing data is to impute the missing values and treat the result as a full dataset.  We now compare the CAM approach to density estimation and regression with a number of state of the art imputation methods.  These include \emph{mean imputation}, where the missing entries for each feature are imputed with the sample mean (of the observed values) for that feature \citep[Section 4.2.1]{LittleRubin:2002}; \emph{Predictive mean matching (PMM) imputation}, here the missing values are replaced by a value sampled from the observed data, where the sample is taken from the observations that are close to the observation with the missing value \citep[Section 4.3.2]{LittleRubin:2002}; and \emph{Random Forest (RF) imputation}, this uses the Random Forests algorithm to predict the missing entries based on the observed data \citep{Breiman:2002,PantanowitzMarwala:2009}.  For the latter two methods, we use multiple imputation; the missing values are imputed five times and we then take the average of the results after fitting the model on the five imputed datasets.  A detailed outline of these methods can be found in \citet[Chapters 4 and 10]{LittleRubin:2002}. Our implementation utilises the \texttt{mice} R package available from CRAN \citep{MICE:2018}.  

In our experiments, the kernel density estimators are computed using the \texttt{ks} package available from CRAN.  In particular, we use the \texttt{kde} function with a Gaussian kernel, and the diagonal bandwidth matrices were chosen using the \texttt{Hpi.diag} function.  In the regression settings, we make use of the \texttt{regpro} package available from \texttt{CRAN}. For the imputation approaches, the corresponding kernel methods are applied to the imputed datasets.

\begin{table}[ht!]
	\centering
	\caption{\label{tab:one}  Average estimated total variation errors ($*10^3$) for the density estimation problems and mean integrated squared error ($*10^3$) for the regression problems. In each case, $n = 500$ and the first component of $X$ is missing with probability 1/2.}
		\begin{tabular}{c | r | r | r | r r r r}
Model &Full & CAM & CC & Mean & PMM & RF \\
\hline
Density\\ 
\hline
1 & 4.97$_{ 0.02 }$ & 6.07$_{ 0.03 }$ & 6.23$_{ 0.03 }$ &  23.77$_{ 0.24 }$ & 6.76$_{ 0.10 }$ & 6.86$_{ 0.10 }$\\
2  & 35.31$_{ 0.07 }$ & 40.28$_{ 0.08 }$ & 40.58$_{ 0.08 }$ &  101.80$_{ 1.04 }$ & 50.21$_{ 1.07 }$ & 38.31$_{ 0.31 }$\\
3  & 51.84$_{ 0.10 }$ & 52.49$_{ 0.10 }$ & 52.72$_{ 0.10 }$ &  69.66$_{ 0.72 }$ & 52.91$_{ 0.34 }$ & 52.85$_{ 0.35 }$\\
\hline
Regression\\
\hline
1  & 2.74$_{ 0.01 }$ & 4.29$_{ 0.02 }$ & 4.35$_{ 0.02 }$ &  6.49$_{ 0.10 }$ & 2.65$_{ 0.04 }$ & 6.40$_{ 0.11 }$\\
2 & 2.86$_{ 0.02 }$ & 4.12$_{ 0.03 }$ & 4.20$_{ 0.03 }$ &  5.82$_{ 0.11 }$ & 12.28$_{ 0.23 }$ & 6.02$_{ 0.12 }$\\
3 & 19.29$_{ 0.11 }$ & 29.51$_{ 0.19 }$ & 29.84$_{ 0.20 }$ &  61.90$_{ 1.01 }$ & 112.22$_{ 1.63 }$ & 24.44$_{ 0.46 }$\\
\end{tabular}
\end{table}

In Table~\ref{tab:one} we present the estimated total variation or mean integrated squared errors (with standard errors in subscript) over 100 repetitions of each experiment.   For comparison, we also present the results of applying the kernel methods to the full dataset of size $n$.  We see that the CAM technique improves on the complete-case approach in every setting.   On the other hand, as discussed in the introduction, the imputation approaches are not always successful: indeed, there are many settings  here where the imputation methods lead to worse performance than the complete-case method.   Notice that occasionally the imputation methods, eg. random forests in Density Model 2 and predictive mean matching in Regression Model 1, perform very well (even outperforming the full data estimator in a few cases) -- it is unclear why this is the case due to the black-box nature of these approaches.

\subsection{The Brandsma school dataset} 
In this subsection, we show how the CAM local constant regression estimator can be used in practice with the \texttt{brandsma} dataset available in the MICE package.   The full data set consists of 4106 observations of 14 features. We simplify the problem by retaining only 4 features, namely the ``verbal IQ score", the ``SES score", ``language score pre", and ``language score post".  Suppose that we are interested in predicting the verbal IQ score, $Y$, from the remaining features, i.e. $X$ is the 3-dimensional vector consisting of ``SES score", ``language score pre", and ``language score post".  We remove 17 observations for which the response is missing.  In the resulting dataset, we have 3464 complete-cases, 302 with $m = m_1 = (0,1,0)^T$, 182 with $m = m_2 = (0,0,1)^T$, and 108 with $m = m_3 = (1,0,0)^T$.  There are a few observations for other values of $m$, but the corresponding sample sizes are all less than 20 and are therefore ignored.  

In order to evaluate the performance of the CAM estimator, we take a subsample of size 1000 from the complete-cases to use as a test set (this is fixed throughout).  We carry out 100 experiments. In each one, we form a training set by taking another sample of size 200 from the remaining 2464 complete-cases (this sample is different in each experiment). The 200 chosen complete-cases are then combined with the observations in $A_{m_1}$, $A_{m_2}$ and $A_{m_3}$ (which are the same in every experiment). Thus, in each experiment, we have $n_0 = 200$, $n_{m_1} = 302$, $n_{m_2} = 182$, and $n_{m_3} = 108$.  

\begin{figure}[ht!]
	\centering
	\includegraphics[width=0.5\textwidth]{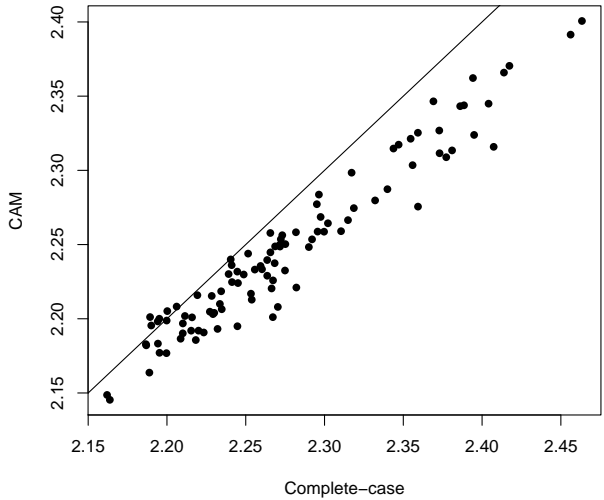} 
	\caption{The average predictive MSE on the test set for the complete-case and CAM estimators for the Brandsma data application. The straight line is ``$y = x$".}
	\label{fig:brandsma}
\end{figure}

In Figure~\ref{fig:brandsma} we plot the average (over the test set) of $\{\hat{\eta}(X) - Y\}^2$ for the complete-case and CAM estimators in each of the 100 experiments.  We use a Gaussian kernel and the bandwidth was chosen using leave-one-out cross-validation.  We see that the CAM estimator has a lower predictive MSE than the complete-case in 94$\%$ of the cases.  In the remaining 6\% of cases, the performance of the two estimators is similar and the discrepancy can be explained by the small sample size used to estimate the MSE.   The overall average MSE (with standard errors) of the complete-case and CAM estimators were 2.28 (0.007) and 2.25 (0.006), respectively, the average improvement in MSE is 0.03 (0.002), whereas the maximum improvement over the 100 experiments was 0.09.

\section{Discussion} 
\label{sec:discuss}
We have seen that our proposed CAM estimator can be used to improve the complete-case estimator in a wide range of statistical problems.   We conclude our paper with a discussion of few extensions of our method. 

First, a key aspect of our CAM proposal is to construct $\hat{\theta}_{0,\mathcal{M}}$ and $\hat{\theta}_{\mathcal{M}}$, so that $\hat{\theta}_{0,\mathcal{M}} - \hat{\theta}_{\mathcal{M}}$ is centered.   One way to guarantee this is to use a resampling method.  Let $n_{0,m} := \min\{n_0, n_m\}$.  Consider the sets of all subsamples (without replacement) of $n_{0,m}$ observations from the data in $A_0$ and $A_m$,  respectively, denoted by $A_0^1, \ldots, A_0^{B_0}$ and $A_m^1, \ldots, A_m^{B_m}$, where $B_0 = {n_0 \choose n_{0,m} }$ and $B_m = {n_m \choose n_{0,m} }$. Notice that at least one of $B_0$ and $B_m$ will be equal to one.   Then, for $m \in \mathcal{M}$, consider the two (independent) statistics  $\bar{\theta}_{0,m} := \frac{1}{B_0} \sum_{b=1}^{B_0} \hat{\theta}_{A_0^b,m}$ and  $\bar{\theta}_{m} := \frac{1}{B_m} \sum_{b=1}^{B_m} \hat{\theta}_{A_m^b,m}$. At least in the missing completely at random setting, we have $\mathbb{E}(\bar{\theta}_{0,m}) = \mathbb{E}(\bar{\theta}_{m})$ -- each of the terms in the sums are estimators calculated on datasets of the same size.   Let $\bar{\theta}_{0,\mathcal{M}} = (\bar{\theta}_{0,m} : m \in \mathcal{M})^T$, and $\bar{\theta}_{ \mathcal{M}} = (\bar{\theta}_m : m \in \mathcal{M})^T$.  Then, for $\gamma \in \mathbb{R}^{|\mathcal{M}|}$, define 
\[
\bar{\theta}^{\mathcal{M}}_{\gamma} :=  \hat{\theta}_{0} - \gamma^T \bigl(\bar{\theta}_{0,\mathcal{M}} - \bar{\theta}_{\mathcal{M}}\bigr).
\]
We have the following corollary to Proposition~\ref{prop:MSE}, which holds for any estimation method and requires no assumptions on the distribution $P$.   The only restriction is that the data is MCAR.  
\begin{corollary} 
\label{cor:msered}
Suppose the data is missing completely at random.  Then
\begin{equation}
\label{eq:subsample}
\mathrm{MSE}(\bar{\theta}^\mathcal{M}_{\gamma}) - \mathrm{MSE}(\hat{\theta}_{0}) = \gamma^T\var(\bar{\theta}_{0,\mathcal{M}} - \bar{\theta}_{\mathcal{M}})\gamma - 2\gamma^T \mathrm{Cov}(\hat{\theta}_{0}, \bar{\theta}_{0,\mathcal{M}}).
\end{equation}
\end{corollary}

Another consideration is the choice of the set $\mathcal{M}$.  This choice is primarily driven by the data -- in the first instance we might let $\mathcal{M} = \mathcal{M}^* := \{m \in \{0,1\}^d\setminus \{0_d \} : |A_m| > 0 \}$.  In some cases, however, we may consider using a different set $\mathcal{M}$.   For instance, for some $m$, the corresponding sample size $n_m$ may be non-zero but small.   In our numerical analysis in Section~\ref{sec:sims}, we drop $m$ if $|A_m|$ is less than 20.   Another potential option here is to use a data integration method.  More specifically, let $\bar{A}_m := \{i \in A_0^c: m_i \leq m\}$, where the partial order on $\{0,1\}^d$ is defined by $m_i \leq m$ if $\{j : m_i^j = 1\} \subseteq \{j : m^j = 1\}$.  Then  $\bar{A}_m$ is the largest set in $A_0^c$ which has complete observations for variables in $\{j \in \{1, \ldots d\} : m^j=0\}$. Notice also that by construction $\bar{A}_m$ and $A_0$ are disjoint.  To include the data integration in our CAM method, we can simply replace $A_m$ with $\bar{A}_m$, for $m \in \mathcal{M}$.  There are other options too. Suppose we have $m_1 \neq m_2 \in \{0,1\}^d\setminus \{0_d\}$, with $|A_{m_1}| > 0$, $|A_{m_2}| = 0$ and $m_1 \leq m_2$.  Then rather than using $A_{m_1}$, we may opt to use $\bar{A}_{m_2} \supseteq A_{m_1}$ instead.

Finally we discuss the challenging non-MCAR settings.  When the data is MAR or MNAR there is an additional challenge.   In this case, the  naive complete-case estimator will potentially be asymptotically biased.  Conditionally on $M_1, \ldots, M_n$, the data in $\mathcal{T}_{A_0,0}$ can be interpreted as independent and identically distributed pairs from $Q_0$, that is the joint distribution of a generic pair $(X,Y) | \{M = 0\}$.   Thus, we expect the complete-case estimator $\hat{\theta}_{0}$ to be close to $\theta(Q_0)$ as opposed to $\theta(P_0)$.  Two natural questions arise: (i) under what conditions do we have $\theta(P_0) = \theta(Q_0)$ and $\mathbb{E}(\hat{\theta}_{0,m} - \hat{\theta}_m) \approx 0$ (cf.~Proposition~\ref{prop:MSE})?; and (ii) if the conditions in (i) do not hold, how can we adapt the CAM estimator.

A partial answer to (i) is provided by the following in the regression setting:  Suppose we observe $n$ independent and identically distributed copies of $Z^M = (X^{M}, Y)$ and are interested in estimating $\eta(x) = \mathbb{E}(Y| X = x)$.  Assume that $M$ is independent of $Y$ given $X^M$. (Note that this is slightly different to the missing at random condition, which assumes that $M$ is independent of $Z$ given $Z^M$.)  In this case, we have that 
$\theta(Q_0) = \mathbb{E}(Y | X = x, M = 0) = \mathbb{E}(Y | X = x) = \theta(P_0)$ and
\[
\theta_m(Q_m) = \mathbb{E}(Y | X^m = x^m, M = m) = \mathbb{E}(Y | X^m = x^m) = \mathbb{E}(Y | X^m = x^m, M = 0) = \theta_m(Q_0).
\]
Thus, we can still expect that the complete-case approach will target $\eta(x)$, and moreover, can hope that $B_m = \mathbb{E}(\hat{\theta}_{0,m} - \hat{\theta}_0) \approx \theta_m(Q_0)-  \theta_m(Q_m) = 0$. 

Relating to the problem in (ii), in the missing at random case, there are many methods that aim to correct the bias of the naive complete-case estimator; see, for instance, \citet[Chapter~3.3]{LittleRubin:2002}.  One approach is to weight the observations according to the probability that they are (non)missing.   For concreteness, we focus on one such idea, which advocates reweighting the observations according to their (inverse) propensity score -- see \citet[Chapter~3.7]{LittleRubin:2002}.  Recall that $p_m(z) = p_{m}(x,y) = \mathbb{P}(M = m | X = x, Y = y)$.  Of course, $p_{m}(x,y)$ is typically unknown, and needs to be estimated; in fact, we can only hope to estimate $p_{m}(x,y)$ from the observed data in the missing at random setting.  

Consider the $U$-Statistics setting, where we are interested in estimating $\theta = \theta(P) = \mathbb{E}\{\phi(Z_1, \ldots,Z_r)\}$.  The complete-case approach we will in fact construct an unbiased estimator for $\theta(Q) = \mathbb{E}\bigl\{\phi(Z_1, \ldots,Z_r) \big| M_1 = \ldots =  M_r = 0 \bigr\} \neq \theta(P)$.  One solution here is to consider the Horvitz-Thompson estimators \citep{HorvitzThompson:1956} in place of $\hat{\theta}_0$, $\hat{\theta}_{0,m}$ and $\hat{\theta}_m$. In this problem, the complete-case analogue is
\[
\tilde{\theta}_0 = \tilde{\theta}_{A_0,0} =\frac{1}{ \sum_{\{i_1, \ldots, i_r\} \subseteq A_0} \frac{1}{\prod_{j=1}^r p_0(Z_{i_j})}} \sum_{\{i_1, \ldots, i_r\} \subseteq A_0} \frac{\phi(Z_{i_1}, \ldots,Z_{i_r})}{\prod_{j=1}^r p_0(Z_{i_j})}. 
\]
Moreover, similar expressions can be derived for $\tilde{\theta}_{0,m}$ and $\tilde{\theta}_{m}$.  The CAM estimator can then be constructed as in the MCAR case.  Of course, in practice,  $p_0(z)$ and $p_m(z^m)$ need to be estimated using the observed data. This is non-trivial, and further study in this direction is left for future work. 

\section*{Acknowledgements}
The work of the second author was partially supported by an NIH award 1R01GM131407. 

\spacingset{1}

\newpage

\spacingset{1.1} 

\appendix

\section{Appendix: Technical arguments} 
\label{sec:tech}

\subsection{Proofs for Section \ref{sec:toy}}
\label{sec:techprelim}
\begin{proof}[Proof of claims in Section \ref{sec:toy}]
	\emph{Claim 1}: The statistic $T = (T_1, T_2)^T = (\hat{\nu}_{X,1} - \hat{\nu}_{Y,1}\Gamma_{12}/\Gamma_{22}, \hat{\nu}_{Y,1} + \hat{\nu}_{Y,2})^T$ is sufficient for $\nu = (\nu_X, \nu_Y)^T$.  To see this we use the factorisation criteria: let $Z_i = (X_i, Y_i)^T$, the full likelihood is 
	\begin{align*}
	L(\nu) & = \prod_{i = 1}^n \frac{1}{\sqrt{2\pi |\Gamma|}} \exp\Bigl(-\frac{1}{2}(Z_i - \nu)^T \Gamma^{-1}(Z_i - \nu)\Bigr) \prod_{i = n+1}^{2n} \frac{1}{\sqrt{2\pi \Gamma_{22}}} \exp\Bigl(-\frac{1}{2\Gamma_{22}}(Y_i - \nu_Y)^2\Bigr)
	\\ & = \frac{1}{(2\pi)^n |\Gamma|^{n/2}\Gamma_{22}^{n/2}}  \exp\Bigl(-\frac{1}{2}\sum_{i=1}^n (Z_i - \nu)^T \Gamma^{-1}(Z_i - \nu) -\frac{1}{2\Gamma_{22}}\sum_{i = n+1}^{2n} (Y_i - \nu_Y)^2\Bigr)
	\\ & \propto  \exp\Bigl(\sum_{i=1}^n \nu^T\Gamma^{-1}Z_i - \frac{1}{2}\nu^T\Gamma^{-1}\nu + \frac{1}{\Gamma_{22}} \sum_{i = n+1}^{2n} Y_i\nu_Y - \frac{nv_Y^2}{2\Gamma_{22}} \Bigr)
	\\ & =  \exp\Bigl(n \nu^T \Gamma^{-1} (\hat{\nu}_{X,1}, \hat{\nu}_{Y,1})^T - \frac{1}{2}\nu^T\Gamma^{-1}\nu + \frac{1}{\Gamma_{22}} n \hat{\nu}_{Y,2} \nu_Y -  \frac{nv_Y^2}{2\Gamma_{22}} \Bigr).
	\end{align*}
	Now write 
	\[
	\nu^T \Gamma^{-1} (\hat{\nu}_{X,1}, \hat{\nu}_{Y,1})^T = \frac{1}{|\Gamma|} \Bigl\{(\Gamma_{22}\hat{\nu}_{X,1} - \Gamma_{12} \hat{\nu}_{Y,1}) \nu_X + (\Gamma_{11}\hat{\nu}_{Y,1} - \Gamma_{12} \hat{\nu}_{X,1})\nu_Y\Bigr\}. 
	\]
	By collecting the coefficients of $\nu_X$ and $\nu_Y$, it follows that 
	\[ 
	\Bigl(\hat{\nu}_{X,1} - \hat{\nu}_{Y,1} \Gamma_{12}/\Gamma_{22}, \hat{\nu}_{X,1} - \hat{\nu}_{Y,1} \Gamma_{11}/\Gamma_{12} - \hat{\nu}_{Y,2} \Gamma_{11}/\Gamma_{12} + \hat{\nu}_{Y,2} \Gamma_{12}/{\Gamma_{22}}\Bigr)^T = \Bigl(T_1, T_1 + \bigl(\frac{\Gamma_{12}}{\Gamma_{22}} - \frac{\Gamma_{11}}{\Gamma_{12}}\bigr) T_2\Bigr)^T
	\]
	is a sufficient statistic for $\nu$. Thus Claim 1 is true. 
	
	\bigskip 
	
	\emph{Claim 2:} We have that $\tilde{\nu}_{X} = \mathbb{E}(\hat{\nu}_{X,1} | T )$.  To prove this, first observe that
		\[
		(\hat{\nu}_{X,1}, \hat{\nu}_{Y,1}, \hat{\nu}_{Y,2})^T \sim N_3((\nu_X,\nu_Y,\nu_Y)^T, \Gamma'),
		\]
		where 
		\[
		\Gamma' = \left( \begin{array}{c c} \Gamma & 0 \\ 0 & \Gamma_{22} \end{array}\right).
		\]
		Therefore
	\[
(\hat{\nu}_{X,1}, T_1, T_2)^T \sim N_3((\nu_X,\nu_X - \nu_Y\Gamma_{12}/\Gamma_{22}, 2\nu_Y)^T, \Gamma''),
    \]
    	where 
    	\[
    	\Gamma'' = \left( \begin{array}{c c c} \Gamma_{11} & \Gamma_{11} - \Gamma_{12}^2/\Gamma_{22} &\Gamma_{12} \\ \Gamma_{11} - \Gamma_{12}^2/\Gamma_{22} & \Gamma_{11} - \Gamma_{12}^2/\Gamma_{22} & 0 \\ \Gamma_{12} & 0 & 2\Gamma_{22} \end{array}\right).
    	\]
	By standard Gaussian distribution theory, it follows that
	\begin{align*}
	& \mathbb{E}\{\hat{\nu}_{X,1} | (T_1, T_2) \}  \\ 
	& = \nu_X + (\Gamma_{11} - \Gamma_{12}^2/\Gamma_{22}, \Gamma_{12}) \left( \begin{array}{c c} \Gamma_{11} - \Gamma_{12}^2/\Gamma_{22} & 0 \\ 0 & 2\Gamma_{22} \end{array}\right)^{-1} \left( \begin{array}{c} T_1 - \nu_X + \nu_Y\Gamma_{12}/\Gamma_{22} \\  T_2 - 2\nu_Y \end{array}\right)
		\\ & = T_1 - T_2\Gamma_{12}/(2\Gamma_{22}) = \tilde{\nu}_X,  
 \end{align*}
 which completes the proof. 
\end{proof}

\subsection{Proofs of the results in Section~\ref{sec:setting}}
\begin{proof}[Proof of Proposition~\ref{prop:MSE}]
		First, we have that 
		\begin{align*}
& \mathbb{E}\{(\hat{\theta}_{\gamma}^{ \mathcal{M}}  - \theta)^2 - (\hat{\theta}_{0}  - \theta)^2\} 
\\  & \hspace{60pt} =  \mathbb{E}\{(\hat{\theta}_{\gamma}^{\mathcal{M}} - \hat{\theta}_{0})^2 + 2\{\hat{\theta}_{\gamma}^{\mathcal{M}} - \hat{\theta}_{0})(\hat{\theta}_{0}  - \theta)\}
		\\ &  \hspace{60pt}=  \gamma^T\mathbb{E}\{(\hat{\theta}_{0,\mathcal{M}} - \hat{\theta}_{\mathcal{M}}) (\hat{\theta}_{0,\mathcal{M}} - \hat{\theta}_{\mathcal{M}})^T\} \gamma- 2\gamma^T\mathbb{E}\{(\hat{\theta}_{0,\mathcal{M}} - \hat{\theta}_{\mathcal{M}}) (\hat{\theta}_{0}  - \theta)\}.
		\end{align*} 
Then, using the fact that $\hat{\theta}_{0,\mathcal{M}}$ and $\hat{\theta}_{\mathcal{M}}$ are independent, write
\begin{align*}
& \mathbb{E}\{(\hat{\theta}_{0,\mathcal{M}} - \hat{\theta}_{\mathcal{M}}) (\hat{\theta}_{0,\mathcal{M}} - \hat{\theta}_{\mathcal{M}})^T\} 
\\ &  \hspace{60pt} = \mathbb{E}(\hat{\theta}_{0,\mathcal{M}}\hat{\theta}_{0,\mathcal{M}}^T ) -\mathbb{E}(\hat{\theta}_{0,\mathcal{M}}\hat{\theta}_{\mathcal{M}}^T)  - \mathbb{E}(\hat{\theta}_{\mathcal{M}}\hat{\theta}_{0,\mathcal{M}}^T) + \mathbb{E}(\hat{\theta}_{\mathcal{M}}\hat{\theta}_{\mathcal{M}}^T) 
\\ &  \hspace{60pt} = \var(\hat{\theta}_{0,\mathcal{M}})  + \var(\hat{\theta}_{\mathcal{M}}) +  \{\mathbb{E}(\hat{\theta}_{0,\mathcal{M}}) - \mathbb{E}(\hat{\theta}_{\mathcal{M}})\}\{\mathbb{E}(\hat{\theta}_{0,\mathcal{M}}) - \mathbb{E}(\hat{\theta}_{\mathcal{M}})\}^T 
\\ &  \hspace{60pt} =   \mathrm{Var}(\hat{\theta}_{0,\mathcal{M}}-\hat{\theta}_{\mathcal{M}}) +  \{\mathbb{E}(\hat{\theta}_{0,\mathcal{M}}) - \mathbb{E}(\hat{\theta}_{\mathcal{M}})\}\{\mathbb{E}(\hat{\theta}_{0,\mathcal{M}}) - \mathbb{E}(\hat{\theta}_{\mathcal{M}})\}^T.
\end{align*}
Moreover, since $\hat{\theta}_{0}$ and $\hat{\theta}_{\mathcal{M}}$ are independent, we have
\begin{align*}
\mathbb{E}\{(\hat{\theta}_{0,\mathcal{M}} - \hat{\theta}_{\mathcal{M}}) (\hat{\theta}_{0}  - \theta)\} & = \mathbb{E}[(\hat{\theta}_{0,\mathcal{M}} - \hat{\theta}_{\mathcal{M}}) \{\hat{\theta}_{0} - \mathbb{E}(\hat{\theta}_{0}) + \mathbb{E}(\hat{\theta}_{0}) - \theta\}]
\\ & =  \mathbb{E}[(\hat{\theta}_{0,\mathcal{M}} - \hat{\theta}_{\mathcal{M}}) \{\hat{\theta}_{0} - \mathbb{E}(\hat{\theta}_{0})\}] + \mathbb{E}(\hat{\theta}_{0,\mathcal{M}} - \hat{\theta}_{\mathcal{M}}) \mathbb{E}(\hat{\theta}_{0} - \theta)
\\ & =  \mathrm{Cov}(\hat{\theta}_{0,\mathcal{M}},\hat{\theta}_{0}) + \mathbb{E}(\hat{\theta}_{0,\mathcal{M}} - \hat{\theta}_{\mathcal{M}}) \mathbb{E}(\hat{\theta}_{0} - \theta).
\end{align*}
The result follows.
\end{proof} 

\subsection{Proofs of the results in Section \ref{sec:UStatistics}}
\label{sec:Utech}
\begin{proof}[Proof of Theorem~\ref{thm:Ustat}]
Note that $\hat{\theta}_{\gamma}^\mathcal{M}$ can be written as
\begin{align*}
\hat{\theta}_{\gamma}^\mathcal{M} = (\hat\theta_0 - \gamma^T\hat{\theta}_{0,\mathcal{M}}) + \gamma^T\hat\theta_\mathcal{M}.
\end{align*}
The two terms on the right hand side are independent because the former uses data in $A_0$ and the latter uses data in $A_0^c$. In addition, the first term on the right hand side is a $U$-Statistic with kernel 
\[
\tilde \phi(z_1,\cdots, z_r) = \phi(z_1,\cdots, z_r)-\sum_{m\in \mathcal{M}}\gamma_m\phi_m(z_1^m,\cdots, z_r^m),
\]
and the second term  is a linear combination of $\mathcal{M}$ independent $U$-Statistics each with kernel $\phi_m(z_1^m,\cdots, z_r^m)$.  If $\bbE\{\phi^2(Z_1,\cdots, Z_r)\}<\infty$ and $\bbE\{\phi_m^2(Z_1^m,\cdots, Z_r^m)\}<\infty$ for a all $m\in \mathcal{M}$, we have
\[
\mathbb{E} \Bigl\{\phi(Z_1,\cdots, Z_r)-\sum_{m\in \mathcal{M}}\gamma_m\phi_m(Z_1^m,\cdots, Z_r^m)\Bigr\}^2<\infty.
\]
Let $\theta_m = \mathbb{E}\{\phi_m\bigl(Z_{1}^m,\ldots, Z_{r}^m\bigr)\}$ and $\theta_\mathcal{M} = (\theta_m, m\in \mathcal{M})^T$. By classical $U$-Statistics theory (see for, example, \citet[Theorem~12.3]{vanderVaart:1998}), we have that
\begin{align*}
&\sqrt{n_0}\Bigl\{\hat\theta_0 - \gamma^T\hat\theta_{0,\mathcal{M}} - (\theta-\gamma^T\theta_\mathcal{M}) \Bigr\} \rightarrow^d N(0, r^2v_0); \quad \sqrt{n_m}(\hat\theta_{m} - \theta_m)\rightarrow^d N(0, r^2v_m)
\end{align*}
as $n \rightarrow \infty$, where $v_0  := \cov\{\tilde\phi(Z_1, Z_2,\cdots, Z_r), \tilde\phi(Z_1, Z_{r+2},\cdots, Z_{2r})\}$ and \\ $v_m := \cov\{\phi_m(Z_1^m, Z_2^m,\cdots, Z_r^m), \phi(Z_1^m, Z_{r+2}^m,\cdots, Z_{2r}^m)\}$.

Now, by the independence of $\hat\theta_0 - \gamma^T\hat\theta_{0,\mathcal{M}}$ and $\hat\theta_{\mathcal{M}}$, and noting that $\lim_{n\rightarrow \infty} \frac{n_m}{n_0} = \frac{q_m}{q_0}$, we have
\[
\sqrt{n_0}\bigl(\hat\theta_{\gamma}^\mathcal{M} - \theta\bigr)\rightarrow^d N\Bigl(0, r^2\bigl(v_0 + \sum_{m\in \mathcal{M}}\gamma_mq_m^{-1}q_0v_m\bigr)\Bigr).
\]
Finally, by the definition of $\tilde\phi(z_1,\cdots, z_r)$ we can derive that $v_0 + \sum_{m\in \mathcal{M}}\gamma_mq_m^{-1}q_0v_m = \psi_1 + \gamma^T\Lambda_1\gamma - 2\gamma^T\Omega_1$. This completes the proof of the theorem. 
 \end{proof} 

\bigskip

\begin{proof}[Proof of Theorem~\ref{thm:Ustat2}]
	First, we formally define our $U$-Statistic estimates of $\Omega_{U}$ and $\Lambda_U$.  We have
\begin{align}
\label{eq:OmegaU}
	\hat{\Omega}_{U,m} &:= \frac{1}{2{n_0 \choose 4r-2} } \sum_{\{i_1, \ldots, i_{4r-2}\} \subseteq A_0} \Bigl[\bigl\{ \phi \bigl(Z_{i_1},\ldots, Z_{i_r}\bigr) - \phi\bigl(Z_{i_{2r}}, \ldots, Z_{i_{3r-1}}\bigr)\bigr\}  \nonumber
	\\ & \hspace{120pt} \bigl\{ \phi_m\bigl(Z^m_{i_1}, Z^m_{i_{r+1}},\ldots, Z^m_{i_{2r-1}}\bigr) - \phi_{m}\bigl(Z_{i_{2r}}^{m}, Z_{i_{3r}}^{m},\ldots, Z_{i_{4r-2}}^{m}\bigr)\bigr\} \Bigr].
\end{align}
Moreover
\begin{align}
\label{eq:Lambdam}
\hat{\Lambda}_{U, m, m} &:= \Bigl(1 + \frac{n_0}{n_m}\Bigr)  \frac{1}{2{n_0 + n_m \choose 4r-2} } \sum_{\{i_1, \ldots, i_{4r-2}\} \subseteq A_0 \cup A_m} \Bigl[\bigl\{ \phi_{m} \bigl(Z_{i_1}^{m},\ldots, Z_{i_r}^{m}\bigr) - \phi_{m}\bigl(Z_{i_{2r}}^{m}, \ldots, Z_{i_{3r-1}}^{m}\bigr)\bigr\}  \nonumber
\\ & \hspace{90pt} \bigl\{ \phi_{m}\bigl(Z_{i_1}^{m}, Z_{i_{r+1}}^{m},\ldots, Z_{i_{2r-1}}^{m}\bigr) - \phi_{m}\bigl(Z_{i_{2r}}^{m}, Z_{i_{3r}}^{m},\ldots, Z_{i_{4r-2}}^{m}\bigr)\bigr\} \Bigr];
\end{align}
and
\begin{align}
\label{eq:Lambdam12}
\hat{\Lambda}_{U, m_1, m_2} &:=\frac{1}{2{n_0 + n_{m_{1,2}} \choose 4r-2} } \sum_{\{i_1, \ldots, i_{4r-2}\} \subseteq A_0 \cup A_{m_{1,2}} } \Bigl[\bigl\{ \phi_{m_1} \bigl(Z_{i_1}^{m_1},\ldots, Z_{i_r}^{m_1}\bigr) - \phi_{m_1}\bigl(Z_{i_{2r}}^{m_1}, \ldots, Z_{i_{3r-1}}^{m_1}\bigr)\bigr\}  \nonumber
	\\ & \hspace{60pt} \bigl\{ \phi_{m_2}\bigl(Z_{i_1}^{m_2}, Z_{i_{r+1}}^{m_2},\ldots, Z_{i_{2r-1}}^{m_2}\bigr) - \phi_{m_2}\bigl(Z_{i_{2r}}^{m_2}, Z_{i_{3r}}^{m_2},\ldots, Z_{i_{4r-2}}^{m_2}\bigr)\bigr\} \Bigr].
\end{align}

Then, by classical $U$-Statistics theory \citep[Theorem~12.3]{vanderVaart:1998}, we can prove that  $\hat\Lambda$ and $\hat\Omega$ are consistent estimators of $\Lambda$ and $\Omega$, respectively. Then the consistency of $\hat\gamma$ to $\gamma^*$ follows automatically. Next note that 
	\begin{align*}
	\sqrt{n_0}(\hat\theta_{\hat\gamma}^\mathcal{M} - \theta) - \sqrt{n_0}(\hat\theta_{\gamma}^\mathcal{M} - \theta)  = \sqrt{n_0}(\hat\theta_{\hat\gamma}^\mathcal{M} - \hat\theta_{\gamma}^\mathcal{M}) = \sqrt{n_0}(\hat{\theta}_{0,\mathcal{M}} - \hat\theta_\mathcal{M})^T (\hat\gamma - \gamma^*). 
	\end{align*}
	By the proof of Theorem \ref{thm:Ustat}, the vector $\sqrt{n_0}(\hat{\theta}_{0,\mathcal{M}} - \hat\theta_\mathcal{M})$ is jointly asymptotically normal. Since $\hat\gamma - \gamma^* = o_p(1)$, it follows from the Slutsky's theorem that 
	\[
		\sqrt{n_0}(\hat\theta_{\hat\gamma}^\mathcal{M} - \theta) - \sqrt{n_0}(\hat\theta_{\gamma}^\mathcal{M} - \theta) \rightarrow^p 0.
	\]
	This, together with Theorem \ref{thm:Ustat}, completes the proof of the theorem. 
	\end{proof}
	
	\bigskip

\subsection{Conditions and proofs for the results in Section \ref{sec:RCD}}
\label{sec:RCDtech}
We first formally state our assumptions. In the following, $L > 0$ is some universal constant.  
\begin{assumption} 
	Suppose that $P_X$ and, for $m \in \{0,1\} ^d \setminus \{(1,\ldots,1)^T\}$, the marginal $X^m$ distribution $P_{X^m}$ have densities $f_X$, and $f_{X^m}$, respectively, that satisfy $|f_X(z_1) - f_X(z_2)| \leq L\|z_1 - z_2\|$, for all $z_1,z_2 \in \mathbb{R}^{d}$, and, for each $m \in \{0,1\} ^d \setminus \{(1,\ldots,1)^T\}$, we have  $|f_{X^m}(z_1^m) - f_{X^m}(z_2^m)| \leq L \|z_1^m - z_2^m\|$, for all $z_1^m,z_2^m \in \mathbb{R}^{d_m}$.
\end{assumption} 

\begin{assumption} 
	Suppose the kernel is such that $\bar{K} := \sup_{z \in \mathbb{R}^d} (1 + \|z\|) K(z) < \infty$, and that $\mu_0 = \mu_0(K) := \int_{\mathbb{R}^d} K(z) \, dz = 1$, $\mu_1 = \mu_1(K) := \int_{\mathbb{R}^d} \|z\| K(z) \, dz < \infty$.   We also ask that $\nu = \nu(K) := \int_{\mathbb{R}^d} K^2(z) \, dz < \infty$.  Moreover, for each $m \in \mathcal{M}$, we have  $\bar{K}_m := \sup_{z \in \mathbb{R}^{d_m}} (1 + \|z\|) K_m(z) < \infty$, $\mu_{0,m} = \mu_{0,m}(K_m) := \int_{\mathbb{R}^{d_m}} K_m(z) \, dz < \infty$, $\mu_{1,m} = \mu_{1,m}(K_m) := \int_{\mathbb{R}^{d_m}} \|z\| K_m(z) \, dz <\infty$, and $\nu_m = \nu_{m}(K_m) := \int_{\mathbb{R}^{d_m}} K_m^2(z) \, dz < \infty$. Finally for $m_1 \neq m_2  \in \mathcal{M}$, letting $m^{1,2} = \mathrm{pmax}\{m_1,m_2\} \in \{0,1\}^d$ and $m_{1,2} = \mathrm{pmim}\{m_1,m_2\} \in \{0,1\}^d$ denote the entrywise maximums and minimums, respectively, of $m_1$ and $m_2$, finally we suppose that $\nu_{m_1, m_2} = \nu_{m_1,m_2}(K_{m_1}, K_{m_2}) := \int_{\mathbb{R}^{d_{m_{1,2}}}} K_{m_1}(z^{m_1}) K_{m_2}(z^{m_2})  \, dz^{m_{1,2}} < \infty.$
\end{assumption} 

\begin{assumption}  
	\label{ass:Lip}
	We have that $|\eta(z_1) - \eta(z_2)| \leq L\|z_1 - z_2\|$, for all $z_1,z_2 \in \mathbb{R}^{d}$, and, for each $m \in \mathcal{M}$, $|\eta_m(z_1^m) - \eta_m(z_2^m)| \leq L \|z_1^m - z_2^m\|$, for all $z_1^m,z_2^m \in \mathbb{R}^{d_m}$. 
\end{assumption} 
\begin{assumption}  
	\label{ass:tauLip}
	For each $m \in \mathcal{M}$, we have $|\tau_m(z_1^m) - \tau_m(z_2^m)| \leq L\|z_1^m - z_2^m\|$, for all $z_1^m,z_2^m \in \mathbb{R}^{d_m}$.  
	Finally, we ask, for all $m_1, m_2 \in \mathcal{M}$, that $|\tau_{m_1,m_2}(z_1^{m_{1,2}}) - \tau_{m_1,m_2}(z_2^{m_{1,2}}) | \leq L\|z_1^{m_{1,2}} - z_2^{m_{1,2}}\|$, for all $z_1^{m_{1,2}} ,z_2^{m_{1,2}} \in \mathbb{R}^{d_{m_{1,2}}}$.
\end{assumption}  

\bigskip

\begin{proof}[Proof of Theorem~\ref{thm:DensityPi}]
	First, we have  
	\begin{align*}
	&\mathbb{E}[\{\hat{f}_{\gamma}^{\mathcal{M}}  - f_X(x)\}^2 - \{\hat{f}_0  - f_X(x)\}^2]  
	\\ & \hspace{60pt} =  \mathbb{E}[(\hat{f}_{\gamma}^{\mathcal{M}}(x) - \hat{f}_0)^2 + 2(\hat{f}_{\gamma}^{\mathcal{M}} - \hat{f}_0)\{\hat{f}_0  - f_X(x)\}]
	\\ & \hspace{60pt}  =  \gamma^T\mathbb{E}\{(\hat{f}_{0,\mathcal{M}} - \hat{f}_{\mathcal{M}}) (\hat{f}_{0,\mathcal{M}} - \hat{f}_{\mathcal{M}})^T\} \gamma- 2\gamma^T\mathbb{E}[(\hat{f}_{0,\mathcal{M}} - \hat{f}_{\mathcal{M}}) \{\hat{f}_0  - f_X(x)\}].
	\end{align*} 
	Now observe that, for each $m \in \mathcal{M}$,
	\[
	\mathbb{E}(\hat{f}_{0,m} - \hat{f}_{m}) = \frac{1}{n_0 h^{d_m}} \sum_{i \in A_0} \mathbb{E}\Bigl\{K_m\Bigl( \frac{X_i^{m} - x^{m}}{h}\Bigr)\Bigr\} - \frac{1}{n_m h^{d_m}} \sum_{i \in A_m} \mathbb{E}\Bigl\{K_m\Bigl( \frac{X_i^{m} - x^{m}}{h}\Bigr) \Bigr\} = 0.
	\]
	Thus, using also that $\hat{f}_{0,\mathcal{M}}$ and $\hat{f}_{\mathcal{M}}$ are independent, we can write
	\begin{align}
	\label{eq:denVar}
	&\mathbb{E}\{(\hat{f}_{0,\mathcal{M}} - \hat{f}_{\mathcal{M}})(\hat{f}_{0,\mathcal{M}} - \hat{f}_{\mathcal{M}})^T\}  \nonumber
	\\ & \hspace{30pt} = \mathbb{E}[\{\hat{f}_{0,\mathcal{M}} - \mathbb{E}(\hat{f}_{0,\mathcal{M}}) + \hat{f}_{\mathcal{M}} - \mathbb{E}(\hat{f}_{\mathcal{M}}) \}\{\hat{f}_{0,\mathcal{M}} - \mathbb{E}(\hat{f}_{0,\mathcal{M}}) + \hat{f}_{\mathcal{M}} - \mathbb{E}(\hat{f}_{\mathcal{M}}) \}^T] \nonumber
	\\ & \hspace{30pt}  = \mathrm{Cov}(\hat{f}_{0,\mathcal{M}}) + \mathrm{Cov}(\hat{f}_{\mathcal{M}}).
	\end{align} 
It remains to show that $\mathrm{Cov}(\hat{f}_{0,\mathcal{M}}) + \mathrm{Cov}(\hat{f}_{\mathcal{M}}) = \Lambda_D(1+o(1))$ and $\mathbb{E}[(\hat{f}_{0,\mathcal{M}} - \hat{f}_{\mathcal{M}}) \{\hat{f}_0  - f_X(x)\}] = \Omega_D(1+o(1))$.

For the diagonal terms in the covariances in \eqref{eq:denVar}, we have 
	\begin{align*}
	\var(\hat{f}_{0,m})  + \var(\hat{f}_{m}) &= \frac{1}{h^{2d_m}}\Bigl(\frac{1}{n_0} + \frac{1}{n_m}\Bigr) \var\Bigl\{K_m\Bigl( \frac{X^{m} - x^{m}}{h}\Bigr) \Bigr\}
	\\ & = \frac{1}{h^{2d_m}}\Bigl(\frac{1}{n_0} + \frac{1}{n_m}\Bigr)\biggl[ \int_{\mathbb{R}^{d_m}} K_m^2\Bigl( \frac{z^{m} - x^{m}}{h}\Bigr) f_{X^m}(z^m) \, dz^m 
	\\ & \hspace{ 120pt} - \Bigl\{\int_{\mathbb{R}^{d_m}} K_m\Bigl( \frac{z^{m} - x^{m}}{h}\Bigr) f_{X^m}(z^m) \, dz^m\Bigr\}^2 \biggr].
	\end{align*}
 After making the substitution $u^m = \frac{z^m - x^m}{h}$ and using assumptions  \textbf{A1} and \textbf{A2}, we deduce that
    \begin{align*} 
 & \Bigl| \int_{\mathbb{R}^{d_m}} K_m^2\Bigl( \frac{z^{m} - x^{m}}{h}\Bigr) f_{X^m}(z^m) \, dz^m - h^{d_m} f_{X^m}(x^m) \nu_{m} \Bigr| 
   \\ & \hspace{60pt} \leq h^{d_m}\int_{\mathbb{R}^{d_m}} K_m^2(u^{m}) |f_{X^m}(x^m + h u^m) - f_{X^m}(x^m)| \, du^m \leq L h^{d_m+1} \bar{K}_m \mu_{1,m}.
    	\end{align*}
Whereas   	
	    \begin{align} 
	    \label{eq:denbias}
&	\Bigl|  \int_{\mathbb{R}^{d_m}} K_m\Bigl( \frac{z^{m} - x^{m}}{h}\Bigr) f_{X^m}(z^m) \, dz^m - h^{d_m}f_{X^m}(x^m) \mu_{0,m} \Bigr| \nonumber
\\ & \hspace{60pt} \leq h^{d_m}\int_{\mathbb{R}^{d_m}} K_m(u^{m}) |f_{X^m}(x^m + h u^m) - f_{X^m} (x^m) | \, du^m \leq L h^{d_m+1} \mu_{1,m}. 
		\end{align}
It follows that 
	\[
	\var(\hat{f}_{0,m})  + \var(\hat{f}_{m}) = \frac{f_{X^m}(x^m) \nu_{m}}{h^{d_m}}\Bigl(\frac{1}{n_0} + \frac{1}{n_m}\Bigr)\{1 + o(1)\},
	\]
as $n \rightarrow \infty$, uniformly for $h \in [n^{-\beta}, n^{-\alpha}]$.  
    
    For the off-diagonal terms in \eqref{eq:denVar}: first, for $m_1 \neq m_2 \in \mathcal{M}$, we have that $\mathrm{Cov}(\hat{f}_{m_1},\hat{f}_{m_2}) = 0,$ since $A_{m_1}$ and $A_{m_2}$ are disjoint for $m_1 \neq m_2$.      For the remaining terms, we have 
	\begin{align*}
	& \mathrm{Cov}(\hat{f}_{0,m_1},\hat{f}_{0,m_2}) 
	\\ &\hspace{10 pt}  = \frac{1}{n_0 h^{d_{m_1} + d_{m_2}}} \mathrm{Cov}\Bigl\{K_{m_1}\Bigl(\frac{X^{m_1} - x^{m_1}}{h}\Bigr), K_{m_2}\Bigl( \frac{X^{m_2} - x^{m_2}}{h}\Bigr) \Bigr\}
	\\ &\hspace{10 pt} = \frac{1}{n_0h^{d_{m_1} + d_{m_2}}} \biggl\{ \int_{\mathbb{R}^{d}} K_{m_1} \Bigl(\frac{z^{m_1} - x^{m_1}}{h}\Bigr) K_{m_2}\Bigl(\frac{z^{m_2} - x^{m_2}}{h}\Bigr) f_{X}(z) \, dz \biggr\}
    \\ & \hspace{40 pt} - \frac{1}{n_0 h^{d_{m_1} + d_{m_2}}} \biggl\{\int_{\mathbb{R}^{d}} K_{m_1}\Bigl( \frac{z^{m_1} - x^{m_1}}{h}\Bigr) f_{X}(z) \, dz \biggr\} \biggl\{\int_{\mathbb{R}^{d}} K_{m_2}\Bigl( \frac{z^{m_2} - x^{m_2}}{h}\Bigr) f_{X}(z) \, dz \biggr\}.
    \end{align*} 
    As above, we make the substitution $u = \frac{z - x}{h}$, which gives
    \begin{align*} 
	& \int_{\mathbb{R}^{d}} K_{m_1}\Bigl( \frac{z^{m_1} - x^{m_1}}{h}\Bigr) K_{m_2}\Bigl(\frac{z^{m_2} - x^{m_2}}{h}\Bigr) f_{X}(z) \, dz 
    \\ & \hspace{90 pt}= h^d \int_{\mathbb{R}^{d}} K_{m_1}(u^{m_1}) K_{m_2}(u^{m_2}) f_{X}(x + hu) \, du 
    \\ & \hspace{90 pt}= h^d \int_{\mathbb{R}^{d_{m_{1,2}}}} \int_{\mathbb{R}^{d - d_{m_{1,2}}}} K_{m_1}(u^{m_1}) K_{m_1}(u^{m_2}) f_{X}(x + hu) \, du^{1_d - m_{1,2}}  du^{m_{1,2}} 
    \\ & \hspace{90 pt}= h^d \int_{\mathbb{R}^{d_{m_{1,2}}}}  K_{m_1}(u^{m_1}) K_{m_2}(u^{m_2})  \int_{\mathbb{R}^{d - d_{m_{1,2}}}} f_{X}(x + hu) \, du^{1_d - m_{1,2}}  du^{m_{1,2}} 
	\\ & \hspace{90 pt}= h^{d_{m_{1,2}}}\int_{\mathbb{R}^{d_{m_{1,2}}}} K_{m_1}(u^{m_1}) K_{m_2}(u^{m_2}) f_{X^{m_{1,2}}}(x^{m_{1,2}} + hu^{m_{1,2}}) \, du^{m_{1,2}} 
	\\ & \hspace{90 pt} = h^{d_{m_{1,2}}}\nu_{m_{1},m_{2}} f_{X^{m_{1,2}}}(x^{m_{1,2}}) \{1+o(1)\},
	\end{align*}
uniformly for $h \in [n^{-\beta}, n^{-\alpha}]$.  It follows from the previous calculation and~\eqref{eq:denbias}, that 
	\[
	 \mathrm{Cov}(\hat{f}_{0,m_1},\hat{f}_{0,m_2}) =  \frac{\nu_{m_{1},m_{2}}f_{X^{m_{1,2}}}(x^{m_{1,2}})}{n_0 h^{d_{m^{1,2}}}}\{1+o(1)\},
    \]
uniformly for $h \in [n^{-\beta}, n^{-\alpha}]$. This proves that first claim that $\mathrm{Cov}(\hat{f}_{0,\mathcal{M}}) + \mathrm{Cov}(\hat{f}_{\mathcal{M}}) = \Lambda_D(1+o(1))$.
	
	Finally, since $\hat{f}_0$ and $\hat{f}_{\mathcal{M}}$ are independent, we have that
	\begin{align*}
	\mathbb{E}[(\hat{f}_{0,\mathcal{M}} - \hat{f}_{\mathcal{M}}) \{\hat{f}_0  - f_X(x)\}] &= \mathbb{E}[(\hat{f}_{0,\mathcal{M}} - \hat{f}_{\mathcal{M}} ) \{\hat{f}_0 - \mathbb{E}(\hat{f}_0)  + \mathbb{E}(\hat{f}_0) - f_X(x)\}]
	\\ & = \mathbb{E}[\{\hat{f}_{0,\mathcal{M}} - \mathbb{E}(\hat{f}_{0,\mathcal{M}}) - \hat{f}_{\mathcal{M}}+ \mathbb{E}(\hat{f}_{\mathcal{M}})\} \{\hat{f}_0 - \mathbb{E}(\hat{f}_0)\}]
	\\ & = \mathrm{Cov}(\hat{f}_{0,\mathcal{M}}, \hat{f}_0) +  \mathrm{Cov}(\hat{f}_{\mathcal{M}}, \hat{f}_0) =  \mathrm{Cov}(\hat{f}_{0,\mathcal{M}}, \hat{f}_0).
	\end{align*} 
	Then, reusing the covariance calculation above, for $m \in \mathcal{M}$, we have
	\begin{align*}
	\mathrm{Cov} (\hat{f}_{0}, \hat{f}_{0,m}) & = \frac{1}{n_0 h^{d+d_m}} \mathrm{Cov}\Bigl\{ K\Bigl(\frac{X - x}{h}\Bigr),  K_m\Bigl(\frac{X^m - x^m}{h}\Bigr) \Bigr\} = \frac{\nu_{0,m}f_X(x)}{n_0h^{d_m}} \{1+o(1)\},
	\end{align*}
	uniformly for $h \in [n^{-\beta}, n^{-\alpha}]$. This completes the proof of the second claim and hence concludes the proof of the theorem.
\end{proof}

\bigskip

\begin{proof}[Proof of Theorem~\ref{thm:RegressionPi}]
First, conditionally on the observed data, we have
\begin{align} \label{eq: MSE-all}
&\mathbb{E}[\{\hat{\eta}^{\mathcal{M}}_{\gamma} - \eta(x)\}^2 |X_1^{m_1}, \ldots, X_n^{m_n}] - \mathbb{E}[\{\hat{\eta}_{0} - \eta(x)\}^2 | X_1^{m_1}, \ldots, X_n^{m_n}] \nonumber
\\ & \hspace{30 pt} = \mathbb{E}\{(\hat{\eta}^{\mathcal{M}}_{\gamma} - \hat\eta_0)^2 |X_1^{m_1}, \ldots, X_n^{m_n}\} + 2\mathbb{E}[\{\hat{\eta}_{0} - \eta(x)\}\{\hat\eta^{\mathcal M}_{\gamma}-\hat\eta_0\} | X_1^{m_1}, \ldots, X_n^{m_n}]. 
\end{align}
We analyse the two terms in \eqref{eq: MSE-all} separately.  We first introduce some notation and facts that will be used repeatedly in the proof. Recall that we can write the local constant estimator as a linear function of the responses $Y_{[n]} := (Y_1, \ldots, Y_n)^T \in  \mathbb{R}^n$. Indeed, let 
\begin{equation}
\label{eq:weights}
W_{A,m} = W_{A,m,h,K_m} := \biggl(K_m\Bigl(\frac{X^{m}_1 - x^{m}}{h}\Bigr)\mathbbm{1}_{\{1 \in A\}}, \ldots, K_m\Bigl(\frac{X^{m}_n - x^{m}}{h}\Bigr)\mathbbm{1}_{\{n \in A\}}\biggr)^T.
\end{equation}
Then $\hat{\eta}_{A,m}(x^{m}) =  H_{A,m,h,K_m}^T Y_{[n]}$, where $H_{A,m} = H_{A,m,h,K_m} := (\sum_{i = 1}^{n} W_{A,m,i})^{-1} W_{A,m}$. For simplicity of presentation, write $H_0 = H_{A_0,0}$, $H_{0,m} = H_{A_0,m}$ and $H_m = H_{A_m,m}$, and let $H_{0,\mathcal M} = (H_{0,m} : m \in \mathcal{M})$ be the $n \times |\mathcal{M}|$ matrix with columns $H_{0,m}$, for $m \in \mathcal{M}$, and similarly let $H_{\mathcal{M}} = (H_m : m \in \mathcal{M}) \in \mathbb{R}^{n \times |\mathcal{M}|}$.

Next, let $E = \mathbb{E}(Y_{[n]} | X_1^{m_1}, \ldots, X_n^{m_n}) =  \bigl(\eta_{m_1}(X_1^{m_1}), \ldots, \eta_{m_n}(X_n^{m_n})\bigr)^T$ be the conditional expectation of the responses, and let $\Gamma :=  \mathrm{diag}\{\tau_{m_1}(X_1^{m_1}),\ldots, \tau_{m_n}(X_n^{m_n})\}$. Further,  let $\Gamma^* = \mathrm{diag}\{\tau_{m_1}(x^{m_1}),\ldots, \tau_{m_n}(x^{m_n})\}$ and  $E^* = \{\eta_{m_1}(x^{m_1}), \ldots, \eta_{m_n}(x^{m_n})\}^T$. We will make use of the following facts (i) $1_n^T H_{A,m} = 1$; (ii) $E^{*T} H_0 = \eta(x)$,  (iii) $H_m^T E^{*} = \eta_m(x^m)$,  (iv) $H_{0,m}^TE^{*} = \eta(x)$, (v) $H_m^T \Gamma^{*}H_m =\tau_m(x^m) H_m^TH_m$, (vi) $H_{0,m}^T \Gamma^{*} = 0$, (vii) $\Gamma H_0 = \Gamma^{*}H_0 = 0$ and (viii) $H_{\mathcal{M}}^TH_0 = 0$.

Furthermore, we claim that, for $m \in \mathcal{M}$, the following results are true
\begin{equation}
\label{eq:HH1}
H_{0}^T H_{0,m} = \frac{\nu_{0,m}}{\mu_{0,m}f_{X^m}(x^m) n_0 h^{d_m}}\{1+O_p(h)\};
\end{equation}
\begin{equation}
\label{eq:HH2}
H_{0,m}^T H_{0,m} = \frac{\nu_{m}}{\mu^2_{0,m}f_{X^m}(x^m) n_0 h^{d_m}}\{1+O_p(h)\};
\end{equation}
and 
\begin{equation}
\label{eq:HH3}
H_{m}^T H_{m} = \frac{\nu_{m}}{\mu^2_{0,m}f_{X^m}(x^m) n_m h^{d_m}}\{1+O_p(h)\},
\end{equation}
uniformly for $h \in [n^{-\beta}, n^{-\alpha}]$. Furthermore, for $m_1 \neq m_2 \in \mathcal{M}$, we have 
\begin{equation}
\label{eq:HH4}
H_{0,m_1}^T H_{0,m_2} = \frac{\nu_{m_{1},m_{2}} f_{X^{m_{1,2}}}(x^{m_{1,2}})} {\mu_{0,m_{1}}f_{X^{m_{1}}}(x^{m_{1}})\mu_{0,m_{2}}f_{X^{m_{2}}}(x^{m_{2}}) n_0 h^{d_{m^{1,2}}}}\{1+O_p(h)\};
\end{equation}
uniformly for $h \in [n^{-\beta}, n^{-\alpha}]$.  To see \eqref{eq:HH1}, write
\begin{align}
\label{eq:numer}
H_{0}^T H_{0,m} & = \frac{\sum_{i\in A_0}K\bigl(\frac{X_i - x}{h}\bigr)K_m\bigl(\frac{X^{m}_i - x^{m}}{h}\bigr)}{\sum_{i \in A_0} K\bigl(\frac{X_i - x}{h}\bigr) \sum_{i \in A_0} K_m\bigl(\frac{X^{m}_i - x^{m}}{h}\bigr)} \nonumber
\\ & = \frac{\frac{1}{n_0^2 h^{d+d_m}} \sum_{i\in A_0}K\bigl(\frac{X_i - x}{h}\bigr)K_m\bigl(\frac{X^{m}_i - x^{m}}{h}\bigr)}{\frac{1}{n_0h^d}\sum_{i \in A_0} K\bigl(\frac{X_i - x}{h}\bigr) \frac{1}{n_0h^{d_m}}\sum_{i \in A_0} K_m\bigl(\frac{X^{m}_i - x^{m}}{h}\bigr)}.
\end{align}
We consider the numerator and denominator above separately. For the numerator, observe that
\begin{align*}
&\mathbb{E}\Bigl\{\frac{1}{n_0^2 h^{d+d_m}} \sum_{i\in A_0}K\bigl(\frac{X_i - x}{h}\bigr)K_m\bigl(\frac{X^{m}_i - x^{m}}{h}\bigr)\Bigr\} 
\\ & \hspace{90pt} = \frac{1}{n_0 h^{d+d_m}} \int_{\mathbb{R}^d} K\bigl(\frac{z - x}{h}\bigr)K_m\bigl(\frac{z^{m} - x^{m}}{h}\bigr) f_X(z) \, dz 
\\ & \hspace{90pt}=  \frac{1}{n_0 h^{d_m}} \int_{\mathbb{R}^d} K(u)K_m(u^m) f_X(x + hu) \, dz 
\\ & \hspace{90pt} =  \frac{\nu_{0,m} f_X(x)}{n_0 h^{d_m}} + \frac{1}{n_0 h^{d_m}} \int_{\mathbb{R}^d} K(u)K_m(u^m) \{f_X(x + hu) - f_X(x)\} \, dz.
\end{align*}
Moreover, by Assumptions \textbf{A1} and \textbf{A2}, we have  
\[
\int_{\mathbb{R}^d} \Bigl| K(u)K_m(u^m) \{f_X(x + hu) - f_X(x)\} \Bigr| \, dz \leq Lh \int_{\mathbb{R}^d} K(u) K_m(u^m) \|u\| \, du \leq L h \mu_1 \bar{K}_m.
\]
Hence, using Markov's inequality, the numerator in \eqref{eq:numer} admits the following expression
\[
\frac{1}{n_0^2 h^{d+d_m}} \sum_{i\in A_0}K\Bigl(\frac{X_i - x}{h}\Bigr)K_m\Bigl(\frac{X^{m}_i - x^{m}}{h}\Bigr) = \frac{\nu_{0,m} f_X(x)}{n_0 h^{d_m}} + O_p(1/(n_0 h^{d_m-1})),
\]
uniformly for $h \in [n^{-\beta}, n^{-\alpha}]$. For the denominator, by appealing to similar arguments to those in the proof of Theorem~\ref{thm:DensityPi}, we have that
\[
\frac{1}{n_0h^d}\sum_{i \in A_0} K\Bigl(\frac{X_i - x}{h}\Bigr) = f_X(x) + O_p(h),
\]
and 
\[
\frac{1}{n_0h^{d_m}}\sum_{i \in A_0} K_m\Bigl(\frac{X^m_i - x^m}{h}\Bigr) = f_{X^m}(x^m) \mu_{0,m} + O_p(h),
\]
uniformly for $h \in [n^{-\beta}, n^{-\alpha}]$. The claim in \eqref{eq:HH1} then follows by Slutsky's Theorem.  Moreover, the claims in \eqref{eq:HH2} and \eqref{eq:HH3} follow by the same argument with only minor changes.   To see \eqref{eq:HH4}, observe that 
\begin{align*}
&\mathbb{E}\Bigl\{\frac{1}{n_0^2 h^{d_{m_1}+d_{m_2}}} \sum_{i\in A_0} K_{m_1}\bigl(\frac{X^{m_1}_i - x^{m_1}}{h}\bigr)K_{m_2}\bigl(\frac{X^{m_2}_i - x^{m_2}}{h}\bigr)\Bigr\} 
\\ & \hspace{60pt} = \frac{1}{n_0 h^{d_{m_1}+d_{m_2}}} \int_{\mathbb{R}^d} K_{m_1}\bigl(\frac{z^{m_1} - x^{m_1}}{h}\bigr)K_{m_2}\bigl(\frac{z^{m_2} - x^{m_2}}{h}\bigr) f_X(z) \, dz 
\\ & \hspace{60pt} = \frac{1}{n_0 h^{d_{m_1}+d_{m_2}}} \int_{\mathbb{R}^{d_{m_{1,2}}}} K_{m_1}\bigl(\frac{z^{m_1} - x^{m_1}}{h}\bigr)K_{m_2}\bigl(\frac{z^{m_2} - x^{m_2}}{h}\bigr) f_{X^{m_{1,2}}}(z^{m_{1,2}}) \, dz^{m_{1,2}} 
\\ & \hspace{60pt}=  \frac{1}{n_0 h^{d_{m_1}+d_{m_2} -d_{m_{1,2}}}} \int_{\mathbb{R}^{d_{m_{1,2}}}} K_{m_1}(u^{m_1}) K_{m_2}(u^{m_2}) f_{X^{m_{1,2}}}(x^{m_{1,2}} + h u^{m_{1,2}}) \, du^{m_{1,2}} 
\\ & \hspace{60pt} = \frac{\nu_{m_1,m_2}f_{X^{m_{1,2}}}(x^{m_{1,2}}) }{n_0 h^{d_{m^{1,2}}}}  + O(1/(n_0 h^{d_{m^{1,2}}-1})),
\end{align*}
uniformly for $h \in [n^{-\beta}, n^{-\alpha}]$. In the last step above we have used the fact that, by Assumptions \textbf{A1} and \textbf{A2}, we have   
\begin{align*}
&\int_{\mathbb{R}^{d_{m_{1,2}}}} K_{m_1}(u^{m_1}) K_{m_2}(u^{m_2}) \{f_{X^{m_{1,2}}}(x^{m_{1,2}} + h u^{m_{1,2}}) - f_{X^{m_{1,2}}}(x^{m_{1,2}})\} \, du^{m_{1,2}} 
\\ & \hspace{20pt}  \leq Lh \int_{\mathbb{R}^{d_{m_{1,2}}}} K_{m_1}(u^{m_1}) K_{m_2}(u^{m_2}) \|u^{m_{1,2}}\| \, du^{m_{1,2}}
\\ & \hspace{20pt}  \leq Lh \int_{\mathbb{R}^{d_{m_{1,2}}}} K_{m_1}(u^{m_1}) K_{m_2}(u^{m_2}) (\|u^{m_{1}}\| + \|u^{m_{2}}\|) \, du^{m_{1,2}} \leq L h(\mu_{1, m_1} \bar{K}_{m_2} + \mu_{1, m_2} \bar{K}_{m_1}).
\end{align*}
The claim in \eqref{eq:HH4} then follows by similar arguments used to prove \eqref{eq:HH1}.  We now return to the main argument.

\textbf{Part I: the first term in \eqref{eq: MSE-all}}:  By definition of $Y_{[n]}$ we have  
\[
\mathbb{E}(Y_{[n]}Y_{[n]}^T | X_1^{m_1}, \ldots, X_n^{m_n}) = E E^T + \sigma^2 I_{n\times n} + \Gamma.
\]
It follows from the definitions of $\hat{\eta}_{\gamma}^{\mathcal{M}}$ and $\hat{\eta}_{0}$ that the first term in \eqref{eq: MSE-all} takes the form
\begin{align}
\label{eq:mse1}
\mathbb{E}\{(\hat{\eta}_{\gamma}^{\mathcal{M}} - \hat{\eta}_{0})^2| X_1^{m_1}, \ldots, X_n^{m_n}\} & = \gamma^T \bigl(H_{0,\mathcal{M}}^T - H_{\mathcal{M}}^T\bigr) (EE^T + \sigma^2 I_{n\times n} + \Gamma) \bigl(H_{0,\mathcal{M}} - H_{\mathcal{M}}\bigr)\gamma \nonumber\\
& = \gamma^T\Lambda_R\gamma + R_1 + R_2,
\end{align}
where
\[
R_1 :=  \gamma^T  H_{\mathcal{M}}^T (\Gamma - \Gamma^*) H_{\mathcal{M}} \gamma
\]
and 
\[
R_2  := \{\gamma^T(H_{0,\mathcal{M}}^T - H_{\mathcal{M}}^T) E\}^2 +   \sigma^2 \gamma^T  H_{0,\mathcal{M}}^T H_{0,\mathcal{M}}\gamma + \gamma^T H_{\mathcal{M}}^T (\sigma^2I + \Gamma^*)H_{\mathcal{M}} \gamma - \gamma^T \Lambda_{\mathrm{R}} \gamma.
\]
Here we have used that $H_{0,\mathcal{M}}^T (\Gamma - \Gamma^*) H_{0,\mathcal{M}} = 0$. We next show that both $R_1$ and $R_2$ are small order terms.  

\textit{To bound $R_{1}$}: By assumption \textbf{A4},  we have that
\[
|\tau_{m}(x_i^{m}) - \tau_{m}(x^{m})| \leq L \|x_i^{m} - x^{m}\|.
\]
Moreover, similarly to \eqref{eq:HH3}, by Assumption \textbf{A2}, we have
\begin{align*}
H_{m}^T\mathrm{diag}(\|X_i^{m_i} - x^{m_i}\|, i = 1, \ldots, n)H_{m} &=  \frac{\sum_{i\in A_m} \|X_i^m - x^m\| K_m^2\bigl(\frac{X^{m}_i - x^{m}}{h}\bigr)}{\sum_{i \in A_m} K_m\bigl(\frac{X_i^m - x^m}{h}\bigr) \sum_{i \in A_m} K_m\bigl(\frac{X^{m}_i - x^{m}}{h}\bigr)}
\\ & \leq \frac{\frac{1}{n_m^2 h^{2d_m-1}} \sum_{i\in A_m} \bar{K}_m K_m\bigl(\frac{X^{m}_i - x^{m}}{h}\bigr)}{\{\frac{1}{n_mh^{d_m}}\sum_{i \in A_m} K_m\bigl(\frac{X^{m}_i - x^{m}}{h}\bigr)\}^2} 
\\ & = O_p\bigl(\frac{1}{n_m h^{d_m-1}}\bigr).
\end{align*}
uniformly for $h \in [n^{-\beta}, n^{-\alpha}]$.   Thus
\begin{align*}
|R_1| & = \Bigl|\gamma^TH_{\mathcal{M}}^T ( \Gamma - \Gamma^*)H_{\mathcal{M}}\gamma\Bigr| 
\\ & \leq \sum_{m \in \mathcal{M}} \gamma_m^2H_{m}^T\mathrm{diag}(\|X_i^{m_i} - x^{m_i}\|, i = 1, \ldots, n)H_{m} =  O_p \Bigl(\sum_{m\in \mathcal{M}}  \frac{\gamma_m^2}{n_m h^{d_m-1}}\Bigr),
\end{align*}
uniformly for $h \in [n^{-\beta}, n^{-\alpha}]$. 

\textit{To bound $R_2$}: First we decompose $R_2$ into the sum of three terms.  To that end, let
\begin{align*}
R_{21} &:= \{\gamma^T(H_{0,\mathcal{M}}^T - H_{\mathcal{M}}^T) E\}^2 - \sum_{m\in \mathcal{M}} \frac{\gamma_m^2\nu_{m} \tau_m(x^m)}{\mu^2_{0,m}f_{X^m}(x^m) n_0 h^{d_m}} 
\\ & \hspace{120pt} - \sum_{m_1 \neq m_2 \in \mathcal{M}}\frac{ \gamma_{m_1}\gamma_{m_2}  \tau_{m_1,m_2}(x^{m_{1,2}}) \nu_{m_{1},m_{2}} f_{X^{m_{1,2}}}(x^{m_{1,2}})} {\mu_{0,m_{1}}f_{X^{m_{1}}}(x^{m_{1}})\mu_{0,m_{2}}f_{X^{m_{2}}}(x^{m_{2}}) n_0 h^{d_{m^{1,2}}}};
\end{align*} 
\begin{align*} 
R_{22} &:= \sigma^2 \gamma^T  H_{0,\mathcal{M}}^T H_{0,\mathcal{M}}\gamma - \sum_{m\in \mathcal{M}} \frac{\gamma_m^2\nu_{m} \sigma^2}{\mu^2_{0,m}f_{X^m}(x^m) n_0 h^{d_m}} 
\\ & \hspace{120pt} - \sum_{m_1 \neq m_2 \in \mathcal{M}}\frac{ \gamma_{m_1}\gamma_{m_2}  \sigma^2 \nu_{m_{1},m_{2}} f_{X^{m_{1,2}}}(x^{m_{1,2}})} {\mu_{0,m_{1}}f_{X^{m_{1}}}(x^{m_{1}})\mu_{0,m_{2}}f_{X^{m_{2}}}(x^{m_{2}}) n_0 h^{d_{m^{1,2}}}};
\end{align*} 
and
\begin{align*} 
R_{23} := \gamma^T H_{\mathcal{M}}^T (\sigma^2I + \Gamma^*)H_{\mathcal{M}} \gamma - \sum_{m\in \mathcal{M}} \frac{\gamma_m^2\nu_{m} \{\sigma^2 + \tau_m(x^m)\}}{\mu^2_{0,m}f_{X^m}(x^m) n_m h^{d_m}}.
\end{align*}
Note that $R_2 = R_{21} + R_{22} + R_{23}$, since
\begin{align*} 
\gamma^T \Lambda_{\mathrm{R}} \gamma &= \sum_{m\in \mathcal{M}} \frac{\gamma_m^2\nu_{m} \{\sigma^2 + \tau_m(x^m)\}}{\mu^2_{0,m}f_{X^m}(x^m) n_0 h^{d_m}} +  \sum_{m\in \mathcal{M}} \frac{\gamma_m^2\nu_{m} \{\sigma^2 + \tau_m(x^m)\}}{\mu^2_{0,m}f_{X^m}(x^m) n_m h^{d_m}}.
\\ & \hspace{120pt} - \sum_{m_1 \neq m_2 \in \mathcal{M}}\frac{ \gamma_{m_1}\gamma_{m_2} \{\sigma^2 +  \tau_{m_1,m_2}(x^{m_{1,2}})\} \nu_{m_{1},m_{2}} f_{X^{m_{1,2}}}(x^{m_{1,2}})} {\mu_{0,m_{1}}f_{X^{m_{1}}}(x^{m_{1}})\mu_{0,m_{2}}f_{X^{m_{2}}}(x^{m_{2}}) n_0 h^{d_{m^{1,2}}}}.
\end{align*} 
We now bound $R_{21}$, $R_{22}$ and $R_{23}$ in turn. 

\textit{To bound $R_{21}$:} Write $E_m =(\eta_{m}(X_1^{m}), \ldots \eta_{m}(X_n^{m}))^T$ and let $\delta_m = E - E_m$. Then $H_m\delta_m = 0$ and
\begin{align}
\nonumber \gamma^T \bigl(H_{0,\mathcal{M}}^T - H_{\mathcal{M}}^T\bigr)E & = \sum_{m\in \mathcal{M}} \gamma_m\bigl(H_{0,m}^T - H_{m}^T\bigr) E 
\\ \nonumber & = \sum_{m\in \mathcal{M}} \gamma_m\bigl(H_{0,m}^T - H_{m}^T\bigr) (E - E_m + E_m)
\\ & = \sum_{m\in \mathcal{M}} \gamma_m\bigl\{H_{0,m}^T\delta_m + (H_{0,m}^T - H_{m}^T)E_m\bigr\}.\label{eq:001}
\end{align}
For $m \in \mathcal{M}$, using the fact $H_{0,m}^T1_n = H_{m}^T1_n= 1$ we have that
\begin{align*}
\label{eq:by1}
&|(H_{0,m}^T - H_{m}^T)E_m|  = |(H_{0,m}^T - H_{m}^T)\{E_m - \eta_m(x^m) 1_n \}|
\\ & = \Bigl| \frac{\frac{1}{n_0h^{d_m}} \sum_{i\in A_0} K_m\bigl(\frac{X^{m}_i - x^{m}}{h}\bigr)\{\eta_m(X_i^m) - \eta_m(x^m)\}}{\frac{1}{n_0h^{d_m}} \sum_{i \in A_0} K_m\bigl(\frac{X^{m}_i - x^{m}}{h}\bigr)}  
\\& \hspace{120pt}  -  \frac{\frac{1}{n_mh^{d_m}} \sum_{i\in A_m} K_m\bigl(\frac{X^{m}_i - x^{m}}{h}\bigr) \{\eta_m(X_i^m) - \eta_m(x^m)\}}{\frac{1}{n_mh^{d_m}}\sum_{i \in A_m} K_m\bigl(\frac{X^{m}_i - x^{m}}{h}\bigr)}\Bigr|
\\ & = \frac{1}{\hat{f}_{0,m} \hat{f}_m n_0n_mh^{2d_m}} \Bigl|\sum_{i\in A_0, j \in A_m} \Bigl[K_m\Bigl(\frac{X^{m}_i - x^{m}}{h}\Bigr)\{\eta_m(X_i^m) - \eta_m(x^m)\}  K_m\Bigl(\frac{X^{m}_j - x^{m}}{h}\Bigr)
\\ & \hspace{120pt} - K_m \Bigl(\frac{X^{m}_j - x^{m}}{h}\Bigr) \{\eta_m(X_j^m) - \eta_m(x^m)\}  K_m\Bigl(\frac{X^{m}_i - x^{m}}{h}\Bigr) \Bigr] \Bigr|,
\end{align*}
The last line in the display above is the sum of $n_0n_m$ mean zero terms. Moreover, for each term we have
\begin{align*}
&\mathbb{E}\Bigl[K_m\bigl(\frac{X^{m}_i - x^{m}}{h}\bigr)\{\eta_m(X_i^m) - \eta_m(x^m)\}  K_m\bigl(\frac{X^{m}_j - x^{m}}{h}\bigr)
\\ & \hspace{120pt} - K_m \bigl(\frac{X^{m}_j - x^{m}}{h}\bigr) \{\eta_m(X_j^m) - \eta_m(x^m)\}  K_m\bigl(\frac{X^{m}_i - x^{m}}{h}\bigr) \Bigr]^2 
\\ & = \int_{\mathbb{R}^{d_m}} \int_{\mathbb{R}^{d_m}}  \Bigl[K_m\bigl(\frac{z_1 - x^{m}}{h}\bigr)\{\eta_m(z_1) - \eta_m(x^m)\}  K_m\bigl(\frac{z_2 - x^{m}}{h}\bigr)
\\ & \hspace{60pt} - K_m \bigl(\frac{z_2 - x^{m}}{h}\bigr) \{\eta_m(z_2) - \eta_m(x^m)\}  K_m\bigl(\frac{z_1 - x^{m}}{h}\bigr) \Bigr]^2  f_{X^m}(z_1^m) f_{X^m}(z_2^m) \, dz_1 dz_2
\\ & = 2 \int_{\mathbb{R}^{d_m}} \int_{\mathbb{R}^{d_m}}  \Bigl[K_m^2\bigl(\frac{z_1 - x^{m}}{h}\bigr)\{\eta_m(z_1) - \eta_m(x^m)\}^2 K_m^2\bigl(\frac{z_2 - x^{m}}{h}\bigr)  f_{X^m}(z_1^m) f_{X^m}(z_2^m)\, dz_1 dz_2
\\ & \hspace{60pt} - 2 \int_{\mathbb{R}^{d_m}} \int_{\mathbb{R}^{d_m}}  K_m^2 \bigl(\frac{z_2 - x^{m}}{h}\bigr) \{\eta_m(z_1) - \eta_m(x^m)\}
\\ & \hspace{120pt}  \{\eta_m(z_2) - \eta_m(x^m)\}  K_m^2\bigl(\frac{z_1 - x^{m}}{h}\bigr)   f_{X^m}(z_1^m) f_{X^m}(z_2^m) \, dz_1 dz_2
\\ & \leq 2 L^2 h^{2d_m + 2} \int_{\mathbb{R}^{d_m}} K_m^2(u) \|u\|^2  f_{X^m}(x^m + hu)  \, du \int_{\mathbb{R}^{d_m}} K_m^2(u)  f_{X^m}(x^m + hu)  \, dz_u 
\\ & \hspace{120 pt} + 2 L^2 h^{2d_m + 2}  \Bigl\{\int_{\mathbb{R}^{d_m}} K_m(u) \|u\| f_{X^m}(x^m + hu) \, du \Bigr\}^2 
\\ & \leq  2 L^2 h^{2d_m + 2} \{f_{X^m}(x^m) \mu_{1,m}\bar{K}_m + Lh\mu_{2,m}\bar{K}_m  + (f_{X^m}(x^m) \mu_{1,m} + Lh  \mu_{1,m}\bar{K}_m)^2\}.
\end{align*} 
Therefore, by Markov's inequality we obtain that 
\begin{align*} 
&\frac{1}{n_0n_mh^{2d_m}} \Bigl|\sum_{i\in A_0, j \in A_m} \Bigl[K_m\bigl(\frac{X^{m}_i - x^{m}}{h}\bigr)\{\eta_m(X_i^m) - \eta_m(x^m)\}  K_m\bigl(\frac{X^{m}_j - x^{m}}{h}\bigr)
\\ & \hspace{30pt} - K_m \bigl(\frac{X^{m}_j - x^{m}}{h}\bigr) \{\eta_m(X_j^m) - \eta_m(x^m)\}  K_m\bigl(\frac{X^{m}_i - x^{m}}{h}\bigr) \Bigr] \Bigr|
= O_p\Bigl(\frac{1}{n_0^{1/2}n_m^{1/2} h^{d_m-1}}\Bigr), 
\end{align*} 
uniformly for $h \in [n^{-\beta}, n^{-\alpha}]$.   Using also the fact that $\hat{f}_{0,m} = f_{X^m}(x^m) + O_p(h)$ and $\hat{f}_{m} = f_{X^m}(x^m) + O_p(h)$, we have by Slutsky's Theorem that
\begin{equation}
\label{eq:R21one}
|(H_{0,m}^T - H_{m}^T)E_m| = O_p\Bigl(\frac{1}{n_0^{1/2}n_m^{1/2} h^{d_m-1}}\Bigr),
\end{equation}
uniformly for $h \in [n^{-\beta}, n^{-\alpha}]$.  Now, for the terms involving $\delta_m$ in \eqref{eq:001}, first write  
\begin{align*} 
H_{0,m}^T\delta_m &=  \frac{\frac{1}{n_0h^{d_m}} \sum_{i\in A_0} K_m\bigl(\frac{X^{m}_i - x^{m}}{h}\bigr)\{\eta(X_i) - \eta_m(X_i^m)\}} {\frac{1}{n_0h^{d_m}} \sum_{i \in A_0} K_m\bigl(\frac{X^{m}_i - x^{m}}{h}\bigr)} 
\\ & = \frac{1}{\hat{f}_{0,m} n_0h^{d_m}} \sum_{i\in A_0} K_m\bigl(\frac{X^{m}_i - x^{m}}{h}\bigr)\{\eta(X_i) - \eta_m(X_i^m)\}.
\end{align*} 
The last term is a sum of $n_0$ mean zero terms, with
\begin{align*} 
 \mathbb{E} \Bigl[K_m^2\bigl(\frac{X^{m}_i - x^{m}}{h}\bigr)  \{\eta(X_i) - \eta_m(X_i^m)\}^2\Bigr] & =  \mathbb{E} \Big\{K_m^2\bigl(\frac{X^{m}_i - x^{m}}{h}\bigr) \tau_m(X_i^m) \Bigr\} 
\\ & =  \int_{\mathbb{R}^{d_m}} K_m^2\bigl(\frac{z^{m} - x^{m}}{h}\bigr) \tau_m(z^m) f_{X^m}(z^m) \, dz 
\\ & = h^{d_m}  \int_{\mathbb{R}^{d_m}} K_m^2(u) \tau_m(x^m + hu) f_{X^m}(x^m + hu) \, du  
\\ & \hspace{ -120 pt} \leq  h^{d_m} [\tau_m(x^m) f_{X^m} (x^m)  +   Lh \bar{K}_m \mu_{0,m} \{\tau_m(x^m)  + f_{X^m} (x^m)\} + L^2 h^2 \bar{K}_m \mu_{1,m}].
\end{align*} 
Thus by Markov's inequality we have $|H_{0,m}^T\delta_m| = O_p(n_0^{-1/2} h^{-d_m/2} )$, uniformly for $h \in [n^{-\beta}, n^{-\alpha}]$, and it follows immediately that $|(H_{0,m}^T-H_m^T)E_m| = O_p\Bigl(h (H_{0,m}^T\delta_m)^2\Bigr)$, uniformly for $h \in [n^{-\beta}, n^{-\alpha}]$.  Thus, the first term in $R_{21}$ can be written as
\[
\{\gamma^T \bigl(H_{0,\mathcal{M}}^T - H_{\mathcal{M}}^T\bigr)E\}^2 = \Bigl(\sum_{m\in \mathcal{M}} \gamma_m H_{0,m}^T\delta_m\Bigr)^2 \{1 +O_p(h)\},
\]
uniformly for $h \in [n^{-\beta}, n^{-\alpha}]$.   Now observe that
\begin{align*}
& \mathbb{E}\{\sum_{m\in \mathcal{M}} \gamma_m^2 (H_{0,m}^T\delta_m)^2 | X_1^m, \ldots, X_n^m\}  - \sum_{m\in \mathcal{M}} \frac{\gamma_m^2\nu_{m} \tau_m(x^m)}{\mu^2_{0,m}f_{X^m}(x^m) n_0 h^{d_m}} 
\\ & = \sum_{m\in \mathcal{M}} \gamma_m^2 H_{0,m}^T\mathbb{E}\{\delta_m \delta_m^T  | X_1^m, \ldots, X_n^m\} H_{0,m}   - \sum_{m\in \mathcal{M}} \frac{\gamma_m^2\nu_{m} \tau_m(x^m)}{\mu^2_{0,m}f_{X^m}(x^m) n_0 h^{d_m}} 
\\ & = \sum_{m\in \mathcal{M}} \gamma_m^2 H_{0,m}^T \mathrm{diag}\{\tau_m(X_1^m), \ldots, \tau_m(X_n^m)\} H_{0,m}  - \sum_{m\in \mathcal{M}} \frac{\gamma_m^2\nu_{m} \tau_m(x^m)}{\mu^2_{0,m}f_{X^m}(x^m) n_0 h^{d_m}} 
\\ & = \sum_{m\in \mathcal{M}} \gamma_m^2 H_{0,m}^T \mathrm{diag}\{\tau_m(X_1^m), \ldots, \tau_m(X_n^m)\} H_{0,m} - \sum_{m\in \mathcal{M}} \gamma_m^2 \tau_{m}(x^m) H_{0,m}^T H_{0,m}  
\\ & \hspace{120pt} +  \sum_{m\in \mathcal{M}} \gamma_m^2 \tau_{m}(x^m) H_{0,m}^T H_{0,m}  - \sum_{m\in \mathcal{M}} \frac{\gamma_m^2\nu_{m} \tau_m(x^m)}{\mu^2_{0,m}f_{X^m}(x^m) n_0 h^{d_m}}.
\end{align*}
Then, by assumption \textbf{A4}, we have that
\begin{align}
\label{eq:by2}
&\Bigl| \sum_{m\in \mathcal{M}} \gamma_m^2 H_{0,m}^T \mathrm{diag}\{\tau_m(X_1^m), \ldots, \tau_m(X_n^m)\} H_{0,m} - \sum_{m\in \mathcal{M}} \gamma_m^2 \tau_{m}(x^m) H_{0,m}^T H_{0,m}  \Bigr|\nonumber
\\  &=\Bigl| \sum_{m\in \mathcal{M}} \gamma_m^2 H_{0,m}^T \mathrm{diag} \{\tau_m(X_i^m) - \tau_m(x^m)\} H_{0,m} \Bigr| \nonumber
\\ & \leq \sum_{m\in \mathcal{M}} \gamma_m^2 \frac{\sum_{i \in A_0} K_m^2\bigl(\frac{X^{m}_i - x^{m}}{h}\bigr)|\tau_m(X_i^m) - \tau_m(x^m)|}{ \{\sum_{i \in A_0} K_m\bigl(\frac{X^{m}_i - x^{m}}{h}\bigr)\}^2 } \nonumber
\\ &  \leq L  \sum_{m\in \mathcal{M}} \gamma_m^2 \frac{\frac{1}{n_0^2 h^{2d_m}} \sum_{i \in A_0} K_m^2\bigl(\frac{X^{m}_i - x^{m}}{h}\bigr) \| X_i^m - x^m\| }{\hat{f}_{0,m}} = O_p\Bigl( \sum_{m\in \mathcal{M}} \gamma_m^2 \frac{1}{n_0h^{d_m -1}}\Bigr),
\end{align}
uniformly for $h \in [n^{-\beta}, n^{-\alpha}]$.  Moreover, by \eqref{eq:HH2}, we have   
\begin{equation}
\label{eq:R21two}
\Bigl| \sum_{m\in \mathcal{M}} \gamma_m^2 \tau_{m}(x^m) H_{0,m}^T H_{0,m}  - \sum_{m\in \mathcal{M}} \frac{\gamma_m^2\nu_{m} \tau_m(x^m)}{\mu^2_{0,m}f_{X^m}(x^m) n_0 h^{d_m}} \Bigr| = O_p\Bigl( \sum_{m\in \mathcal{M}} \gamma_m^2 \frac{1}{n_0h^{d_m -1}}\Bigr),
 \end{equation}
 uniformly for $h \in [n^{-\beta}, n^{-\alpha}]$.  Similarly
\begin{align*}
& \mathbb{E}\Bigl\{\sum_{m_1 \neq m_2 \in \mathcal{M}} \gamma_{m_1}\gamma_{m_2} H_{0,m_1}^T\delta_{m_1} H_{0, m_2}^T\delta_{m_2}\Big| X_1^{m_{1,2}}, \ldots, X_n^{m_{1,2}}\Bigr\} 
\\ & \hspace{60pt} - \sum_{m_1 \neq m_2 \in \mathcal{M}}\frac{ \gamma_{m_1}\gamma_{m_2}  \tau_{m_1,m_2}(x^{m_{1,2}}) \nu_{m_{1},m_{2}} f_{X^{m_{1,2}}}(x^{m_{1,2}})} {\mu_{0,m_{1}}f_{X^{m_{1}}}(x^{m_{1}})\mu_{0,m_{2}}f_{X^{m_{2}}}(x^{m_{2}}) n_0 h^{d_{m^{1,2}}}};
\\& \hspace{30pt}  = \sum_{m_1 \neq m_2 \in \mathcal{M}} \gamma_{m_1}\gamma_{m_2} H_{0,m_1}^T\biggl[\mathbb{E}\Bigl\{ \delta_{m_1} \delta_{m_2}^T \Big| X_1^{m_{1,2}}, \ldots, X_n^{m_{1,2}}\Bigr\} - \tau_{m_1, m_2}(x^{m_{1,2}}) I_{n \times n} \biggr]  H_{0, m_2} 
\\ &  \hspace{90pt} +  \sum_{m_1 \neq m_2 \in \mathcal{M}} \gamma_{m_1}\gamma_{m_2} H_{0,m_1}^T H_{0, m_2} \tau_{m_1, m_2}(x^{m_{1,2}})  
\\ & \hspace{120pt} - \sum_{m_1 \neq m_2 \in \mathcal{M}}\frac{ \gamma_{m_1}\gamma_{m_2}  \tau_{m_1,m_2}(x^{m_{1,2}}) \nu_{m_{1},m_{2}} f_{X^{m_{1,2}}}(x^{m_{1,2}})} {\mu_{0,m_{1}}f_{X^{m_{1}}}(x^{m_{1}})\mu_{0,m_{2}}f_{X^{m_{2}}}(x^{m_{2}}) n_0 h^{d_{m^{1,2}}}}.
\end{align*} 
Now, by the last part of Assumption \textbf{A4} and similar arguments to those used above, we have
\begin{align}
\label{eq:by3} 
&\Bigl| \sum_{m_1 \neq m_2 \in \mathcal{M} }  H_{0,m_1}^T\biggl[\mathbb{E}\Bigl\{ \delta_{m_1} \delta_{m_2}^T \Big| X_1^{m_{1,2}}, \ldots, X_n^{m_{1,2}}\Bigr\} - \tau_{m_1, m_2}(x^{m_{1,2}}) I_{n \times n} \biggr] H_{0, m_2} \Bigr| \nonumber
\\ & \hspace{120pt}  =  O_p\Bigl(\sum_{m_1 \neq m_2 \in \mathcal{M} } \gamma_{m_1} \gamma_{m_2}  \frac{1}{n_0 h^{d_{m^{1,2}} -1}}\Bigr),
\end{align}
uniformly for $h \in [n^{-\beta}, n^{-\alpha}]$.  Furthermore, by \eqref{eq:HH4}, we have 
\begin{align} 
\label{eq:R21three}
& \Bigl|\sum_{m_1 \neq m_2 \in \mathcal{M}} \! \! \gamma_{m_1}\gamma_{m_2} H_{0,m_1}^T H_{0, m_2} \tau_{m_1, m_2}(x^{m_{1,2}})  - \! \! \! \sum_{m_1 \neq m_2 \in \mathcal{M}}\frac{ \gamma_{m_1}\gamma_{m_2}  \tau_{m_1,m_2}(x^{m_{1,2}}) \nu_{m_{1},m_{2}} f_{X^{m_{1,2}}}(x^{m_{1,2}})} {\mu_{0,m_{1}}f_{X^{m_{1}}}(x^{m_{1}})\mu_{0,m_{2}}f_{X^{m_{2}}}(x^{m_{2}}) n_0 h^{d_{m^{1,2}}}}\Bigr| \nonumber
\\ & \hspace{150pt} = O_p\Bigl( \sum_{m_1 \neq m_2 \in \mathcal{M}} \frac{1} {n_0 h^{d_{m^{1,2}}-1}}\Bigr),
\end{align} 
uniformly for $h \in [n^{-\beta}, n^{-\alpha}]$.  Then, using \eqref{eq:R21one}, \eqref{eq:by2}, \eqref{eq:R21two}, \eqref{eq:by3} and \eqref{eq:R21three}, we conclude that 
\[
|R_{21}| = O_p\Bigl(\sum_{m \in \mathcal{M}} \frac{1}{n_0 h^{d_m-1}} +  \sum_{m_1 \neq m_2 \in \mathcal{M}} \frac{1} {n_0 h^{d_{m^{1,2}}-1}}\Bigr),
\]
uniformly for $h \in [n^{-\beta}, n^{-\alpha}]$. 

Furthermore, by \eqref{eq:HH2} and \eqref{eq:HH4}, we have that 
\[
|R_{22}| = O_p \Bigl(\sum_{m \in \mathcal{M}} \frac{1}{n_0 h^{d_m-1}} +  \sum_{m_1 \neq m_2 \in \mathcal{M}} \frac{1} {n_0 h^{d_{m^{1,2}}-1}}\Bigr),
\]
uniformly for $h \in [n^{-\beta}, n^{-\alpha}]$. Finally,  by \eqref{eq:HH3} we have 
\[
|R_{23}| = O_p \Bigl(\sum_{m \in \mathcal{M}} \frac{1}{n_m h^{d_m-1}} \Bigr),
\]
uniformly for $h \in [n^{-\beta}, n^{-\alpha}]$.  This concludes the bound on $R_2$ and Part I of the proof. 
 
\textbf{Part II: the second term in \eqref{eq: MSE-all}}. Using the definition of local linear estimators and noting that $\Gamma H_0 = 0$ and $H_{\mathcal M}^TH_0=0$, the second term in \eqref{eq: MSE-all} can be written as
\begin{align}
\label{eq:mse2}
& \mathbb{E}[(\hat{\eta}_\gamma^{\mathcal{M}} - \hat{\eta}_{0})\{\hat{\eta}_{0} - \eta(x)\} | X_1^{m_1}, \ldots, X_n^{m_n}] \nonumber
\\ &\hspace{5pt}  = \mathbb{E}\Bigl[(\hat{\eta}_\gamma^{\mathcal{M}} - \hat{\eta}_{0})\bigl\{\hat{\eta}_{0}- \mathbb{E}(\hat{\eta}_{0}|  X_1^{m_1}, \ldots, X_n^{m_n})+ \mathbb{E}(\hat{\eta}_{0}|  X_1^{m_1}, \ldots, X_n^{m_n})  -  \eta(x)\bigr\} \big| X_1^{m_1}, \ldots, X_n^{m_n}\Bigr] \nonumber
\\ & \hspace{5 pt} = -\gamma^T\bigl(H_{0,\mathcal{M}}^T - H_{\mathcal{M}}^T\bigr)  \mathbb{E}\{Y_{[n]}(Y_{[n]}^T - E^T)| X_1^{m_1}, \ldots, X_n^{m_n} \} H_0  \nonumber
\\ & \hspace{150 pt} - \gamma^T\bigl(H_{0,\mathcal{M}}^T - H_{\mathcal{M}}^T\bigr) E\{E^TH_0 - \eta(x)\} \nonumber
\\ & \hspace{5 pt} = -\gamma^T\bigl(H_{0,\mathcal{M}}^T - H_{\mathcal{M}}^T\bigr) (\sigma^2 I + \Gamma) H_0   - \gamma^T\bigl(H_{0,\mathcal{M}}^T - H_{\mathcal{M}}^T\bigr) E\{E^TH_0 - \eta(x)\}\nonumber
\\ & \hspace{5 pt} = -\sigma^2\gamma^TH_{0,\mathcal{M}}^T H_0-\gamma^T\bigl(H_{0,\mathcal{M}}^T - H_{\mathcal{M}}^T\bigr) E\{E^TH_0 - \eta(x)\}.
\end{align}

For the first term in \eqref{eq:mse2}, by \eqref{eq:HH1} we have 
\[
\sigma^2 \gamma^T H_{0,\mathcal{M}}^T H_{0}  =  \sigma^2 \gamma^T\Omega_{\mathrm{R}} \{1+O_p(h)\},
\]
uniformly for $h \in [n^{-\beta}, n^{-\alpha}]$.

It remains to bound the second term. Let 
\[
R_3 = -\gamma^T(H_{0,\mathcal M}-H_{\mathcal M})E\{E^TH_0-\eta(x)\}.    
\]
Recall from Part I of this proof that $\gamma^T(H_{0,\mathcal M}-H_{\mathcal M})E = O_p\bigl(\sum_{m \in \mathcal{M}} \frac{\gamma_m}{n_0 h^{d_{m}}} \bigr)$. Moreover, by assumption \textbf{A3}, we have that
\[
|\eta(X_i) - \eta(x)| \leq L \|X_i - x\|,
\]
for $i = 1, \ldots, n$. Recall that $E^{*T}H_0 = \eta(x)$. Therefore
\begin{align*}
|E^TH_0 - \eta(x)| &= |(E^T - E^{*T})H_0|  \leq \frac{\sum_{i \in A_0} K\bigl(\frac{X_i - x}{h}\bigr) |\eta(X) - \eta(x)|}{\sum_{i \in A_0} K\bigl(\frac{X_i - x}{h}\bigr)} 
\\ & \leq \frac{1}{\hat{f}_0 n_0 h^d} \sum_{i \in A_0} K\bigl(\frac{X_i - x}{h}\bigr) |\eta(X) - \eta(x)| = O_p(h),
\end{align*} 
uniformly for $h \in [n^{-\beta}, n^{-\alpha}]$, where the last equality is by Markov's inequality due to the facts that $K\bigl(\frac{X_i - x}{h}\bigr) |\eta(X) - \eta(x)|>0 $ and
\[ 
\mathbb{E}\bigl\{ K\bigl(\frac{X_i - x}{h}\bigr) |\eta(X) - \eta(x)| \bigr\} \leq L h^{d+1}  \int_{\mathbb{R}^d} K(u) \|u\| f_{X}(x + hu)  \, du \leq L h^{d+1} \{\mu_1 f_{X}(x) + L h \mu_2 \}.
\]
This completes the proof. 

\end{proof}

\bigskip

 \subsection{Proofs for the results in Section \ref{sec:discuss}}
 \begin{proof}[Proof of Corollary~\ref{cor:msered}]  
 	Consider the bias and variance separately: we claim
 	\begin{equation}
 	\label{eq:bias}
 	\quad \mathbb{E} (\hat{\theta}^*_{m}) = \mathbb{E} (\hat{\theta}_{0})
 	\end{equation}
 	and 
 	\begin{equation}
 	\label{eq:var}
 	\quad \mathrm{Var}(\hat{\theta}^*_{m}) =  \mathrm{Var} (\hat{\theta}_{0})  - \frac{\mathrm{Cov}(\hat{\theta}_{0},\bar{\theta}_{A_0,m})^2}{\mathrm{Var}(\bar{\theta}_{A_0,m} - \bar{\theta}_{\bar{A}_m,m})}  \leq  \mathrm{Var} (\hat{\theta}_{0}).
 	\end{equation}
 	To see~\eqref{eq:bias} it suffices to show that $\mathbb{E}(\bar{\theta}_{A_0,m}) = \mathbb{E}(\bar{\theta}_{\bar{A}_m,m})$.  First, since the missing data is MCAR, we have that the data in $\mathcal{T}_{A_0,m} \cup \mathcal{T}_{\bar{A}_m,m}$ are independent and identically distributed with distribution $Q_m$.  Furthermore, for each $b_1$, $b_2$, the estimators $\hat{\theta}_{\tilde{A}_0^{b_1},m}$ and $\hat{\theta}_{\tilde{A}_m^{b_2},m}$  are constructed using $n_{0,m}$ pairs from $\mathcal{T}_{A_0,m}$ and  $\mathcal{T}_{\bar{A}_m,m}$, respectively.   Thus, these two estimators have the same distribution, and in particular, they have the same mean.  It follows that  $\mathbb{E}(\bar{\theta}_{A_0,m}) = \mathbb{E}(\bar{\theta}_{\bar{A}_m,m})$.    The result in~\eqref{eq:var} follows via a direct calculation.
 \end{proof}

\bigskip 

\subsection{Auxilliary Lemmas} 
\label{sec:Kernelchoice} 
The results in this section motivate our choice of the $d_m$-dimensional kernel, $K_m$, used to construct $\hat{f}_{0,m}$ and $\hat{f}_m$ in the density estimation problem, and $\hat{\eta}_{0,m}$ and $\hat{\eta}_m$ in the regression problem.  For $m \in \{0,1\}^d$, let $m^c = (1,\ldots,1)^T - m$.  Recall that we write $f_{X|X^m}(x^{m^c};x^m)$, for the conditional density of $X$ given $X^m = x^m$ at $x^{m^c}$.  

We see from Lemma \ref{lem:Kernel} that, up to rescaling, if the kernel can be factorised as $K(t) = K_m(t^m)K_{m^c}(t^{m^c})$, then the optimal choice of $\hat{f}_{0,m}$ is well approximated by 
\[
\frac{1}{n_0 h^{d_m}}\sum_{i \in A_0} K_m\Bigl(\frac{X_i^m - x^m}{h}\Bigr). 
\] 
On the other hand, Lemma \ref{lem:KernelReg} shows that the optimal choice of $\hat{\eta}_{0,m}$ is well approximated by  
\[
\frac{\sum_{i \in A_0} Y_i K_m\bigl(\frac{X_i^m - x^m}{h}\bigr)}{\sum_{i \in A_0} K_m\bigl(\frac{X_i^m - x^m}{h}\bigr)}. 
\] 

Any $d$-dimensional kernel constructed from the product of 1-dimensional kernels  satisfy the factor assumption in Lemmas~\ref{lem:Kernel} and~\ref{lem:KernelReg} .  For example the condition is satisfied by the Gaussian kernel $\frac{1}{(2\pi)^{d/2}}\exp(-\|t\|^2/2)$ and the box kernel $\frac{1}{2^d} \mathbbm{1}_{\{\max_{j = 1, \ldots, d} |t_j| \leq 1\}}$. Moreover, other kernels, such as $\mathbbm{1}_{\{\|t\|^2 \leq 1\}}$, are not covered by the lemma but enjoy similar properties to that in \eqref{eq:opt-kernel}, which can be proved by using similar ideas.  

\begin{lemma}
	\label{lem:Kernel} 
	Assume \textbf{A1} and suppose that, for $t \in \mathbb{R}^d$ and $m \in \{0,1\}^d$, we have $K(t) = K_m(t^m)K_{m^c}(t^{m^c})$, for some $K_m : \mathbb{R}^{d_m} \rightarrow [0,\infty)$ and $K_{m^c} : \mathbb{R}^{d - d_m} \rightarrow [0,\infty)$, that satisfy 
	$\mu_{0,m} = \int_{\mathbb{R}^{d_{m}}} K_{m}(z) \, dz < \infty$, $\mu_{1,m} = \int_{\mathbb{R}^{d_{m}}} \|z\| K_{m}(z) \, dz < \infty$, $\mu_{0,m^c} = \int_{\mathbb{R}^{d_{m^c}}} K_{m^c}(z) \, dz < \infty$	and $\mu_{1,m^c} = \int_{\mathbb{R}^{d_{m^c}}} \|z\| K_{m^c}(z) \, dz < \infty$.  
	Then, for $z \in \mathbb{R}^{d_m}$,
	\begin{align} 
	\label{eq:opt-kernel}
	& \biggl|\mathbb{E}\Bigl\{K\Bigl(\frac{X-x}{h}\Bigr) \Big| X^m = z^m\Bigr\} - h^{d - d_m} K_m\Bigl(\frac{z^m - x^m}{h}\Bigr) f_{X|X^m}(x^{m^c};z^m) \mu_{0, m^c} \biggr|\nonumber 
	\\ & \hspace{250pt} \leq \frac{L h^{1 + d - d_m}}{f_{X^m}(z^m) } K_m\Bigl(\frac{z^m - x^m}{h}\Bigr) \mu_{1, m^c} .
	\end{align}
	Therefore, for $0 < \alpha < \beta < 1/d$, 
\[
\frac{1}{n_0 h^{d}}\sum_{i \in A_0}\mathbb{E}\Bigl\{K\Bigl(\frac{X_i-x}{h}\Bigr) \Big| X_i^m\Bigr\} = \frac{\mu_{0, m^c}}{n_0 h^{d_m}}\sum_{i \in A_0} K_m\Bigl(\frac{X_i^m - x^m}{h}\Bigr) f_{X|X^m}(x^{m^c};x^m) + O_p(h)
\]
as $n \rightarrow \infty$, uniformly for $h \in [n^{-\beta}, n^{-\alpha}]$.
\end{lemma}
\begin{proof}
To see \eqref{eq:opt-kernel}, first observe that by making the substitution $u = \frac{z^{m^c} - x^{m^c}}{h}$, we have
\begin{align*}
\mathbb{E}\Bigl\{K\Bigl(\frac{X-x}{h}\Bigr) \Big| X^m = z^m\Bigr\}& = \int_{\mathbb{R}^{d-d_m}} K_m\Bigl(\frac{z^m - x^m}{h}\Bigr) K_{m^c}\Bigl(\frac{z^{m^c} - x^{m^c}}{h} \Bigr) f_{X|X^m}(z^{m^c};z^m)\, dz^{m^c}
\\ & = h^{d - d_m} K_m\Bigl(\frac{z^m - x^m}{h}\Bigr) \int_{\mathbb{R}^{d-d_m}} K_{m^c}(u) f_{X|X^m}(x^{m^c} + h u;z^m)\, du.  
\end{align*} 
Now, write 
\[
f_{X|X^m}(z^{m^c};z^m) = \frac{f_X(z)}{f_{X^m}(z^m)}. 
\]
Thus, by assumption \textbf{A1}, we have that
\[
|f_{X|X^m}(x^{m^c} + hu;z^m) - f_{X|X^m}(x^{m^c};z^m)|  \leq \frac{Lh \|u\|}{f_{X^m}(z^m)}.  
\]
It follows that
	\begin{align*} 
	& \biggl|\mathbb{E}\Bigl\{K\Bigl(\frac{X-x}{h}\Bigr) \Big| X^m = z^m\Bigr\} - h^{d - d_m} K_m\Bigl(\frac{z^m - x^m}{h}\Bigr) f_{X|X^m}(x^{m^c};z^m) \mu_{0, m^c} \biggr|\nonumber 
	\\ & \hspace{180pt} \leq \frac{L h^{1 + d - d_m}}{f_{X^m}(z^m)} K_m\Bigl(\frac{z^m - x^m}{h}\Bigr) \int_{\mathbb{R}^{d - d_m}} K_{m^c}(u) \|u\| \, du.
	\end{align*}
	This proves \eqref{eq:opt-kernel}. 
	
	For the remainder of the proof, first observe that 
\[ 
\mathbb{E}\Bigl\{\frac{1}{f_{X^m}(X_i^m)} K_m\Bigl(\frac{X_i^m - x^m}{h}\Bigr)\Bigr\} = \int_{\mathbb{R}^{d_m}} K_m\Bigl(\frac{z^m - x^m}{h}\Bigr) \, dz = h^{d_m} \mu_{0,m}.
\]
It remains to bound the following:
\begin{align*}
&\mathbb{E}\Bigl|\frac{1}{n_0 h^{d_m}}\sum_{i \in A_0} K_m\Bigl(\frac{X_i^m - x^m}{h}\Bigr) \{f_{X|X^m}(x^{m^c};X_i^m) - f_{X|X^m}(x^{m^c};x^m)\}\Bigr|
\\ & = \frac{1}{h^{d_m}}\int_{\mathbb{R}^{d_m}}  \Bigl| K_m\Bigl(\frac{z^m-x^m}{h}\Bigr)\{f_{X|X^m}(x^{m^c};z^m) - f_{X|X^m}(x^{m^c};x^m)\} \Bigr| f_{X^m}(z^m) \, dz^m
\\ & = \int_{\mathbb{R}^{d_m}} K_m(u) | f_{X|X^m}(x^{m^c};x^m + hu) - f_{X|X^m}(x^{m^c};x^m)|  f_{X^m}(x^m + hu) \, dz^m 
\\ & = \biggl|\int_{\mathbb{R}^{d_m}} K_m(u) \bigl|f_{X}(x^{m^c},x^m + hu) - \frac{f_{X}(x)f_{X^m}(x^m + hu)}{f_{X^m}(x^m)} \bigr| \, du 
\\ & \leq L h \Bigl\{1 + \frac{f_{X}(x)}{f_{X^m}(x^m)} \Bigr\}  \int_{\mathbb{R}^{d_m}} K_m(u) \|u\|\, du = L h \Bigl\{1 + \frac{f_{X}(x)}{f_{X^m}(x^m)} \Bigr\} \mu_{1,m}. 
\end{align*} 
The proof is completed using Markov's inequality. 
\end{proof}

\bigskip

Our results on the choice of $K_m$ in the regression problem need a slightly stronger condition on the joint distribution of $(X, Y)$: 
\begin{assumption} 
	Suppose that $P$ and $P_{m}$, for $m \in \{0,1\} ^d \setminus \{(1,\ldots,1)^T\}$, have densities $f_{X,Y}$, and $f_{X^m,Y}$, respectively, that satisfy $|f_{X,Y}(z_1,y) - f_{X,Y}(z_2,y)| \leq L\|z_1 - z_2\|$, for all $z_1,z_2 \in \mathbb{R}^{d}, y \in \mathbb{R}$, and, for each $m \in \{0,1\} ^d \setminus \{(1,\ldots,1)^T\}$, we have  $|f_{X^m,Y}(z_1^m,y) - f_{X^m,Y}(z_2^m,y)| \leq L \|z_1^m - z_2^m\|$, for all $z_1^m,z_2^m \in \mathbb{R}^{d_m}, y\in \mathbb{R}$.  Moreover, we ask that the densities $f_{X}$, $f_{X^m,Y}$ are bounded, and that $f_{X^{m^c} | X^m,Y}(x^m ; z^m, y)$ is a bounded function of $z^m$ and $y$.
\end{assumption} 
\bigskip 
\begin{lemma} 
\label{lem:KernelReg}
Assume \textbf{A3} and \textbf{A5}, and suppose that $Y$ is supported on $\mathcal{D}_{Y} \subseteq [-\bar{Y}, \bar{Y}]$, for some $\bar{Y} > 0$.  Suppose further that for $t \in \mathbb{R}^d$ and $m \in \{0,1\}^d$, we have $K(t) = K_m(t^m)K_{m^c}(t^{m^c})$, for some $K_m : \mathbb{R}^{d_m} \rightarrow [0,\infty)$ and $K_{m^c} : \mathbb{R}^{d - d_m} \rightarrow [0,\infty)$, that satisfy $\bar{K}_m := \sup_{z \in \mathbb{R}^{d_m}} (1 + \|z\|) K_m(z) < \infty$,
	$\mu_{0,m} = \int_{\mathbb{R}^{d_{m}}} K_{m}(z) \, dz < \infty$, $\mu_{1,m} = \int_{\mathbb{R}^{d_{m}}} \|z\| K_{m}(z) \, dz < \infty$, $\mu_{0,m^c} = \int_{\mathbb{R}^{d_{m^c}}} K_{m^c}(z) \, dz < \infty$	and $\mu_{1,m^c} = \int_{\mathbb{R}^{d_{m^c}}} \|z\| K_{m^c}(z) \, dz < \infty$.  
Then, for  $0 < \alpha < \beta < 1/d$,
\[
\sum_{i \in A_0}Y_i \mathbb{E}\Bigl\{\frac{K\bigl(\frac{X_i-x}{h}\bigr)}{\sum_{j \in A_0}K\bigl(\frac{X_j-x}{h}\bigr) } \Big| \mathcal{T}_{A_0,m}\Bigr\} = \frac{\eta(x)}{\eta_m(x_m)}  \frac{\sum_{i \in A_0}Y_iK_m\bigl(\frac{X^m_i-x^m}{h}\bigr)}{\sum_{j \in A_0}K_m\bigl(\frac{X^m_j-x^m}{h}\bigr)}  + o_p(1)
\]
as $n \rightarrow \infty$, uniformly for $h \in [n^{-\beta}, n^{-\alpha}]$. 
\end{lemma}
\begin{proof} 
\textit{Part I}: We first show that 
\begin{equation}
\label{eq:claimI1}
\Pi_0 := \sum_{i \in A_0}Y_i \mathbb{E}\Bigl\{\frac{K\bigl(\frac{X_i-x}{h}\bigr)}{\sum_{j \in A_0}K\bigl(\frac{X_j-x}{h}\bigr) } \Big| \mathcal{T}_{A_0,m}\Bigr\} - \sum_{i \in A_0}Y_i \frac{\mathbb{E}\bigl\{K\bigl(\frac{X_i-x}{h}\bigr) \big| \mathcal{T}_{A_0,m}\bigr\}}{\mathbb{E} \bigl\{\sum_{j \in A_0}K\bigl(\frac{X_j-x}{h}\bigr) \big| \mathcal{T}_{A_0,m}\bigr\} } = o_p(1).
\end{equation}
To see \eqref{eq:claimI1}, let
\[
W_1 :=  \frac{1}{n_0h^d} \sum_{i \in A_0} Y_i \Bigl[K\Bigl(\frac{X_i-x}{h}\Bigr) -  \mathbb{E}\Bigl\{K\Bigl(\frac{X_i-x}{h}\Bigr) \Big| X_i^m, Y_i \Bigr\} \Bigr]
\]
and 
\[
W_2 := \frac{1}{n_0 h^d} \sum_{i \in A_0} \Big[K\Bigl(\frac{X_i-x}{h}\Bigr) -  \mathbb{E}\Bigl\{K\Bigl(\frac{X_i-x}{h}\Bigr) \Big| X_i^m, Y_i \Bigr\}\Big].
\]
Then we can write
\[
\hat{\eta}_0 =\frac{ \sum_{i \in A_0}Y_i \mathbb{E}\bigl\{K\bigl(\frac{X_i-x}{h}\bigr) \big| \mathcal{T}_{A_0,m}\bigr\}}{\mathbb{E} \bigl\{\sum_{j \in A_0}K\bigl(\frac{X_j-x}{h}\bigr) \big| \mathcal{T}_{A_0,m}\bigr\}} + \frac{W_1}{\hat{f}_0} - \frac{W_2}{\hat{f}_0 }  \frac{\sum_{i \in A_0}Y_i \mathbb{E}\bigl\{K\bigl(\frac{X_i-x}{h}\bigr) \big| \mathcal{T}_{A_0,m}\bigr\}}{\mathbb{E}(\hat{f}_0 | \mathcal{T}_{A_0,m} )}
\]
It follows that 
\[
\Pi_0 = \mathbb{E}\Bigl(\frac{W_1}{\hat{f}_0}  \Big| \mathcal{T}_{A_0,m}\Bigr) +  \mathbb{E}\Bigl(\frac{W_2}{\hat{f}_0} \Big| \mathcal{T}_{A_0,m} \Bigr)  \frac{ \frac{1}{n_0h^d}\sum_{i \in A_0}Y_i \mathbb{E}\bigl\{K\bigl(\frac{X_i-x}{h}\bigr) \big| \mathcal{T}_{A_0,m}\bigr\}}{\mathbb{E}(\hat{f}_{0} | \mathcal{T}_{A_0,m})}.
\]

Now, by Markov's inequality, we have
\[
\mathbb{P}(|W_1| > t | \mathcal{T}_{A_0,m}) \leq \frac{\mathbb{E}(W_1^2| \mathcal{T}_{A_0,m})}{t^2} = \frac{1}{n_0^2 h^{2d}t^2} \sum_{i \in A_0} Y_i^2 \var \Bigl\{K\Bigl(\frac{X_i-x}{h}\Bigr) \Big| X_i^m, Y_i\Bigr\}
\]
and 
\[
\mathbb{P}(|W_2| > t | \mathcal{T}_{A_0,m}) \leq \frac{\mathbb{E}(W_2^2| \mathcal{T}_{A_0,m})}{t^2} = \frac{1}{n_0^2 h^{2d}t^2} \sum_{i \in A_0} \var \Bigl\{K\Bigl(\frac{X_i-x}{h}\Bigr) \Big| X_i^m, Y_i\Bigr\}.
\]
Note further that, since $|Y_i| \leq \bar{Y}$, for each $i$, we have that $|\Pi_0| \leq 2\bar{Y}$.  Therefore, for $0 < t < \mathbb{E}(\hat{f}_0|\mathcal{T}_{A_0,m})/2$,
\begin{align}
\label{eq:tbound1}
\bigl|\mathbb{E}\bigl(\Pi_0 | \mathcal{T}_{A_0,m}\bigr)\bigr|  & \leq \mathbb{E}\bigl(| \Pi_0| \mathbbm{1}_{\{|W_1| < t\}} \mathbbm{1}_{\{|W_2| < t\}}  | \mathcal{T}_{A_0,m} \bigr) + \mathbb{E}\big\{ |\Pi_0| (\mathbbm{1}_{\{|W_1| > t\}} + \mathbbm{1}_{\{|W_2| > t\}} )  | \mathcal{T}_{A_0,m} \bigr\} \nonumber
\\ & \leq  \frac{2t}{ \mathbb{E}(\hat{f}_0|\mathcal{T}_{A_0,m})} + \frac{2t}{ \mathbb{E}(\hat{f}_0|\mathcal{T}_{A_0,m})}  \frac{ \frac{1}{n_0h^d}\sum_{i \in A_0}Y_i \mathbb{E}\bigl\{K\bigl(\frac{X_i-x}{h}\bigr) \big| \mathcal{T}_{A_0,m}\bigr\}}{\mathbb{E}(\hat{f}_{0} | \mathcal{T}_{A_0,m})}  \nonumber
\\ & \hspace{120pt} +  \frac{2\bar{Y} (1+\bar{Y}^2)}{n_0^2 h^{2d}t^2} \sum_{i \in A_0} \var \bigl\{K\bigl(\frac{X_i-x}{h}\bigr) \big| X_i^m, Y_i\bigr\}.
\end{align}

Now, using the tower property of expectation, we have 
\[
\mathbb{E} \bigl| \mathbb{E}(\hat{f}_0 | \mathcal{T}_{A_0,m}) - f_X(x) \bigr| \leq  \mathbb{E}\bigl| \hat{f}_0 - f_X(x) \bigr| \leq L h \mu_{1},
\]
by assumption \textbf{A1}. Thus,  we have $\mathbb{E}(\hat{f}_0 | \mathcal{T}_{A_0,m}) =  f_X(x) + O_p(h)$, by Markov's inequality.  Similarly 
\[
\frac{1}{n_0 h^d} \sum_{i \in A_0}Y_i \mathbb{E}\Bigl\{K\Bigl(\frac{X_i-x}{h}\Bigr) \Big| \mathcal{T}_{A_0,m}\Bigr\} = \eta(x) f_X(x) + O_p(h).
\]
Finally, using the facts that 
\[
\mathbb{E}\Bigl[\var \Bigl\{K\Bigl(\frac{X_i-x}{h}\Bigr) \Big| X_i^m, Y_i\Bigr\}\Bigr] \leq \var\Bigl\{K\Bigl(\frac{X_i-x}{h}\Bigr)\Bigr\} \leq h^d \{\nu f_X(x) +  Lh \bar{K} \mu_0\},
\]
and that the conditional variances are all non-negative, by applying Markov's inequality  we have 
\[
\frac{1}{n_0}\sum_{i\in A_0}\var \Bigl\{K\Bigl(\frac{X_i-x}{h}\Bigr) \Big| X_i^m, Y_i\Bigr\} = O_p(h^d),
\]
uniformly for $h \in [n^{-\beta}, n^{-\alpha}]$.  The proof of Part I is then completed by taking $t = t_n := \sqrt{\frac{\log{n_0}}{n_0h^d}}$ in \eqref{eq:tbound1} and using Markov's inequality. 
 
\bigskip

\textit{Part II:} We claim that 
\begin{equation}
\label{eq:claimII1}
\Pi_1 := \frac{1}{n_0 h^d} \sum_{i \in A_0} Y_i \Bigl[\mathbb{E}\Bigl\{K\bigl(\frac{X_i-x}{h}\bigr) \Big| \mathcal{T}_{A_0,m}\Bigr\} - \frac{\mu_{0,m^c} f_{X}(x) \eta(x)}{f_{X^m}(x^m)\eta_m(x^m)}  K_m\bigl(\frac{X^m_i-x^m}{h}\bigr) \Bigr] = o_p(1), 
\end{equation}
uniformly for $h \in [n^{-\beta}, n^{-\alpha}]$, and 
\begin{equation}
\label{eq:claimII2}
\Pi_2 = \frac{1}{n_0 h^d} \sum_{i \in A_0} \mathbb{E}\Bigl\{K\bigl(\frac{X_i-x}{h}\bigr) \Big| \mathcal{T}_{A_0,m}\Bigr\} - \frac{\mu_{0,m^c}f_{X}(x)}{f_{X^m}(x^m)n_0h^{d_m}}  \sum_{i \in A_0} K_m\bigl(\frac{X^m_i-x^m}{h}\bigr) = o_p(1), 
\end{equation}
uniformly for $h \in [n^{-\beta}, n^{-\alpha}]$.

To see \eqref{eq:claimII1}, first observe that, for all $z^m \in \mathbb{R}^{d_m}$ and $y \in \mathbb{R}$, we have 
\begin{align*} 
&\Bigl|\mathbb{E}\Big\{K_{m^c}\bigl(\frac{X_i^{m^c} - x^{m^c}}{h}\bigr) \Big| (X_i^m, Y_i) = (z^m, y)\Bigr\} -  h^{d_{m^c}}\mu_{0,m^c} f_{X^{m^c}| X^m,  Y} (x^{m^c}; z^m, y)\Bigr|
\\& \hspace{10pt} = \Bigl|\int_{\mathbb{R}^{d_{m^c}}} K_{m^c}\bigl(\frac{z^{m^c} - x^{m^c}}{h}\bigr) f_{X^{m^c} | X^m,Y}(z^{m^c}; z^m, y) \, dz^{m^c} -   h^{d_{m^c}}\mu_{0,m^c} f_{X^{m^c}| X^m,  Y} (x^{m^c}; z^m, y)\Bigr| 
\\ & \hspace{10pt}= h^{d_{m^c}} \Bigl| \int_{\mathbb{R}^{d_{m^c}}} K_{m^c}(u) f_{X^{m^c} | X^m,Y}(x^{m^c} + h u; z^m, y) \, du^{m^c} -  \mu_{0,m^c} f_{X^{m^c}| X^m,  Y} (x^{m^c}; z^m, y)\Bigr| 
\\ & \hspace{10pt}\leq \frac{L h^{d_{m^c} + 1} \mu_{1, m^c}}{ f_{X^m,Y}(z^m, y)}.
\end{align*}
Therefore 
\begin{align*} 
&\mathbb{E} \Bigl|\frac{1}{n_0 h^d} \sum_{i \in A_0}Y_i \Bigl[K_m\bigl(\frac{X_i^m - x^m}{h}\bigr) \mathbb{E}\bigl\{K_{m^c}\bigl(\frac{X_i^{m^c}-x^{m^c}}{h}\bigr) \big| X_i^m, Y_i\bigr\} 
\\ & \hspace{ 120pt} -  K_m\bigl(\frac{X^m_i-x^m}{h}\bigr) \mu_{0,m^c} f_{X^{m^c}| X^m,  Y} (x^{m^c}; X_i^m, Y_i)\Bigr] \Bigr| 
\\ & \leq  L h \mu_{1, m^c} \mathbb{E} \Bigl\{ \frac{1}{n_0 h^{d_m}} \sum_{i \in A_0} \frac{|Y_i| K_m\bigl(\frac{X_i^m - x^m}{h}\bigr)}{ f_{X^m,Y}(X_i^m, Y_i)} \Bigr\} 
\\ & = L h \mu_{1, m^c} \frac{1}{h^{d_m}} \int_{\mathbb{R}^{d_m} \times \mathcal{D}_Y} |y| K_m\bigl(\frac{z^m - x^m}{h}\bigr) \, dz dy  \leq L h \mu_{1, m^c} \mu_{0,m} \bar{Y}^2.
\end{align*} 
It follows that
\begin{equation}
\label{eq:term}
\Pi_1  = \frac{\mu_{0,m^c}}{n_0h^{d_m} }\sum_{i \in A_0} Y_i  K_m\bigl(\frac{X^m_i-x^m}{h}\bigr) \Bigl[f_{X^{m^c}| X^m,  Y} (x^{m^c}; X_i^m, Y_i)  -  \frac{f_{X}(x) \eta(x)}{f_{X^m}(x^m)\eta_m(x^m)} \Bigr] + R_4,
\end{equation}
where $R_4 = O_p(h)$, uniformly for $h \in [n^{-\beta}, n^{-\alpha}]$, by Markov's inequality.  Now, writing (with a slight abuse of notation) $\eta(z^{m^c}, z^m) := \mathbb{E}(Y | X^{m^c} = z^{m^c}, X^m = z^m)$, the first term in \eqref{eq:term} is the average of $n_0$ independent and identically distributed terms each with expectation
\begin{align*} 
&\frac{1}{h^{d_m}} \int_{\mathbb{R}^d \times \mathbb{R}} y K_m\bigl(\frac{z^m-x^m}{h}\bigr)\Bigl\{ f_{X^{m^c}| X^m,  Y} (x^{m^c}; z^m, y)  -  \frac{f_{X}(x)\eta(x)}{f_{X^m}(x^m)\eta_m(x^m)} \Bigr\} f_{X^m, Y}(z^m, y) \, dz^m dy
\\& = \frac{1}{h^{d_m}} \int_{\mathbb{R}^d \times \mathbb{R}} y K_m\bigl(\frac{z^m-x^m}{h}\bigr) \Bigl\{ f_{X^{m^c}, X^m,  Y} (x^{m^c}, z^m, y)  - \frac{f_{X}(x) \eta(x) f_{X^m, Y}(z^m, y) }{f_{X^m}(x^m)\eta_m(x^m)} \Bigr\} \, dz^m dy
\\ & = \frac{1}{h^{d_m}} \int_{\mathbb{R}^d}  K_m\bigl(\frac{z^m-x^m}{h}\bigr) \Bigl\{ \eta(x^{m^c}, z^m)f_{X^{m^c}, X^m} (x^{m^c}, z^m)  - \frac{f_{X}(x) \eta(x) \eta_m(z^m) f_{X^m}(z^m) }{f_{X^m}(x^m)\eta_m(x^m)} \Bigr\} \, dz^m
\\ & = \int_{\mathbb{R}^d}  K_m(u ) \Bigl\{ \eta(x^{m^c}, x^m + hu)f_{X^{m^c}, X^m} (x^{m^c}, x^m+hu)  
\\ &  \hspace{ 150pt} - \frac{f_{X}(x) \eta(x) \eta_m(x^m + hu) f_{X^m}(x^m + hu) }{f_{X^m}(x^m)\eta_m(x^m)} \Bigr\} \, du
\\ &  = \int_{\mathbb{R}^d}  K_m(u ) \Bigl\{ \eta(x^{m^c}, x^m + hu)f_{X^{m^c}, X^m} (x^{m^c}, x^m+hu)  - f_{X}(x) \eta(x) \Bigr\} \, du
\\ & \hspace{ 30 pt} +  \mu_{0,m^c} \int_{\mathbb{R}^d}  K_m(u ) \Bigl\{ f_{X}(x) \eta(x)  - \frac{f_{X}(x) \eta(x) \eta_m(x^m + hu) f_{X^m}(x^m + hu) }{f_{X^m}(x^m)\eta_m(x^m)} \Bigr\} \, du
\\ & \leq Lh \Bigl[\mu_{1,m} \bigl\{\eta(x) + f_X(x) + \frac{f_{X^m}(x^m) + \eta_m(x^m) }{f_{X^m}(x^m) \eta_m(x^m)}\bigr\} + L h \mu_{1,m} \bar{K}_m \bigl\{1 +  \frac{1 }{f_{X^m}(x^m) \eta_m(x^m)} \bigr\} \Bigr]. 
\end{align*}
Furthermore, we show below that
\begin{equation}
\label{eq:varbound1}
\var\Big\{\frac{\mu_{0,m^c}}{h^{d_m} }Y_i  K_m\bigl(\frac{X^m_i-x^m}{h}\bigr) \Bigl(f_{X^{m^c}| X^m,  Y} (x^{m^c}; X_i^m, Y_i)  -  \frac{f_{X}(x) \eta(x)}{f_{X^m}(x^m)\eta_m(x^m)} \Bigr)\Bigr\} = O\Bigl(\frac{1}{h^{d_m}}\Bigr),
 \end{equation}
 uniformly for $h \in [n^{-\beta}, n^{-\alpha}]$.  It follows that $|\Pi_1| = o_p(1)$, uniformly for $h \in [n^{-\beta}, n^{-\alpha}]$, by Chebychev's inequality -- this proves \eqref{eq:claimII1}.

Next, to see \eqref{eq:claimII2}, write
\begin{align}
\label{eq:term1}
\Pi_2 &= \frac{1}{n_0} \sum_{i \in A_0}  \Bigl[ \frac{1}{h^{d}} \mathbb{E}\bigl\{K\bigl(\frac{X_i-x}{h}\bigr) \big| \mathcal{T}_{A_0,m}\bigr\} - \frac{\mu_{0,m^c} f_{X}(x)}{f_{X^m}(x^m) h^{d_m}} K_m\bigl(\frac{X^m_i-x^m}{h}\bigr) \Bigr] \nonumber
\\ & = \frac{1}{n_0} \sum_{i \in A_0}K_m\bigl(\frac{X_i^m - x^m}{h}\bigr) \Bigl[ \frac{1}{h^d}  \mathbb{E}\Bigl\{K_{m^c}\bigl(\frac{X_i^{m^c}-x^{m^c}}{h}\bigr) \Big| X_i^m, Y_i\Bigr\} - \frac{\mu_{0,m^c} f_{X}(x)}{f_{X^m}(x^m)h^{d_m}} \Bigr] \nonumber
\\ & = \frac{ \mu_{0,m^c}}{n_0 h^{d_m}} \sum_{i \in A_0} K_m\bigl(\frac{X^m_i-x^m}{h}\bigr) \Bigl[f_{X^{m^c}| X^m,  Y} (x^{m^c}; X_i^m, Y_i)  - \frac{f_{X}(x)}{f_{X^m}(x^m)} \Bigr]+ R_5,
\end{align} 
where, using the same technique used to bound $R_4$, we have
\begin{align*} 
\mathbb{E}|R_{5}| & \leq  L h \mu_{1, m^c} \mathbb{E} \Bigl\{ \frac{1}{n_0 h^{d_m}} \sum_{i \in A_0} \frac{K_m\bigl(\frac{X_i^m - x^m}{h}\bigr)}{ f_{X^m,Y}(X_i^m, Y_i)} \Bigr\} 
\\ & = L h \mu_{1, m^c} \frac{1}{h^{d_m}} \int_{\mathbb{R}^{d_m} \times \mathcal{D}_Y} K_m\bigl(\frac{z^m - x^m}{h}\bigr) \, dz dy  \leq 2 L h \mu_{1, m^c} \mu_{0,m}\bar{Y}. 
\end{align*} 
Thus by Markov's inequality $R_5=O_p(h)$, uniformly for $h\in [n^{\beta},n^{-\alpha}]$. 
Finally, the remaining term in \eqref{eq:term1} is the average of $n_0$ independent and identically distributed terms, each with expectation given by
\begin{align*} 
&\frac{1}{h^{d_m}} \int_{\mathbb{R}^d \times \mathbb{R}} K_m\bigl(\frac{z^m-x^m}{h}\bigr)\Bigl\{ f_{X^{m^c}| X^m,  Y} (x^{m^c}; z^m, y)  -  \frac{f_{X}(x)}{f_{X^m}(x^m)} \Bigr\} f_{X^m, Y}(z^m, y) \, dz^m dy
\\& = \frac{\mu_{0,m^c}}{h^{d_m}} \int_{\mathbb{R}^d \times \mathbb{R}} K_m\bigl(\frac{z^m-x^m}{h}\bigr) \Bigl\{ f_{X^{m^c}, X^m,  Y} (x^{m^c}, z^m, y)  - \frac{f_{X}(x) f_{X^m, Y}(z^m, y) }{f_{X^m}(x^m)} \Bigr\} \, dz^m dy
\\ & = \frac{\mu_{0,m^c}}{h^{d_m}} \int_{\mathbb{R}^d}  K_m\bigl(\frac{z^m-x^m}{h}\bigr) \Bigl\{ f_{X^{m^c}, X^m} (x^{m^c}, z^m)  - \frac{f_{X}(x) f_{X^m}(z^m) }{f_{X^m}(x^m)} \Bigr\} \, dz^m
\\ & = \mu_{0,m^c} \int_{\mathbb{R}^d}  K_m(u ) \Bigl\{f_{X^{m^c}, X^m} (x^{m^c}, x^m+hu)  - \frac{f_{X}(x)f_{X^m}(x^m + hu) }{f_{X^m}(x^m)} \Bigr\} \, du
\\ & = \mu_{0,m^c} \int_{\mathbb{R}^d}  K_m(u ) \Bigl\{ f_{X^{m^c}, X^m} (x^{m^c}, x^m+hu)  - \frac{f_{X}(x)  f_{X^m}(x^m + hu) }{f_{X^m}(x^m)} \Bigr\} \, du
\\ & \leq Lh \mu_{0,m^c}\mu_{1,m} \Bigl\{1 + \frac{f_{X}(x)}{f_{X^m}(x^m)}\Bigr\}. 
\end{align*} 
Moreover, we will show below that
\begin{equation}
\label{eq:varbound2}
\var\Big\{\frac{\mu_{0,m^c}}{h^{d_m} } K_m\bigl(\frac{X^m_i-x^m}{h}\bigr) \Bigl(f_{X^{m^c}| X^m,  Y} (x^{m^c}; X_i^m, Y_i)  -  \frac{f_{X}(x) \eta(x)}{f_{X^m}(x^m)\eta_m(x^m)} \Bigr)\Bigr\} = O\Bigl(\frac{1}{h^{d_m}}\Bigr),
 \end{equation}
 uniformly for $h \in [n^{-\beta}, n^{-\alpha}]$. It follows that  $|\Pi_2| = o_p(1)$, uniformly for $h \in [n^{-\beta}, n^{-\alpha}]$.  
 
 It remains to show \eqref{eq:varbound1} and  \eqref{eq:varbound2}.   We have that 
\begin{align*} 
&\var\Big\{Y_i  K_m\bigl(\frac{X^m_i-x^m}{h}\bigr) \Bigl(f_{X^{m^c}| X^m,  Y} (x^{m^c}; X_i^m, Y_i)  -  \frac{f_{X}(x) \eta(x)}{f_{X^m}(x^m)\eta_m(x^m)} \Bigr)\Bigr\}
\\ & \leq \mathbb{E} \Bigl[\Big\{ Y_i  K_m\bigl(\frac{X^m_i-x^m}{h}\bigr) \Bigl(f_{X^{m^c}| X^m,  Y} (x^{m^c}; X_i^m, Y_i)  -  \frac{f_{X}(x) \eta(x)}{f_{X^m}(x^m)\eta_m(x^m)} \Bigr)\Bigr\}^2\Bigr]
\\ & \leq  \sup_{z^m \in \mathbb{R}^{d_m}, y \in \mathcal{D}_Y }  \Bigl[y^2 \Bigl\{f_{X^{m^c}|X^m,  Y} (x^{m^c}; z^m, y)  
\\ & \hspace {150 pt} - \frac{f_{X}(x) \eta(x)  }{f_{X^m}(x^m)\eta_m(x^m)} \Bigr\}^2  f_{X^m}(z^m)\Bigr]   \int_{\mathbb{R}^{d_m}} K_m^2\bigl(\frac{z^m-x^m}{h}\bigr)\, dz^m 
\\ & =  \sup_{z^m \in \mathbb{R}^{d_m}, y \in \mathcal{D}_Y } \Bigl[ y^2 \Bigl\{ f_{X^{m^c}|X^m,  Y} (x^{m^c}; z^m, y)  - \frac{f_{X}(x) \eta(x)  }{f_{X^m}(x^m)\eta_m(x^m)} \Bigr\}^2  f_{X^m}(z^m)\Bigr] \nu_m h^{d_m} 
\end{align*} 
Thus, \eqref{eq:varbound1} holds since the density functions are bounded.   The claim in \eqref{eq:varbound2} can be seen by the same argument. 

The proof is competed by combining Part I, Part II and Slutsky's Theorem. 
\end{proof} 

\end{document}